\documentclass[unnumsec,webpdf,modern,medium]{oup-authoring-template}

\onecolumn
\usepackage{fancyhdr,amsmath,amssymb, amsthm,mathtools}
\usepackage{bm}
\usepackage{dsfont}
\usepackage{graphicx}
\usepackage{algorithm}
\usepackage[export]{adjustbox}
\usepackage[british]{babel}
\usepackage{calc} 
\setcitestyle{round}
\usepackage{adjustbox}
\usepackage{hyphenat}
\usepackage{bbm}

\newlength{\leftcolumnwidth}
\newcounter{globalrem}
\newenvironment{jrssremark}[1]{
    \refstepcounter{globalrem}
  \par\medskip\noindent
  \begin{minipage}[t]{\leftcolumnwidth}
      \raggedleft \textbf{#1}
  \end{minipage}
  \hspace{1em}
  \begin{minipage}[t]{\linewidth - \leftcolumnwidth - 1em}
      \setlength{\parindent}{1em}
      \ignorespaces
}{
  \end{minipage}\par\medskip
}

\newcounter{globaldef}
\newenvironment{jrssdefinition}[1]{
    \refstepcounter{globaldef}
  \par\medskip\noindent
   \begin{minipage}[t]{\leftcolumnwidth}
      \raggedleft \textbf{#1}
  \end{minipage}
  \hspace{1em}
  \begin{minipage}[t]{\linewidth - \leftcolumnwidth - 1em}
      \setlength{\parindent}{1em}
      \ignorespaces
}{
  \end{minipage}\par\medskip
}

\newcounter{globalprop}
\newenvironment{jrssproposition}[1]{
    \refstepcounter{globalprop}
  \par\medskip\noindent
  \begin{minipage}[t]{\leftcolumnwidth}
     \raggedleft \textbf{#1}
  \end{minipage}
  \hspace{1em}
  \begin{minipage}[t]{\linewidth - \leftcolumnwidth - 1em}
      \setlength{\parindent}{1em}
      \ignorespaces
}{
  \end{minipage}\par\medskip
}

\newcounter{globalem}
\newenvironment{jrsslemma}[1]{
    \refstepcounter{globalem}
  \par\medskip\noindent
  \begin{minipage}[t]{\leftcolumnwidth}
     \raggedleft \textbf{#1}
  \end{minipage}
  \hspace{1em}
  \begin{minipage}[t]{\linewidth - \leftcolumnwidth - 1em}
      \setlength{\parindent}{1em}
      \ignorespaces
}{
  \end{minipage}\par\medskip
}

\newcounter{globatheo}
\newenvironment{jrsstheorem}[1]{
    \refstepcounter{globatheo}
  \par\medskip\noindent
  \begin{minipage}[t]{\leftcolumnwidth}
     \raggedleft \textbf{#1}
  \end{minipage}
  \hspace{1em}
  \begin{minipage}[t]{\linewidth - \leftcolumnwidth - 1em}
      \setlength{\parindent}{1em}
      \ignorespaces
}{
  \end{minipage}\par\medskip
}

\newcounter{globalcor}
\newenvironment{jrsscorollary}[1]{
    \refstepcounter{globalcor}
  \par\medskip\noindent
  \begin{minipage}[t]{\leftcolumnwidth}
     \raggedleft \textbf{#1}
  \end{minipage}
  \hspace{1em}
  \begin{minipage}[t]{\linewidth - \leftcolumnwidth - 1em}
      \setlength{\parindent}{1em}
      \ignorespaces
}{
  \end{minipage}\par\medskip
}

\newcounter{globalas}
\newenvironment{jrssassumption}[1]{
    \refstepcounter{globalas}
  \par\medskip\noindent
  \begin{minipage}[t]{\leftcolumnwidth}
     \raggedleft \textbf{#1}
  \end{minipage}
  \hspace{1em}
  \begin{minipage}[t]{\linewidth - \leftcolumnwidth - 1em}
      \setlength{\parindent}{1em}
      \ignorespaces
}{
  \end{minipage}\par\medskip
}

\newcommand{\bY}{ {\bf Y} }
\newcommand{\by}{ {\bf y} }

\newcommand{\bX}{ {\bf X} }
\newcommand{\bx}{ {\bf x} }

\newcommand{\bJ}{ {\bf J} }
\newcommand{\bh}{ {\bf h} }
\newcommand{\bOmega}{ {\boldsymbol \Omega} }
\newcommand{\E}{\mathbb{E}}
\newcommand{\btheta}{ {\boldsymbol \theta} }
\newcommand{\bbeta}{ {\boldsymbol \beta} }
\newcommand{\bbbeta}{ {\boldsymbol \eta} }
\newcommand{\balpha}{ {\boldsymbol \alpha} }
\newcommand{\bgamma}{ {\boldsymbol \gamma} }
\newcommand{\bSigma}{ {\boldsymbol \Sigma} }

\begin{document}
\journaltitle{--}
\copyrightyear{--}
\pubyear{--}
\appnotes{$_.$}

\firstpage{1}
\setcounter{secnumdepth}{2}

\fontsize{9.9}{10.55}\selectfont

\title[Skew-symmetric approximations of posterior distributions]{Skew-symmetric approximations of posterior distributions}

\author[1]{Francesco Pozza}
\author[1,2]{Daniele Durante}
\author[1,2]{Botond Szabo}

\authormark{Pozza, Durante and Szabo}

\address[1]{\orgdiv{Bocconi Institute for Data Science and Analytics},  \orgname{Bocconi University},  \country{Via Roentgen 1, Milan, Italy}}
\address[2]{\orgdiv{Department of Decision Sciences},  \orgname{Bocconi University},  \country{Via Roentgen 1, Milan, Italy}}

\corresp[]{Address for correspondence: Francesco Pozza, Bocconi Institute for Data Science and Analytics, Bocconi University, Via Roentgen 1, Milan, Italy. Email: \href{email: francesco.pozza2@unibocconi.it}{francesco.pozza2@unibocconi.it}}

\abstract{
Popular deterministic approximations of posterior distributions from, e.g.\ the  Laplace method, variational Bayes and expectation-propagation, generally rely on symmetric approximating families, often taken to be Gaussian. This choice facilitates optimization and inference, but typically affects the quality of the overall approximation. In fact, even in basic parametric models, the posterior distribution often displays asymmetries that yield bias and a reduced accuracy when considering symmetric approximations. Recent research has moved towards more flexible approximating families which incorporate skewness. However, current solutions are often model specific, lack a general supporting theory, increase the computational complexity of the optimization problem, and do not provide a broadly applicable solution to incorporate skewness in any symmetric approximation. This article addresses such a gap by introducing a general and provably optimal strategy to perturb any off-the-shelf symmetric approximation of a generic posterior distribution. This novel perturbation scheme is derived without additional optimization steps, and yields a similarly tractable approximation within the class of skew-symmetric densities that provably enhances the finite sample accuracy of the original symmetric counterpart. Furthermore, under suitable assumptions, it improves the convergence rate to the exact posterior by at least a $\sqrt{n}$ factor, in asymptotic regimes. These advancements are illustrated in numerical studies focusing on skewed perturbations of state-of-the-art Gaussian approximations. 
}


\maketitle

%###############################################################################
%###############################################################################
\section{\large 1. Introduction} \label{sec_1}
Deterministic approximations of intractable posterior distributions provide a popular alternative to general Markov chain Monte Carlo methods within Bayesian inference  \citep[e.g.][]{blei2017variational, minka2013expectation, inla_paper,walker1969asymptotic}. Noteworthy examples in such a class are the approximations obtained under the classical Laplace method and its generalizations  \citep[][]{rossell2021approximate, inla_paper, tierney1986accurate,walker1969asymptotic}, variational Bayes (VB) \citep[][]{blei2017variational,opper2009variational}, including its widely applicable stochastic black-box extensions \citep[][]{kucukelbir2017automatic,ranganath2014black}, and expectation-propagation  (EP) \citep[][]{minka2013expectation,vehtari2020expectation}; see also \citet[][Chapters 4.4 and 10]{bishop2006pattern} and \citet{Chopin_2017} for a general overview. Albeit derived under different arguments and optimization strategies, all these solutions share a common trade-off between the attempt to facilitate optimization and posterior inference via a sufficiently tractable approximation, and the need to avoid an overly simplified characterization of the intractable posterior distribution, which may undermine inference accuracy. The overarching focus to date has been on addressing the first objective at the expense of the second. This has been accomplished by searching for tractable approximating densities, often within the Gaussian family. Under  the classical Laplace method this choice is implicitly encoded in the second-order Taylor expansion of the log-posterior density \citep[e.g.][Chapter 4.4]{bishop2006pattern}, while for EP \citep[e.g.][]{minka2013expectation} and VB \citep[e.g.][]{challis2013gaussian,opper2009variational} it is enforced as a constraint within a suitable optimization problem seeking the closest approximating density to the target posterior under a selected divergence. Even when replacing such a constraint in VB with mean-field factorization assumptions \citep[e.g.][]{blei2017variational}, Gaussian approximations can still arise either as a solution of the optimization problem under routinely implemented Bayesian models, such as logit  \citep{durante2019conditionally} and probit   \citep{consonni2007mean} regression, or as a  building block to ease implementation in more general settings \citep[e.g.][]{carbonetto2012scalable,kucukelbir2017automatic,ray2022variational, rezende2015variational, wang2013variational}. 

The above approximations yield state-of-the-art solutions which are implemented in routinely employed \texttt{R} libraries (e.g.\ \texttt{rstan}). In addition, these solutions also find theoretical justification in the asymptotic normality of the posterior distribution ensured by Bernstein--von Mises type results  \citep[e.g.][]{van2000asymptotic} that have been extensively studied in the context of classical Laplace  \citep[e.g.][]{kasprzak2022good, panov2015finite,spokoiny2023inexact,spokoiny2025accuracy,walker1969asymptotic}, VB  \citep[e.g.][]{katsevich2023approximation,wang2019frequentist}  and EP  \citep[e.g.][]{dehaene2018expectation} approximations. While such results yield theoretical support within asymptotic regimes, in practice the accuracy of Gaussian approximations ultimately depends on the actual shape of the posterior distribution in  finite samples, which often displays non-negligible skewness that might affect inference accuracy when relying on Gaussian approximations. In fact, although the design of tractable and computationally efficient skewed approximations of posterior distributions has attracted less interest to date,  the available contributions along these lines have proven effective in refining the accuracy of the Gaussian counterparts \citep[see, e.g.][]{anceschi2022bayesian,challis2012affine,durante2023skewed,fasanoscalable,fasano2022class,katsevich2023tight,ormerod2011skew,inla_paper,ruli,salomone2023structured,smith2020high,tan2023variational,zhou2023skew}.  While such gains should stimulate an increasing adoption of skewed approximations, there is still a gap between the promising methodological avenues opened by this perspective and its routine use in practice. In fact, when compared to classical approximations based on Gaussians, or on other symmetric densities, there are at least three major barriers, which still undermine the routine adoption of these skewed extensions. 

\begin{itemize}
\item  {\bf (a) Methods:} The current solutions are developed under specific models and, even when targeting more general posterior distributions, the focus remains on improving a single  strategy, often VB. Moreover, the overarching interest is on including skewness in Gaussian approximating  densities via extensions of skew-normal families \citep[e.g.][]{arellano2006unification,gonzalez2004additive}, rather than  in more general symmetric ones.
\item  {\bf (b) Computation:}  The  schemes to derive the above  approximations require further optimization steps that often yield substantial increments in computational costs. In some cases these schemes necessitate even the calculation of the exact moments for the intractable posterior, a requirement  in contrast with the motivations for considering deterministic approximations.
\item  {\bf (c) Theory:} Unlike for the classical versions of the  Laplace method, VB  and EP, most skewed extensions lack general supporting theory on the accuracy in approximating the whole posterior distribution and its functionals, both in finite samples and also in asymptotic regimes.
\end{itemize}

Two exceptions to point (c) can be found in the  recent contributions by \citet{fasanoscalable} and \citet{durante2023skewed}. The former demonstrates that, in Bayesian probit regression, a partially-factorized generalization of  mean-field VB leads to a unified skew-normal (SUN)  \citep{arellano2006unification} approximation, which  provably matches the exact posterior for fixed sample size $n$ and dimension $d \to \infty$. The latter derives, instead, a refined version of the Bernstein--von Mises theorem \citep[see, e.g.][]{van2000asymptotic} which replaces the classical Gaussian limiting law with a valid skewed perturbation (arising from a third-order version of the Laplace method) to obtain an improvement in the convergence rates to the target posterior by a $\sqrt{n}$ factor, under the total variation (TV) distance; see also \citet{katsevich2023tight} for alternative higher-order skewed expansions, which, however, do not guarantee valid  approximating densities. Albeit being derived in specific models  \citep{fasanoscalable}, or with a focus to a single approximation strategy  \citep{durante2023skewed}, these results suggest that the inclusion of skewness in popular symmetric approximations has the concrete potential to substantially improve inference accuracy. However, it is still not clear if a broadly applicable, computationally tractable and theoretically supported solution addressing points (a)--(c) above actually exists and can be effectively derived.

In this article we provide a positive and innovative answer to the above question by deriving a broadly applicable  strategy to perturb, at no additional optimization costs, any off-the-shelf symmetric approximation of a generic posterior distribution. Crucially, such a perturbation yields a provably more accurate skewed  approximation of the target posterior, which belongs to the tractable class of skew-symmetric distributions reviewed in Section~\ref{sec_2.a} \citep[][]{azzalini2003distributions,genton2005generalized, ma2004flexible,wang2004skew}. Our proposed solution follows from two fundamental results, which we derive in Section~\ref{sec_32}. The first (see Proposition~\ref{prop1}), shows that the density of any generic posterior distribution admits a skew-symmetric representation decomposing this density as the product of (i) a suitably symmetrized version of the target posterior and (ii)  a tractable skewing factor,  which is available in closed form and does not depend on additional unknown parameters. This novel characterization, whose scope extends beyond our contribution, relates to an important existence  result of  skew-symmetric densities \citep[][]{wang2004skew} which has been never explored within the context of Bayesian inference and approximations, despite its unique potential. In fact, our first result in Proposition~\ref{prop1} opens the avenues to prove a second important property (see Lemma~\ref{lemma_1}), which ensures that any symmetric approximation of a posterior distribution, including the popular ones reviewed in Section~\ref{sec_2}, is more accurate in approximating the symmetrized version of such a posterior in its skew-symmetric representation, rather than the original posterior itself. This holds under the TV distance, Kullback--Leibler (KL) divergence  \citep[][]{kullback1951information}, reverse-KL and generic $\alpha$-divergences.

The aforementioned results yield, therefore, a natural strategy to derive skewed approximations of a generic posterior distribution via direct perturbation of any symmetric approximation for such a posterior. In particular, these can be obtained, at no additional optimization costs, by simply replacing the symmetrized version of the posterior density in its skew-symmetric representation with the available symmetric approximation to be improved. As proved in Section~\ref{sec_33}, this yields a more accurate asymmetric alternative that matches exactly the closed-form skewing factor of the target posterior, while  preserving the tractability of its unperturbed symmetric counterpart. This tractability follows from the fact that the density of the resulting skewed approximation falls in the skew-symmetric family \citep[see, e.g.][]{azzalini2003distributions,ma2004flexible,wang2004skew}. Such a family admits a closed-form normalizing constant along with straightforward i.i.d.\ sampling schemes, which enable inference via Monte Carlo evaluation of any functional of the approximate posterior. These schemes only require simulation from the original symmetric approximation and  evaluation of the skewing factor. As shown in Section~\ref{sec_32}, this factor crucially depends only on the ratios of posterior densities, thereby allowing cancellation of the intractable normalizing constant, and hence, direct computation without additional optimizations. 

Our theoretical results within Section~\ref{sec_33} show that the strategy we propose is not only tractable, but also yields skew-symmetric approximations that are provably more accurate than the original symmetric counterpart, under any sample size $n$ and several popular divergences (i.e.\ TV distance and any $\alpha$-divergence, including its limiting KL and reverse-KL forms \citep[see, e.g.][]{cichocki2010families,margossian2024ordering}). Even more, the derived skewing factor  is shown to be the optimal among those yielding approximations within the  skew-symmetric class. This allows to formally interpret the proposed solution as the optimum of a minimization problem that seeks the closest approximation to the target posterior in the entire class of skew-symmetric densities generated by perturbations of the known symmetric approximation to be improved. The accuracy improvements formalized in Section~\ref{sec_33} are quantified asymptotically in Section~\ref{sec_34} with a focus on  perturbations of specific symmetric approximations which can be shown to improve by at least a  $\sqrt{n}$ factor the convergence rate to the target posterior relative to the original unperturbed counterpart, under all the divergences considered within Section~\ref{sec_33}. As illustrated in the simulation studies in Section~\ref{sec_4}, the theoretical results presented in Sections~\ref{sec_33} and \ref{sec_34} find empirical evidence also in practice, including in challenging high-dimensional contexts beyond the  regimes considered by the asymptotic theory within Section~\ref{sec_34}. In Section~\ref{sec_4_app}, these accuracy gains are further confirmed in real-data applications under realistic Bayesian hierarchical formulations showcasing both the relative and  absolute improvements of the proposed skew-symmetric perturbation of several state-of-the-art Gaussian approximations  from the Laplace method, black-box VB, and EP. In these real-data applications we also illustrate the remarkable gains in the effective sample size (ESS) when replacing such  Gaussian approximations with the corresponding skew-symmetric counterparts  as proposals in importance sampling \citep[e.g.][Chapter 8]{chopin2020introduction}. These improvements clarify the broad impact and applicability of the ideas presented in this article. Section~\ref{sec_5}  concludes by further stressing this point through the discussion of promising research directions motivated by the general scope of the skew-symmetric representation of posterior densities behind our contribution.  Proofs can be found in the Supplementary Material. Code is available at \url{https://github.com/Francesco16p/SkewAppr_Post}.

%###############################################################################
%###############################################################################
\section{\large 2. Skew-symmetric approximations of posterior distributions} \label{sec_3}
Let $L(\btheta; \by_{1:n}) = \prod_{i = 1}^{n} p(\by_{i}\mid \btheta )$ denote the likelihood that is induced by a generic statistical model characterizing the probabilistic generative mechanism of the data $\by_{1:n}=(\by_1, \ldots, \by_n)$ through a family of distributions indexed by a parameter \smash{$\btheta \in \Theta  \subseteq \mathbb{R}^{d}$}. To perform inference on $\btheta$, Bayesian statistics defines a prior $\pi(\btheta)$ for $\btheta$, which, combined with $L(\btheta; \by_{1:n})$, yields the corresponding posterior density $\pi_{n}(\btheta):=\pi(\btheta \mid \by_{1:n})$ via Bayes rule, i.e. $\pi_{n}(\btheta)=\pi(\btheta)L(\btheta; \by_{1:n})/c(\by_{1:n})$, with $c(\by_{1:n})=\int_ \Theta\pi(\btheta)L(\btheta; \by_{1:n}) \mbox{d}\btheta$. Inference under such a posterior is then performed through evaluation of functionals of interest. The feasibility of this task inherently depends on the tractability of $\pi_{n}(\btheta)$. In fact, $\pi_{n}(\btheta)$ is even not available analytically in most Bayesian models, due, in general, to the lack of a closed-form expression for $c(\by_{1:n})$. To  address such an issue,  standard practice relies either on Markov chain Monte Carlo schemes or on tractable deterministic approximations of $\pi_{n}(\btheta)$. Here we consider the latter class of solutions and focus  on improving those state-of-the-art symmetric approximations that routinely arise in practice. 

More specifically, let $\bar{q}_{n,{\btheta}^*}(\btheta)$ be an already-derived symmetric approximation of the posterior density $\pi_{n}(\btheta)$, with ${\btheta}^*$ denoting the known symmetry point, i.e. $\bar{q}_{n,{\btheta}^*}(\btheta)=\bar{q}_{n,{\btheta}^*}(2{\btheta}^*-\btheta)$ for all $\btheta \in \Theta $. Our main objective is to obtain a more accurate approximation {$q_{n,{\btheta}^*}(\btheta)$} from the perturbation of {$\bar{q}_{n,{\btheta}^*}(\btheta)$} via a skewing factor $w_{n,{\btheta}^*}(\btheta)$, which does not require optimization of additional parameters, while  preserving the inference tractability of the original symmetric approximation. Although the novel strategy we derive within Section~\ref{sec_32} is guaranteed to improve the accuracy of any symmetric approximating density  {$\bar{q}_{n,{\btheta}^*}(\btheta)$}, in practice it is natural to focus on perturbing the outputs of routinely implemented Gaussian approximations from, e.g.\ Laplace, VB and EP; see Section~\ref{sec_2} for an overview of these methods and refer to \citet[][Chapters 4.4 and 10]{bishop2006pattern}, \citet{challis2013gaussian}, \citet{kucukelbir2017automatic}, \citet{opper2009variational} and \citet{vehtari2020expectation} for a detailed treatment. As clarified in these contributions, such symmetric approximations provide popular solutions which are effectively coded in standard platforms (e.g. \textsc{stan}) and softwares (e.g. \texttt{R}, \texttt{python} and \texttt{julia}). Since the proposed skew-symmetric approximation arises as the direct perturbation of an already-available symmetric one, without further optimization, it is useful to place particular emphasis on those Gaussian approximations provided by state-of-the-art softwares so as to stimulate direct adoption of the proposed perturbation. Note also that standard symmetric approximations often depend on additional parameters beyond the symmetry point ${\btheta}^*$. For instance, Gaussian approximations also require estimation of the covariance matrix. As clarified in  Section~\ref{sec_32}, the skewing factor we derive only depends on ${\btheta}^*$ and is applied to an already-available symmetric approximation whose parameters have been previously estimated under standard algorithms. Hence, to ease notation, we index $\bar{q}_{n,{\btheta}^*}(\btheta)$ only by ${\btheta}^*$, and leave the remaining, already-estimated, parameters implicit. Theoretical guarantees for the finite sample and asymptotic accuracy of the proposed approximation are derived in Section~\ref{sec_theo}.

%###############################################################################
\subsection{\large 2.1. An overview   of symmetric approximations of posterior distributions} \label{sec_2}
When seeking a tractable approximation of a generic posterior distribution for the parameter $\btheta$, a simple option is to consider the classical Gaussian approximation that  arises from the Laplace method \citep[see, e.g.][Chapter 4.4]{bishop2006pattern}. Such a solution follows directly from a second-order Taylor expansion of the un-normalized log-posterior $\log[\pi(\btheta)L(\btheta; \by_{1:n})]$ at the maximum a posteriori (MAP). This yields a Gaussian approximating density, which is centered at the MAP and has covariance matrix given by the inverse of the negative Hessian for $\log[\pi(\btheta)L(\btheta; \by_{1:n})]$, again evaluated at the MAP  \citep[e.g.][Chapter 13]{gelman1995bayesian}. The simplicity of this solution  has stimulated broad applicability and several extensions, including, among others, integrated nested Laplace approximation (INLA)   \citep{inla_paper} and approximate Laplace approximation (ALA)   \citep{rossell2021approximate}. INLA combines efficient numerical integration and analytical approximations to derive accurate characterizations of posterior marginals for parameters of interest, under latent Gaussian models. The scheme is inspired by \citet{tierney1986accurate} Laplace approximation of the marginal posterior density for a subset of parameters of interest, which expresses this marginal as proportional to the ratio between the joint posterior and the full conditional density of the remaining parameters, and then applies the Gaussian approximation from the Laplace method to the latter density. This strategy can yield asymmetric approximations of marginal posterior densities, but still relies on nested Gaussian approximations. As such, the skew-symmetric solution derived in Section~\ref{sec_32} has the potential to further improve also the accuracy of INLA. Our proposal can be also applied directly to improve ALA. This scheme does not differ from the classical Laplace method in the shape of the approximating density. Rather, it  provides a scalable strategy which avoids MAP estimation.

While Laplace type schemes yield simple approximation strategies, the resulting solution may fail to incorporate global characteristics of the posterior beyond the local behavior at the MAP \citep[][Chapter 4.4]{bishop2006pattern}. This issue has motivated interest in alternative schemes, with a main focus on VB \citep[e.g.][]{blei2017variational} and EP \citep[e.g.][]{vehtari2020expectation}. VB specifies a tractable family $\mathcal{F}$ of approximating densities and then identifies, within this family, the one that is closest to the intractable posterior  in KL divergence  \citep[][]{kullback1951information}. The common practice in specifying  $\mathcal{F}$ relies either on parametric, often Gaussian, families \citep[see, e.g.][]{challis2013gaussian,kucukelbir2017automatic,opper2009variational,tan2018gaussian}, or on mean-field assumptions \citep[e.g.][]{blei2017variational} enforcing the joint approximating density for $\btheta$ to factorize as the product of marginals for suitably selected non-overlapping subsets of parameters. In the first case the final output of the optimization problem is, by definition, a symmetric density which can be readily improved under our proposed skew-symmetric approximation. Conversely, the second can yield skewed solutions matching the shape of the actual, not necessarily symmetric, full-conditional density for the subset of parameters in $\btheta$ comprising the mean-field factors \citep[e.g.][]{blei2017variational}. Nonetheless, several routine use Bayesian models, such as logit  \citep{durante2019conditionally} and probit   \citep{consonni2007mean} regression, admit Gaussian full-conditionals for the coefficients. Therefore, the proposed perturbation might be useful even within mean-field VB to improve the accuracy of the symmetric density factors, and hence, of the overall approximation for the entire posterior. Finally, it is important to emphasize that these Gaussian approximations appear either as the final output or as a building block also within Laplace variational inference, delta-method variational inference  \citep{wang2013variational} and stochastic black-box extensions of VB \citep{ranganath2014black}, including automatic differentiation variational inference  (ADVI)  \citep{kucukelbir2017automatic} and solutions relying on normalizing flows \citep[e.g.][]{papamakarios2021normalizing,rezende2015variational}. Such methods are crucial to facilitate the derivation of generally applicable  and more flexible VB approximations that can be automatically implemented even in settings lacking closed-form expressions for the quantities involved in the optimization problem. Therefore, the proposed skew-symmetric perturbation can be either applied directly to improve the approximation accuracy of the Gaussian output provided by these methods (see Sections~\ref{sec_4}--\ref{sec_4_app}), or may be even  included within the stochastic optimization scheme to enlarge the family (and hence the quality) of the approximating densities, without arguably affecting the tractability of the  routine. This broad applicability of the simple skew-symmetric perturbation we propose is a unique feature, which generally lacks in more sophisticated solutions yielding accurate, possibly asymmetric, approximations. Successful ones are those provided by the VB strategies leveraging normalizing flows \citep[e.g.][]{papamakarios2021normalizing,rezende2015variational}. Interestingly, as discussed in Section~\ref{sec_5}, such strategies often use suitable Gaussians as initial densities, and hence, our proposal could be possibly leveraged to provide improved (i.e. closer to the target posterior), yet  tractable,  skew-symmetric perturbations of these initial Gaussian densities.
 
Although VB is arguably the most widely studied and implemented deterministic approximation strategy, commonly used variational approximations often suffer from an underestimation of posterior uncertainty \citep[see, e.g.][]{blei2017variational,giordano2018covariances}. This is implicit in the expression of the KL minimized under VB that penalizes densities placing mass to areas of low posterior probability, but it does not enforce a similarly strong penalty in the opposite direction. As a consequence,  minimizing such a  KL  yields more concentrated approximations avoiding those regions where the target posterior does not place substantial mass. While improved estimates of variances and covariances have been proposed in the context of VB  \citep{giordano2018covariances}, another natural solution is to consider EP  \citep[e.g.][]{vehtari2020expectation}, which addresses such an issue by minimizing a reverse form of the KL considered under VB. This implies penalizations in the opposite direction than those of  VB, and hence, a tendency to favor more global {\em mass-covering} approximations that match more closely the variability encoded in the target posterior. To obtain these approximations, EP postulates that the target posterior density itself can be expressed as a product of factors, and then iteratively approximates each of these factors with an element of a tractable parametric family, almost always Gaussian. This results in a computational scheme updating each factor at-a-time via moment matching between the global approximating density and a hybrid one, more tractable  than the target posterior, where the other factors are kept fixed at the most recent approximation \citep[e.g.][]{vehtari2020expectation}. Being often Gaussian, these EP  approximating densities can be readily perturbed under our proposed strategy to further improve the quality of the approximation. In fact, although several empirical studies have highlighted the remarkable accuracy of Gaussian EP \citep[][]{anceschi2022bayesian, Chopin_2017,vehtari2020expectation} relative to Laplace  and VB, there is still a lack of tractable solutions which are capable of including skewness within these Gaussian approximations to further improve the quality of EP. As illustrated in the empirical studies in Sections~\ref{sec_4} and  \ref{sec_4_app}, accounting for these skewed behaviors yields further  improvements over such an already-successful strategy.

It is important to emphasize that our strategy applies directly to any symmetric approximation, not necessarily Gaussian. This can be useful to perturb further extensions of the aforementioned methods aimed at capturing higher-order properties (e.g.\ tails). For example, generalizations from Gaussian approximating densities to Student-$t$ ones have been explored in the context of Laplace,   VB and EP \citep[e.g.][]{daudel2023monotonic,ding2011t,futami2017expectation,gelman1995bayesian,liang2022fat}, but there are no simple strategies, to date, for including skewness within such extensions. Our proposed solution applies also to these symmetric approximations. As proved in Section~\ref{sec_34}, the perturbation of these increasingly accurate symmetric densities opens also the avenues to further improve the asymptotic convergence rates to the target posterior.

%###############################################################################
\subsection{\large 2.2. An  overview   of skew-symmetric distributions} \label{sec_2.a}
As anticipated in Section~\ref{sec_1}, the novel strategy we propose to perturb an already-available symmetric approximation of a generic target posterior leads to an improved asymmetric counterpart, which crucially belongs to the skew-symmetric family \citep[e.g.][]{azzalini2003distributions,genton2005generalized,ma2004flexible,wang2004skew}. This broad family has been originally proposed to increase the flexibility of symmetric densities via the inclusion of a suitably designed skewing factor that preserves the tractability of the unperturbed counterparts. Definition~\ref{def_1_or} \citep[e.g.][]{azzalini2013skew,wang2004skew},
 clarifies that this objective can be accomplished through a direct perturbation scheme that meets simple conditions.

\begin{jrssdefinition}{Definition 1}
\noindent (Skew-symmetric densities). Let  {$\bar{q}_{{\btheta}^*}({\btheta})$} be a density for $\btheta \in \Theta \subseteq \mathbb{R}^{d}$, which is symmetric at ${\btheta}^*$, i.e.\ $\bar{q}_{{\btheta}^*}(\btheta)=\bar{q}_{{\btheta}^*}(2{\btheta}^*-\btheta)$ for all $\btheta \in \Theta $, and denote with $w_{{\btheta}^*}(\btheta)$ a skewing factor such that $w_{{\btheta}^*}(\btheta) \in [0,1]$ and $w_{{\btheta}^*}(\btheta)=1-w_{{\btheta}^*}(2 {\btheta}^*-\btheta)$ for any $\btheta \in \Theta $. Then
\begin{equation}
q_{{\btheta}^*}(\btheta)=2\bar{q}_{{\btheta}^*}({\btheta})w_{{\btheta}^*}(\btheta),
 \label{skew:sym:gen}
 \end{equation}
is a skew-symmetric density.
\label{def_1_or}
\end{jrssdefinition}

Definition~\ref{def_1_or} provides a general skewness-inducing mechanism that leverages a suitably designed perturbation factor  $w_{{\btheta}^*}(\btheta)$ to redistribute the density between each point $\btheta$ and its polar opposite $2 {\btheta}^*-\btheta$, with respect to a notion of central symmetry. While general, such a definition crucially encompasses several special cases of routine use in practice. A remarkable example is provided by the sub-class of generalized skew-elliptical distributions which can be obtained by replacing $\bar{q}_{{\btheta}^*}({\btheta})$ in \eqref{skew:sym:gen} with an elliptical density, including, for instance, those of multivariate Gaussians, Student-$t$, and Cauchy, among others \citep[e.g.][]{genton2005generalized}. Perturbing these latter densities through a specific skewing factor $w_{{\btheta}^*}(\btheta)$ set equal to a suitably defined cumulative distribution function induced by the family of densities $\bar{q}_{{\btheta}^*}({\btheta})$ further yields classical skew-normal, skew-$t$ and skew-Cauchy densities, among others \citep[][]{azzalini1999statistical,azzalini2003distributions,azzalini1996multivariate,genton2005generalized}. 

Albeit different, all the above alternatives share a common property, which characterizes the general skew-symmetric family in  Definition~\ref{def_1_or}. In particular, as is clear from \eqref{skew:sym:gen}, the tractability of $q_{{\btheta}^*}(\btheta)$ inherently depends on the one of its symmetric counterpart $\bar{q}_{{\btheta}^*}({\btheta})$ and on the choice of the skewing factor $w_{{\btheta}^*}(\btheta)$. When both quantities are tractable, the density in \eqref{skew:sym:gen} can be computed analytically, and further admits a straightforward i.i.d.\ sampling scheme that facilitates direct Monte Carlo evaluation of any functional of $q_{{\btheta}^*}(\btheta)$. This  follows from the stochastic representation of skew-symmetric distributions in Proposition~\ref{prop1_aug_data} \citep[e.g.][]{wang2004skew}.

\begin{jrssproposition}{Proposition 1}
\noindent (Stochastic representation).
Let  ${\btheta}$ be a random variable having density ${q}_{{\btheta}^*}({\btheta})$ as in \eqref{skew:sym:gen}, and denote with $\mathds{1}[\cdot]$ the indicator function, then
\begin{equation*}
{\btheta}\stackrel{d}{=}\mathds{1}[\textsc{u} \leq w_{{\btheta}^*}(\bar{\btheta})]\bar{\btheta}+(1-\mathds{1}[\textsc{u} \leq w_{{\btheta}^*}(\bar{\btheta})])(2{{\btheta}^*}-\bar{\btheta}), 
\end{equation*}
\noindent where $\bar{\btheta}$ has symmetric density $\bar{q}_{{\btheta}^*}(\bar {\btheta})$, and $\textsc{u}$ is uniformly distributed  in $[0,1]$.
 \label{prop1_aug_data}
 \end{jrssproposition}

As a consequence of Proposition~\ref{prop1_aug_data}, i.i.d.\ sampling from skew-symmetric distributions reduces to drawing values from the symmetric component and then retaining the sampled value or its symmetric with respect to ${\btheta}^*$ depending on whether the skewing factor evaluated at such a sampled value is below or above the realization from a uniform distribution in $[0,1]$. This strategy is crucial to facilitate tractable Monte Carlo inference under the proposed skew-symmetric approximation {${q}_{n,{\btheta}^*}(\btheta)$}  of posterior densities $\pi_n(\btheta)$. As clarified in Section~\ref{sec_32}, this approximation arises from the perturbation of tractable symmetric counterparts $\bar{q}_{n,{\btheta}^*}(\btheta)$ (often Gaussians)  via a skewing factor $w_{n,{\btheta}^*}(\btheta)$  that involves only prior $\pi(\btheta)$ and likelihood $L(\btheta; \by_{1:n})$ evaluation.

%###############################################################################
\subsection{\large 2.3 Skew-symmetric perturbation of symmetric approximations} \label{sec_32}
Recalling Section~\ref{sec_1}, the class of skewed approximations we propose is motivated by a fundamental, but yet overlooked, skew-symmetric representation of posterior densities stated in Proposition~\ref{prop1}.

\begin{jrssproposition}{Proposition 2}
\noindent (Skew-symmetric representation of posterior densities).
Consider the generic posterior density ${\pi}_{n}(\btheta)=\pi(\btheta)L(\btheta; \by_{1:n})/c(\by_{1:n})$  for the parameter $\btheta \in \Theta $, where $\pi(\btheta)$ denotes the prior, $L(\btheta; \by_{1:n})$ corresponds to the likelihood, while $c(\by_{1:n})$ is the normalizing constant. Moreover, denote with 
\begin{equation}
\bar{\pi}_{n,{\btheta}^*}(\btheta)=\frac{\pi_{n}(\btheta)+\pi_{n}(2{\btheta}^*-\btheta)}{2},
\label{eq1}
\end{equation}
the symmetrized form of this posterior density about a point ${\btheta}^*\in \Theta$, and let 
\begin{equation} \label{eq_w}
w_{n,{\btheta}^*}(\btheta)=	\frac{\pi(\btheta)L(\btheta; \by_{1:n})}{\pi(\btheta)L(\btheta; \by_{1:n}) +\pi(2 {\btheta}^* -\btheta)L(2 {\btheta}^* -\btheta; \by_{1:n})},
\end{equation}
with the convention that $w_{n,{\btheta}^*}(\btheta)=0.5$ if $\pi(\btheta)L(\btheta; \by_{1:n}) = \pi(2 {\btheta}^* -\btheta)L(2 {\btheta}^* -\btheta; \by_{1:n}) = 0$ (see Remark \ref{remark:w:0.5} for details).
 Then, the posterior density ${\pi}_{n}(\btheta)$ can be equivalently re-expressed in skew-symmetric form as
\begin{equation}
{\pi}_{n}(\btheta)=2\bar{\pi}_{n,{\btheta}^*}(\btheta)w_{n,{\btheta}^*}(\btheta),
\label{eq2}
 \end{equation}
for any ${\btheta}^* \in \Theta$ and  $n$, with $w_{n,{\btheta}^*}(\btheta) \in [0,1]$ and $w_{n,{\btheta}^*}(\btheta)=1-w_{n,{\btheta}^*}(2 {\btheta}^*-\btheta)$.
 \label{prop1}
 \end{jrssproposition}

 \begin{jrssremark}{Remark 1}
\noindent While other forms of symmetrization can be considered, the one defined in~\eqref{eq1} is arguably the most natural, direct and simple, in that it redistributes equally to every pair of symmetric points the total posterior density assigned to this pair. Although it has been generally overlooked and, to our knowledge, never explored in the context of Bayesian approximations,  such a symmetrization has been successfully employed in classical frequentist literature to improve standard estimators of empirical distribution functions associated with underlying symmetric densities  \citep[e.g.][]{hinkley1976estimating,lo1985estimation,schuster1975estimating,schuster1987identifying}. Our contribution leverages such a symmetrization with a substantially different focus. Namely that of designing principled and tractable skew-symmetric approximations of generic posterior densities from the direct perturbation of already-derived symmetric ones. 
\label{rema1}
\end{jrssremark}

The proof of Proposition~\ref{prop1}  simply requires to notice that  ${\pi}_{n}(\btheta)$ can be equivalently re-written as $2{\bar{\pi}_{n,{\btheta}^*}(\btheta)}[{\pi}_{n}(\btheta)/(2{\bar{\pi}_{n,{\btheta}^*}(\btheta)})]$. Therefore, replacing in the ratio ${\pi}_{n}(\btheta)/(2\bar{\pi}_{n,{\btheta}^*}(\btheta))$ the quantity ${\bar{\pi}_{n,{\btheta}^*}(\btheta)}$ with its expression in \eqref{eq1}, and ${\pi}_{n}(\btheta)$ with $\pi(\btheta)L(\btheta; \by_{1:n})/c(\by_{1:n})$,   yields the skewness-inducing factor $w_{n,{\btheta}^*}(\btheta)$ in \eqref{eq_w}, after noticing that the intractable normalizing constant $c(\by_{1:n})$ cancels out within the ratio between the target posterior density and its symmetrized form. The $[0,1]$ range for {$w_{n,{\btheta}^*}(\btheta)$} and the property $w_{n,{\btheta}^*}(\btheta)=1-w_{n,{\btheta}^*}(2 {\btheta}^*-\btheta)$ follow from its definition in~\eqref{eq_w}. As clarified later, such properties for {$w_{n,{\btheta}^*}(\btheta)$} are fundamental to ensure that the proposed perturbation for $\bar{q}_{n,{\btheta}^*}(\btheta)$ belongs to the skew-symmetric family \citep[e.g.][]{azzalini2003distributions,ma2004flexible,wang2004skew}.

Proposition~\ref{prop1} relates to a core result in  \citet{wang2004skew}, which establishes the existence and uniqueness of skew-symmetric representations for generic densities.  This parallel ensures that the equivalent expression for ${\pi}_{n}(\btheta)$ in Proposition~\ref{prop1} is the one of a skew-symmetric density. According to \eqref{eq2}, such a density is equal to  the product between a tractable and analytically available skewing factor $w_{n,{\btheta}^*}(\btheta)$ (where the troublesome normalizing constant  $c(\by_{1:n})$ cancels out) and a symmetrized posterior $\bar{\pi}_{n,{\btheta}^*}(\btheta)$ that is often as intractable as the target ${\pi}_{n}(\btheta)$. Such a latter issue prevents  \eqref{eq2} from having a direct impact in facilitating posterior inference and, for this reason, the representation in Proposition~\ref{prop1}  has never appeared within Bayesian statistics. One possibility to address the above issue, would be to replace $\bar{\pi}_{n,{\btheta}^*}(\btheta)$ in  \eqref{eq2} with a tractable symmetric approximation. Lemma~\ref{lemma_1} suggests that a natural candidate is the original, already-derived, symmetric approximation  $\bar{q}_{n,{\btheta}^*}(\btheta)$ of the target posterior density ${\pi}_{n}(\btheta)$, which we aim to improve through a skewed perturbation. In particular,  Lemma~\ref{lemma_1} proves that $\bar{q}_{n,{\btheta}^*}(\btheta)$ approximates $\bar{\pi}_{n,{\btheta}^*}(\btheta)$ more accurately than  ${\pi}_{n}(\btheta)$, under  the TV distance $\mathcal{D}_{\textsc{tv}}[p \mid \mid q]=(1/2) \int|p(\btheta)-q(\btheta)|\mbox{d} \btheta$ and any generic $\alpha$-divergence $\mathcal{D}_{\alpha}[p \mid \mid q]=[1/(\alpha(1-\alpha))]\smash{( 1-\int p(\btheta)^{\alpha}q(\btheta)^{1-\alpha}\mbox{d} \btheta)}$ for $\alpha \in \mathbb{R} \setminus \{0,1\}$ \citep[see][]{cichocki2010families,margossian2024ordering,poczos2011estimation}, including its limiting KL ($\textsc{kl}[q \mid \mid p]=\int q(\btheta)\log[q(\btheta)/p(\btheta)]\mbox{d} \btheta$) and reverse-KL ($\textsc{kl}[p \mid \mid q]=\int p(\btheta)\log[p(\btheta)/q(\btheta)]\mbox{d} \btheta$) forms, which can be obtained by letting $\alpha \to 0$ and $\alpha \to 1$, respectively, in   $\mathcal{D}_{\alpha}$  \citep[e.g.][]{cichocki2010families}. Hence, Lemma~\ref{lemma_1}, and the subsequent theoretical results we derive, hold under the most widely used divergences in the context of deterministic approximations. Recalling Section~\ref{sec_2}, the divergences $\textsc{kl}[q \mid \mid p]$ and $\textsc{kl}[p \mid \mid q]$  enter the optimization problems in VB and EP  \citep[e.g.][]{blei2017variational,vehtari2020expectation}, while TV is a reference distance in the study of the asymptotic accuracy of standard approximations.   Recent literature has also explored variational approximations based on generic $\alpha$-divergences \citep[][]{hernandez2016black,yang2020alpha}. Hence, our methods and theory apply also to the symmetric approximations from these strategies.

\begin{jrsslemma}{Lemma 1}
\noindent Let ${\pi}_{n}(\btheta)$ be a generic posterior density for the parameter $\btheta \in \Theta $, and denote with $\bar{q}_{n,{\btheta}^*}(\btheta)$ an already-derived approximation of ${\pi}_{n}(\btheta)$ which is symmetric about the point ${\btheta}^*\in \Theta$. Moreover, consider the symmetrized posterior density $\bar{\pi}_{n,{\btheta}^*}(\btheta)$  about ${\btheta}^*$ defined as in \eqref{eq1}. Then,  for any ${\btheta}^* \in \Theta$ and  $n$, it holds
\begin{equation*}
\mathcal{D}[\bar{\pi}_{n,{\btheta}^*} \mid \mid \bar{q}_{n,{\btheta}^*}] \leq \mathcal{D}[\pi_{n} \mid \mid \bar{q}_{n,{\btheta}^*}],
 \end{equation*}
 where $\mathcal{D}$ is either the TV  distance, KL, reverse-KL or a generic $\alpha$-divergence.
\label{lemma_1}
\end{jrsslemma}

As clarified in Definition~\ref{def_1}, combining Proposition~\ref{prop1} with Lemma~\ref{lemma_1} suggests a natural strategy to obtain the newly proposed skew-symmetric approximation. In particular, by Lemma~\ref{lemma_1}, $\bar{q}_{n,{\btheta}^*}(\btheta)$ approximates $\bar{\pi}_{n,{\btheta}^*}(\btheta)$ more accurately than  ${\pi}_{n}(\btheta)$, where $\bar{\pi}_{n,{\btheta}^*}(\btheta)$ can be in turn expressed via  $\bar{\pi}_{n,{\btheta}^*}(\btheta)={\pi}_{n}(\btheta)/[2w_{n,{\btheta}^*}(\btheta)]$ as a direct consequence of \eqref{eq2}. Combining these two results yields  $\bar{q}_{n,{\btheta}^*}(\btheta)\approx\bar{\pi}_{n,{\btheta}^*}(\btheta)={\pi}_{n}(\btheta)/[2w_{n,{\btheta}^*}(\btheta)]$, which implies $2\bar{q}_{n,{\btheta}^*}(\btheta)w_{n,{\btheta}^*}(\btheta)\approx{\pi}_{n}(\btheta)$. Thus, setting $q_{n,{\btheta}^*}(\btheta)=2\bar{q}_{n,{\btheta}^*}(\btheta)w_{n,{\btheta}^*}(\btheta)$ gives the  improved skew-symmetric approximation in Definition~\ref{def_1}.

\begin{jrssdefinition}{Definition 2}
\noindent (Skew-symmetric approximation of posterior densities). Consider the generic posterior density ${\pi}_{n}(\btheta)$ for $\btheta \in \Theta $, and  let $\bar{q}_{n,{\btheta}^*}(\btheta)$  denote an already-derived approximation of ${\pi}_{n}(\btheta)$ which is symmetric about the point ${\btheta}^*\in \Theta$. Moreover, let $w_{n,{\btheta}^*}(\btheta) \in [0,1]$ be the skewing factor  in \eqref{eq_w}. Then, the skew-symmetric approximation of ${\pi}_{n}(\btheta)$ arising from the perturbation of  $\bar{q}_{n,{\btheta}^*}(\btheta)$  is defined as
\begin{equation} 
\begin{split}
q_{n,{\btheta}^*}(\btheta)&=2\bar{q}_{n,{\btheta}^*}(\btheta) w_{n,{\btheta}^*}(\btheta)\\
&=2\bar{q}_{n,{\btheta}^*}(\btheta)\frac{\pi(\btheta)L(\btheta; \by_{1:n})}{\pi(\btheta)L(\btheta; \by_{1:n})+\pi(2 {\btheta}^*-\btheta)L(2 {\btheta}^*-\btheta;\by_{1:n})},
\end{split}
\label{skew:sym:post}
 \end{equation}
 for every  symmetry point ${\btheta}^* \in \Theta$ and sample size $n$, where $\pi(\btheta)$ and $L(\btheta; \by_{1:n})$ are, respectively, the prior and likelihood inducing the posterior  ${\pi}_{n}(\btheta)$.
 \label{def_1}
\end{jrssdefinition}

 \begin{jrssremark}{Remark 2}
\noindent Notice that $w_{n,{\btheta}^*}(\btheta)$ in  \eqref{skew:sym:post} admits a natural  interpretation, in that it coincides with the relative proportion of the posterior density at  $\btheta \in \Theta $ with respect to the total assigned to such a $\btheta$ and its symmetric counterpart $2 {\btheta}^*-\btheta$. This yields a skewing factor which quantifies differences in the posterior density at the symmetric points $\{\btheta,2 {\btheta}^*-\btheta\}$ for any $\btheta \in \Theta $. Therefore, if the posterior density is actually symmetric about ${\btheta}^*$, then $w_{n,{\btheta}^*}(\btheta)=0.5$ for all $\btheta \in \Theta $, and hence, $q_{n,{\btheta}^*}(\btheta)$ reduces to $\bar{q}_{n,{\btheta}^*}(\btheta)$, as expected. Conversely, whenever there are asymmetries within ${\pi}_{n}(\btheta)$, the original symmetric approximation $\bar{q}_{n,{\btheta}^*}(\btheta)$  is  re-weighted by $w_{n,{\btheta}^*}(\btheta)$ in order to properly redistribute the total density at each  pair $\{\btheta,2 {\btheta}^*-\btheta\}$ according to the one assigned by the actual posterior to $\btheta $ and $2{\btheta}^*-\btheta$. This yields an improved approximation $q_{n,{\btheta}^*}(\btheta)$ incorporating the skewness of ${\pi}_{n}(\btheta)$ with respect to the known symmetry point ${\btheta}^*$.
\label{re1}
\end{jrssremark}

 \begin{jrssremark}{Remark 3}
\noindent In the skew-symmetric representation in Proposition \ref{prop1}, we rely on the convention that $w_{n,{\btheta}^*}(\btheta)=0.5$ if $\pi(\btheta)L(\btheta; \by_{1:n}) = \pi(2 \btheta^* -\btheta)L(2 \btheta^* -\btheta; \by_{1:n}) = 0$, where the latter equality also implies that $\pi_n(\btheta) = \pi_n(2  \btheta^* - \btheta) = 0$. In this setting, which arises, for example, when the support of $\btheta$ is bounded, both the numerator and the denominator of the skewing factor in \eqref{eq_w} are zero, and hence, $w_{n,{\btheta}^*}(\btheta)$ is undefined, thus requiring an alternative specification. Note that this alternative specification is not necessary to guarantee the validity of the skew-symmetric representation for $\pi_n(\btheta)$, provided that also the symmetrized posterior $\bar{\pi}_{n,{\btheta}^*}(\btheta)$ is zero whenever $\pi_n(\btheta) = \pi_n(2 \btheta^* - \btheta) = 0$. Hence, in \eqref{eq2}  different values of $w_{n,{\btheta}^*}(\btheta)$ are irrelevant at these zero-density points. However, $ \pi_n(\btheta) = \pi_n(2 \btheta^* - \btheta) = 0$ does not  imply that $ \bar{q}_{n,{\btheta}^*}(\btheta) = 0$. This situation arises, for example, when Gaussian approximations are considered for parameters $\btheta$  with bounded support. As such, when $\pi_n(\btheta) = \pi_n(2 \btheta^* - \btheta) = 0$ and $\bar{q}_{n,{\btheta}^*}(\btheta)  > 0 $, the skewing factor must be specified alternatively  to guarantee that  \eqref{skew:sym:post} yields a valid  skew-symmetric density. From Definition~\ref{def_1_or}, a key condition for this result to hold is that   $w_{n,{\btheta}^*}(\btheta) =1 - w_{n,{\btheta}^*}(2 {\btheta}^* - \btheta)$ with $w_{n,{\btheta}^*}(\btheta) \in [0,1]$, for any point $\btheta$ satisfying $\bar{q}_{n,{\btheta}^*}(\btheta) > 0 $. To this end, setting $w_{n,{\btheta}^*}(\btheta) = 0.5 $ whenever $\pi_n(\btheta) = \pi_n(2  \btheta^* - \btheta) = 0$ and $\bar{q}_{n,{\btheta}^*}(\btheta) > 0 $ is arguably the most direct strategy to meet such a condition. Recalling Remark~\ref{re1}, $w_{n,{\btheta}^*}(\btheta) = 0.5 $ implies that $\pi_n(\btheta) = \pi_n(2\btheta^* - \btheta) $, and hence, the convention that $w_{n,{\btheta}^*}(\btheta) = 0.5 $ if $\pi_n(\btheta) = \pi_n(2\btheta^* - \btheta) = 0$ is also the most natural one. 
\label{remark:w:0.5}
\end{jrssremark}

As clarified in Definition~\ref{def_1}, the proposed approximation $q_{n,{\btheta}^*}(\btheta)$ results from the re-weighting of the known $\bar{q}_{n,{\btheta}^*}(\btheta)$ by a skewing factor $w_{n,{\btheta}^*}(\btheta)$ via a strategy which does not necessitate additional optimization costs relative to those required for deriving $\bar{q}_{n,{\btheta}^*}(\btheta)$. In fact, the expression for $w_{n,{\btheta}^*}(\btheta)$ in \eqref{eq_w} does not depend on additional unknown parameters beyond ${\btheta}^*$, which is in turn available as the output of the already-solved optimization problem that targeted the posterior ${\pi}_{n}(\btheta) $ via the symmetric density $\bar{q}_{n,{\btheta}^*}(\btheta)$ to be perturbed. Proposition~\ref{prop2} below guarantees that, albeit more flexible than $\bar{q}_{n,{\btheta}^*}(\btheta)$, the deterministic approximation  $q_{n,{\btheta}^*}(\btheta)$ in  Definition~\ref{def_1} preserves similar tractability in inference, in that it belongs to the known class of skew-symmetric densities presented in Section~\ref{sec_2.a}.

\begin{jrssproposition}{Proposition 3}
\noindent The expression for $q_{n,{\btheta}^*}(\btheta)$ given in~\eqref{skew:sym:post} coincides with the density of a skew-symmetric distribution having $\bar{q}_{n,{\btheta}^*}(\btheta)$ as symmetric component and  $w_{n,{\btheta}^*}(\btheta)$ as skewing factor.
 \label{prop2}
 \end{jrssproposition}

Proposition~\ref{prop2} follows directly from Definition~\ref{def_1_or}, after noticing that, by construction, $\bar{q}_{n,{\btheta}^*}(\btheta)$ is symmetric at ${\btheta}^*$ and, in view of Proposition~\ref{prop1},  $w_{n,{\btheta}^*}(\btheta)$ has support within $[0,1]$ and satisfies $w_{n,{\btheta}^*}(\btheta)=1-w_{n,{\btheta}^*}(2 {\btheta}^*-\btheta)$. As anticipated in Section~\ref{sec_2.a}, the connection with skew-symmetric distributions established in Proposition~\ref{prop2}  is crucial in facilitating inference also under $q_{n,{\btheta}^*}(\btheta)$. More concretely, the stochastic representation in Proposition~\ref{prop1_aug_data} yields a simple  rejection-free i.i.d.\ sampling scheme from any density within the skew-symmetric family, including $q_{n,{\btheta}^*}(\btheta)$, thereby allowing tractable and effective Monte Carlo evaluation of any functional of interest under the improved skewed approximation. This sampling scheme from $q_{n,{\btheta}^*}(\btheta)$ is outlined in Algorithm~\ref{alg1}.

 \begin{algorithm}
\caption{i.i.d.\ sampling  from the skew-symmetric approximation in \eqref{skew:sym:post}}\label{alg1}
\begin{algorithmic}
         \For{$s=1, \ldots, N_{\textsc{sampl}}$}    
\State  \textbf{1}. Sample $\bar{\btheta}^{(s)}$ from the distribution with symmetric density $\bar{q}_{n,{\btheta}^*}(\bar \btheta)$.
\State  \textbf{2}. Sample $\textsc{u}^{(s)} \sim \mbox{Unif}[0,1]$.
\State \textbf{3}. If $ \textsc{u}^{(s)} \leq w_{n,{\btheta}^*}(\bar{\btheta}^{(s)}) $ set $\btheta^{(s)}=\bar{\btheta}^{(s)}$, otherwise set $\btheta^{(s)}=2 {\btheta}^*-\bar{\btheta}^{(s)}$.
         \EndFor  \\
         \textbf{output:} i.i.d.\ samples $\btheta^{(1)}, \ldots, \btheta^{(N_{\textsc{sampl}})}$  from the  skew-symmetric approximation in \eqref{skew:sym:post}.
\end{algorithmic}
\end{algorithm}    

Notice that Algorithm~\ref{alg1} only requires simulation from the  symmetric approximation $\bar{q}_{n,{\btheta}^*}(\btheta)$ and computation of the  skewing factor $w_{n,{\btheta}^*}(\btheta)$, which is analytically-available in Definition~\ref{def_1}, and it does not depend on intractable quantities. The first task is straightforward whenever the unperturbed density $\bar{q}_{n,{\btheta}^*}(\btheta)$ arises from one of the routinely implemented approximation schemes discussed in Section~\ref{sec_2}. Moreover, when the interest lies in more complex functionals, sampling is often needed also for inference under $\bar{q}_{n,{\btheta}^*}(\btheta)$. The second task requires instead the evaluation of the likelihood, which enters (together with the prior) the definition of  $w_{n,{\btheta}^*}(\btheta)$. Although this yields some increments in sampling costs relative to inference under $\bar{q}_{n,{\btheta}^*}(\btheta)$, it is important to emphasize that multiple internal likelihood evaluations are  standard in popular algorithms for Bayesian inference. Important examples include Metropolis--Hastings, Hamiltonian Monte Carlo, sequential Monte Carlo, EP, coordinate ascent variational inference, and its  black-box extensions \citep[e.g.][]{blei2017variational,chopin2020introduction,Chopin_2017,kucukelbir2017automatic,ranganath2014black,tan2023variational,vehtari2020expectation}. However, unlike for Algorithm~\ref{alg1}, these schemes introduce additional complexities, which often imply further tuning and non-negligible increments in the overall runtimes. Hence, the cost of inference under $q_{n,{\btheta}^*}(\btheta)$ is often dominated by the one required to obtain its less accurate symmetric counterpart $\bar{q}_{n,{\btheta}^*}(\btheta)$ via state-of-the-art algorithms, and it substantially reduces the  runtimes of other routine use computational methods. These important gains are illustrated in Section~\ref{sec_logistic} for a challenging real-data application with $n=30{,}524$ and $d = 62$. In this case, sampling $10{,}000$ values from the  skew-symmetric approximation under Algorithm~\ref{alg1} requires only $8$ seconds. This runtime is orders of magnitude lower than, for example, the $\approx$ 1 hour required to obtain a sample of the same size from Hamiltonian Monte Carlo implemented in \texttt{rstan}. 

Note also that, unlike for most of the aforementioned computational methods, Algorithm~\ref{alg1}  is inherently amenable to parallel implementations both across samples and likelihood factors. Thus, the runtime of Algorithm~\ref{alg1} can be easily reduced to fractions of seconds. Moreover, as clarified within Section~\ref{eff_ev} below, in several routinely implemented models such as, e.g.\ generalized linear models, it is possible to devise computationally efficient strategies that further optimize the cost of the likelihood evaluations required to compute the skewing factor.

%###############################################################################
\subsection{\large 2.4 Efficient evaluation of the skewing factor}\label{eff_ev}
Recalling the above discussion, inference under the newly proposed skew-symmetric approximation  $q_{n,{\btheta}^*}(\btheta)$ requires the evaluation of the two un-normalized posterior densities $\pi(\btheta)L(\btheta; \by_{1:n})$ and $\pi(2  \btheta^* - \btheta)L(2  \btheta^* - \btheta; \by_{1:n})$ in the skewing factor $w_{n,{\btheta}^*}(\btheta)$. Hence, the increments in cost relative to performing inference under $\bar{q}_{n,{\btheta}^*}(\btheta)$ mainly depend on those of evaluating the likelihoods $L(\btheta; \by_{1:n})$ and $L(2  \btheta^* - \btheta; \by_{1:n})$. As mentioned above, multiple likelihoods evaluations are standard in Bayesian computation and can be performed with negligible cost for most models of practical interest. Nevertheless, designing strategies that further optimize this cost could lead to additional reductions in the overall runtimes. This is the case, for example, in high-dimensional (i.e.\ large $d$) models, with  the parameters entering the likelihood function through a linear predictor, as in routinely implemented generalized linear models.

Under these models, Algorithm \ref{alg2} provides an efficient approach for computing  $w_{n,{\btheta}^*}(\btheta)$  at a cost which is essentially that of a single likelihood evaluation instead of the two required in the expression of $w_{n,{\btheta}^*}(\btheta)$. More specifically, Algorithm \ref{alg2} can be applied to the broad class of models whose log-likelihood  $\ell(\btheta; \by_{1:n}) =\log L(\btheta; \by_{1:n})$ can be expressed as $\ell(\btheta; \by_{1:n}) = \sum_{i=1}^{n} g_i(\by_{i},\eta_{\btheta,i} ),$ where $g_i(\cdot),\, i = 1,\dots, n,$ denote functions that can be evaluated with limited computational effort, while $\eta_{\btheta,i} = \bx_i^\intercal \btheta$ for $i = 1,\dots, n,$ are linear predictors depending on $\btheta$ and on a $d$-dimensional vector of explanatory variables $\bx_i$. Within this context, for (at least moderately) large $d$, the most expensive operation is the evaluation of the linear predictor. Algorithm~\ref{alg2} reduces the cost required for computing $w_{n,{\btheta}^*}(\btheta)$ by exploiting the fact that 
\begin{equation*}
\ell(\btheta; \by_{1:n}) = \sum\nolimits_{i=1}^{n} g_i(\by_{i},\eta_{ \btheta^* ,i} + \eta_{\btheta - \btheta^* , i} ),  \ \ \mbox{and} \ \  \ell(2 \btheta^* - \btheta; \by_{1:n}) = \sum\nolimits_{i=1}^{n} g_i(\by_{i},\eta_{\btheta^* ,i} - \eta_{\btheta - \btheta^* , i} ),
\end{equation*}
where  $\eta_{\btheta^* ,i} = \bx_i^\intercal \btheta^*$ and  $ \eta_{\btheta -  \btheta^* , i} = \bx_i^\intercal (\btheta -  \btheta^*),\, i = 1,\dots,n$. Crucially, $\bbbeta_{\btheta^*} = (\eta_{ \btheta^* ,1}, \dots,\eta_{\btheta^* , n} )^\intercal $ can be pre-computed, meaning that only a single evaluation of $\bbbeta_{\btheta -\btheta^*} = (\eta_{\btheta - \btheta^* ,1}, \dots, \eta_{\btheta - \btheta^* ,n})^\intercal  $ is required in order to compute the skewing factor. As a result, since the cost of $\bx_i^\intercal  \btheta$ and $g_i(\cdot)$ is $O(d)$ and $O(1)$, respectively, Algorithm \ref{alg2} allows to evaluate $w_{n,{\btheta}^*}(\btheta)$ in a number of operations which, for $d$ large, is essentially reduced by a factor of two.    This yields a sampling strategy whose computational complexity interestingly matches, or even improves, those of other available skewed approximations \citep[see, e.g.][]{anceschi2022bayesian,durante2023skewed,fasanoscalable}. Notice that, unlike for the general perturbation strategy we propose, these alternatives are developed in the context of specific deterministic approximations (e.g.\ Laplace and VB) and/or models (e.g.\ probit regression and its extensions). Hence, matching the costs of these ad-hoc implementations is arguably a remarkable result.

 \begin{algorithm}
\caption{Efficient evaluation of $w_{n,{\btheta}^*}(\btheta)$ in models with linear predictors}\label{alg2}
\begin{algorithmic}
\State \textbf{require:} $\bbbeta_{\btheta^*}$, $\btheta -  \btheta^*$ and   the design matrix $\bX$ with rows $\bx_i^\top$. 
\State \textbf{do} 
\State \quad  \textbf{1}. Evaluate $\bbbeta_{\btheta - \btheta^*,i} = \bx_i^{\intercal}(\btheta -  \btheta^*)$, for $i=1, \ldots, n$.
\State \quad \textbf{2}. Compute $\ell(\btheta; \by_{1:n}) = \sum_{i=1}^n g_i(\by_{i},\eta_{\btheta^*,i} + \eta_{\btheta -\btheta^*,i})$.
\State \quad \textbf{3}. Compute  $\ell(2 \btheta^* - \btheta; \by_{1:n}) = \sum_{i=1}^n g_i(\by_{i},\eta_{\btheta^*,i} - \eta_{\btheta -  \btheta^*,i})$.\\
         \textbf{output:}  $$w_{n,{\btheta}^*}(\btheta) = \frac{\pi(\btheta)\exp(\ell(\btheta; \by_{1:n}) )}{ \pi(\btheta)\exp( \ell(\btheta; \by_{1:n}) ) + \pi(2 \btheta^* - \btheta)\exp( \ell(2 \btheta^* -\btheta; \by_{1:n}) ) }.$$
\end{algorithmic}
\end{algorithm}    

%###############################################################################
%###############################################################################
\section{\large 3. Theoretical properties of skew-symmetric approximations} \label{sec_theo}
Sections~\ref{sec_33}--\ref{sec_34} clarify that the skew-symmetric approximation  $q_{n,{\btheta}^*}(\btheta)$ derived in Section~\ref{sec_3} (see also Definition~\ref{def_1}) is not only computationally tractable, but also yields a provably more accurate characterization of the exact  posterior ${\pi}_{n}(\btheta)$. See the Supplementary Material for proofs.

%###############################################################################
\subsection{\large 3.1 Finite sample properties and optimality} \label{sec_33}
The original motivation behind the skew-symmetric approximation $q_{n,{\btheta}^*}(\btheta)$ in   \eqref{skew:sym:post} is to improve the accuracy of the unperturbed $\bar{q}_{n,{\btheta}^*}(\btheta)$. Theorem~\ref{teo_1} provides finite sample theoretical support  to such an accuracy gain and clarifies that the quality of $q_{n,{\btheta}^*}(\btheta)$ only depends on how accurate is $\bar{q}_{n,{\btheta}^*}(\btheta)$ in approximating the symmetrized posterior $\bar{\pi}_{n,{\btheta}^*}(\btheta)$ in~\eqref{eq1}. These results are deepened in Theorem~\ref{teo_2}, which proves  that the skewing factor $w_{n,{\btheta}^*}(\btheta)$ in Definition~\ref{def_1} is optimal among all those yielding a skew-symmetric approximation for ${\pi}_{n}(\btheta)$, with $\bar{q}_{n,{\btheta}^*}(\btheta)$ as symmetric component. 

\begin{jrsstheorem}{Theorem 1}
\noindent (Finite sample accuracy).
Consider the generic posterior density ${\pi}_{n}(\btheta)$  for the parameter $\btheta \in \Theta $, and let $\bar{q}_{n,{\btheta}^*}(\btheta)$ correspond to an already-derived approximation for ${\pi}_{n}(\btheta)$ which is symmetric about the point $\btheta^*\in \Theta$. Moreover, let $q_{n,{\btheta}^*}(\btheta)=2\bar{q}_{n,{\btheta}^*}(\btheta)w_{n,{\btheta}^*}(\btheta)$, where $w_{n,{\btheta}^*}(\btheta)$ is defined as in~\eqref{eq_w}. Then, for any symmetry point ${\btheta}^* \in \Theta$ and sample size $n$, it holds
\begin{equation} \label{skew:sym:equiv:sym}
	\mathcal{D}[{\pi}_{n} \mid \mid q_{n,{\btheta}^*}] = \mathcal{D}[\bar{\pi}_{n,{\btheta}^*} \mid \mid \bar{q}_{n,{\btheta}^*}], 
\end{equation}
 where $\mathcal{D}$ is either the TV  distance, KL, reverse-KL or a generic $\alpha$-divergence, while $\bar{\pi}_{n,{\btheta}^*}(\btheta)$ corresponds to the symmetrized posterior density defined in~\eqref{eq1}.  In view of Lemma \ref{lemma_1}, the result in \eqref{skew:sym:equiv:sym} implies also
\begin{equation} \label{skew:sym:non:asymp}
\mathcal{D}[{\pi}_{n} \mid \mid q_{n,{\btheta}^*}] \leq \mathcal{D}[\pi_{n} \mid \mid\bar{q}_{n,{\btheta}^*}],
\end{equation}
for any ${\btheta}^* \in \Theta$ and sample size  $n$.  
\label{teo_1}
 \end{jrsstheorem}
 
 Theorem~\ref{teo_1} states two important results. First, as clarified in  \eqref{skew:sym:equiv:sym}, the overall quality of $q_{n,{\btheta}^*}(\btheta)$  coincides with the one achieved by the unperturbed $\bar{q}_{n,{\btheta}^*}(\btheta)$ in approximating the symmetrized posterior $\bar{\pi}_{n,{\btheta}^*}(\btheta)$ defined in~\eqref{eq1}. Second, according to \eqref{skew:sym:non:asymp}, $q_{n,{\btheta}^*}(\btheta)$ is, provably, never less accurate than the original   $\bar{q}_{n,{\btheta}^*}(\btheta)$  in approximating the target posterior density ${\pi}_{n}(\btheta)$, irrespectively of the chosen  $\bar{q}_{n,{\btheta}^*}(\btheta)$, its symmetry point  ${\btheta}^*$, and the sample size $n$. Notice also that \eqref{skew:sym:equiv:sym} is not only an intermediate step to obtain \eqref{skew:sym:non:asymp}, but it is of direct practical interest. Indeed, among the approximating densities $\bar{q}_{n,{\btheta}^*}(\btheta)$   with the same symmetry point  ${\btheta}^*$, it suggests to prioritize those that give a more accurate approximation of the symmetrized posterior $\bar{\pi}_{n,{\btheta}^*}(\btheta)$  in~\eqref{eq1}, rather than the original ${\pi}_{n}(\btheta)$. Therefore, although this objective goes beyond our original scope of perturbing an already-available $\bar{q}_{n,{\btheta}^*}(\btheta)$,  equations \eqref{skew:sym:equiv:sym} and \eqref{skew:sym:non:asymp} stimulate the development of novel symmetric approximations explicitly targeting $\bar{\pi}_{n,{\btheta}^*}(\btheta)$, rather than the original posterior ${\pi}_{n}(\btheta)$. In fact, as a consequence of \eqref{skew:sym:equiv:sym}--\eqref{skew:sym:non:asymp}, the  skewed perturbation of these approximations under the proposed strategy can yield an increasingly accurate characterization of ${\pi}_{n}(\btheta)$.

The results in Theorem~\ref{teo_1} follow from the specific form of the skewing factor $w_{n,{\btheta}^*}(\btheta)$ defined in \eqref{eq_w}. Theorem~\ref{teo_2} proves the optimality of such a factor.

\begin{jrsstheorem}{Theorem 2}
\noindent (Optimality of the skewing factor).
Let ${\pi}_{n}(\btheta)$, $\bar{q}_{n,{\btheta}^*}(\btheta)$  and $q_{n,{\btheta}^*}(\btheta)$ be defined as in Theorem~\ref{teo_1}. Moreover, consider the alternative skew-symmetric perturbation $\tilde{q}_{n,{\btheta}^*}(\btheta)=2\bar{q}_{n,{\btheta}^*}(\btheta)\tilde{w}_{{\btheta}^*}(\btheta)$ of $\bar{q}_{n,{\btheta}^*}(\btheta)$, where $\tilde{w}_{{\btheta}^*}(\btheta)$ is a generic skewing factor such that $\tilde{w}_{{\btheta}^*}(\btheta) \in [0,1]$ and $\tilde{w}_{{\btheta}^*}(\btheta)=1-\tilde{w}_{{\btheta}^*}(2 {\btheta}^*-\btheta)$. Then, for any ${\btheta}^* \in \Theta$, sample size $n$, and skewing factor $\tilde{w}_{{\btheta}^*}(\btheta)$, it holds 
\begin{equation*}
\mathcal{D}[{\pi}_{n} \mid \mid q_{n,{\btheta}^*}]  \leq  \mathcal{D}[{\pi}_{n} \mid \mid \tilde{q}_{n,{\btheta}^*}], 
 \end{equation*}
where $\mathcal{D}$ is either the TV  distance, KL, reverse-KL or a generic $\alpha$-divergence.
\label{teo_2}
 \end{jrsstheorem}

According to Theorem~\ref{teo_2}, the skewing factor  $w_{n,{\btheta}^*}(\btheta)$ in \eqref{eq_w} is guaranteed to provide a perturbed version $q_{n,{\btheta}^*}(\btheta)$ of $\bar{q}_{n,{\btheta}^*}(\btheta)$ that is never less accurate in approximating the target posterior  ${\pi}_{n}(\btheta) $ when compared to any other skew-symmetric density  $\tilde{q}_{n,{\btheta}^*}(\btheta)$ with symmetric component $\bar{q}_{n,{\btheta}^*}(\btheta)$ and generic skewing factor $\tilde{w}_{{\btheta}^*}(\btheta)$. Notice that to ensure $\tilde{q}_{n,{\btheta}^*}(\btheta)=2\bar{q}_{n,{\btheta}^*}(\btheta)\tilde{w}_{{\btheta}^*}(\btheta)$ is a skew-symmetric density as in Definition~\ref{def_1_or} it suffices that the skewing factor satisfies $\tilde{w}_{{\btheta}^*}(\btheta) \in [0,1]$ and $\tilde{w}_{{\btheta}^*}(\btheta)=1-\tilde{w}_{{\btheta}^*}(2 {\btheta}^*-\btheta)$. Hence, in principle, there are infinitely many options to perturb the original symmetric approximation so that the resulting density falls within the skew-symmetric class. As discussed in Section~\ref{sec_2.a}, some interesting examples of skewing factors have been derived in  \citet[][]{azzalini2003distributions}, \citet{genton2005generalized} and \citet[][]{ma2004flexible} with a focus on generalizations of skew-normal and skew-elliptical densities, which belong to the skew-symmetric family. According to Theorem~\ref{teo_2}, all these options are suboptimal relative to $w_{n,{\btheta}^*}(\btheta)$. This is because, unlike for other  alternatives,  $w_{n,{\btheta}^*}(\btheta)$ exactly matches the skewing factor of the target posterior, when expressed in skew-symmetric form as in Proposition~\ref{prop1}.

Besides proving the optimality of the skewing factor $w_{n,{\btheta}^*}(\btheta)$, Theorem~\ref{teo_2}   allows also to formalize the proposed skew-symmetric approximation in Definition~\ref{def_1} as the solution of a well defined optimization problem. This result is stated in Corollary~\ref{cor_1}, and is useful to establish direct connections with state-of-the-art approximation strategies that arise from the optimization of specific divergences, such as, for example, VB \citep[e.g.][]{blei2017variational} and EP \citep[e.g.][]{vehtari2020expectation}. Moreover, as discussed in Remark~\ref{re2}, it provides the premises to further expand the scope of the novel perspective considered in this article. 

\begin{jrsscorollary}{Corollary 1}
\noindent Let ${\pi}_{n}(\btheta)$, $\bar{q}_{n,{\btheta}^*}(\btheta)$  and $q_{n,{\btheta}^*}(\btheta)$ be defined as in Theorem~\ref{teo_1}. Moreover, let 
\begin{equation*}
\mathcal{Q}= \{\tilde{q}_{n,{\btheta}^*}(\btheta): \tilde{q}_{n,{\btheta}^*}(\btheta)=2\bar{q}_{n,{\btheta}^*}(\btheta)\tilde{w}_{{\btheta}^*}(\btheta)\},
 \end{equation*}
denote the general family of skew-symmetric densities, which arise from the perturbation of $\bar{q}_{n,{\btheta}^*}$ through a generic skewing factor $\tilde{w}_{{\btheta}^*}(\btheta)\in [0,1]$ satisfying $\tilde{w}_{{\btheta}^*}(\btheta)=1-\tilde{w}_{{\btheta}^*}(2 {\btheta}^*-\btheta)$. Then,  for any symmetry point  ${\btheta}^* \in \Theta$ and sample size $n$, it holds
\begin{equation*}
\mbox{\normalfont min}_{\tilde{q}_{n,{\btheta}^*} \in \mathcal{Q}}\mathcal{D}[{\pi}_{n} \mid \mid \tilde{q}_{n,{\btheta}^*}]=\mathcal{D}[{\pi}_{n}\mid \mid q_{n,{\btheta}^*}],
 \end{equation*}
where $\mathcal{D}$ is either the TV  distance, KL, reverse-KL or a generic $\alpha$-divergence, and, recalling  \eqref{skew:sym:equiv:sym}, $\mathcal{D}[{\pi}_{n} \mid \mid q_{n,{\btheta}^*}] = \mathcal{D}[\bar{\pi}_{n,{\btheta}^*} \mid \mid \bar{q}_{n,{\btheta}^*}]$.
\label{cor_1}
 \end{jrsscorollary}
 
 Although the proposed skew-symmetric approximation has not been derived in Section~\ref{sec_3} as the solution of an optimization problem, Corollary~\ref{cor_1} clarifies that, in fact,  $q_{n,{\btheta}^*}(\btheta)$ can be formalized also under such a perspective. In particular,  the skew-symmetric density $q_{n,{\btheta}^*}(\btheta)$ in Definition~\ref{def_1} actually coincides with the solution of the constrained minimization for a suitable divergence $\mathcal{D}$ between the target posterior ${\pi}_{n}(\btheta)$ and a given approximating density $\tilde{q}_{n,{\btheta}^*}(\btheta)$ within the family $\mathcal{Q}$. Such a family comprises all the skew-symmetric densities having symmetric component fixed at the already-available approximating density $\bar{q}_{n,{\btheta}^*}(\btheta)$. 
 
 The above interpretation allows us to establish direct connections with the optimization-based perspectives of VB \citep[e.g.][]{blei2017variational} and EP \citep[e.g.][]{vehtari2020expectation} solutions. However, unlike for these  strategies, Corollary~\ref{cor_1} holds under a broader class of divergences, rather than a specific one, and, when compared to routinely implemented VB and EP schemes yielding symmetric approximations, it considers an expanded family which ensures improvements in accuracy. Note that, consistent with the focus of this article, $\bar{q}_{n,{\btheta}^*}(\btheta)$  and, as a consequence, ${\btheta}^*$ are known and fixed in Corollary~\ref{cor_1}. Hence, the only quantity to be derived is the skewing factor. Crucially, as clarified in Theorem~\ref{teo_2}, the solution $w_{n,{\btheta}^*}(\btheta)$ of this minimization with respect to $\tilde{w}_{{\btheta}^*}(\btheta)$ does not require optimization of additional parameters beyond the already-available ${\btheta}^*$. Although extending the optimization problem in  Corollary~\ref{cor_1} to the case in which also  $\bar{q}_{n,{\btheta}^*}(\btheta)$   is unknown goes beyond our scope, as clarified in Remark~\ref{re2}, this direction can be of substantial interest to further improve the accuracy of $q_{n,{\btheta}^*}(\btheta)$, and our results open several avenues to stimulate future advancements along these lines. 
 
  \begin{jrssremark}{Remark 4}
\noindent Recalling Sections~\ref{sec_1} and \ref{sec_3}, the overarching focus of this article is to improve the accuracy of state-of-the-art symmetric approximations of posterior distributions via a broadly applicable perturbation scheme, which can be derived at no additional optimization costs and applied directly to the output  $\bar{q}_{n,{\btheta}^*}(\btheta)$  of standard implementations. To this end, $\bar{q}_{n,{\btheta}^*}(\btheta)$  is kept fixed and known within our derivations. However, although the optimization of such a symmetric component goes beyond the scope of our contribution, combining the results in Theorem~\ref{teo_1} and Corollary~\ref{cor_1} with the skew-symmetric representation of posterior densities in Proposition~\ref{prop1},  opens  promising directions to further improve the approximation accuracy via the additional optimization  of the symmetric component. In particular, as clarified in Corollary~\ref{cor_1}, minimizing $\mathcal{D}[{\pi}_{n} \mid \mid \tilde{q}_{n,{\btheta}^*}]$ also with respect to the symmetric component in $\tilde{q}_{n,{\btheta}^*}(\btheta)$, simply requires to find the closest symmetric density to the symmetrized posterior in \eqref{eq1}, and then perturb such a density with the already-derived optimal skewing factor $w_{n,{\btheta}^*}(\btheta)$. This is expected to further improve accuracy relative to perturbations of currently implemented symmetric approximations that target the actual posterior instead of its symmetrized version. In fact, to our knowledge, such a different target has never been considered before and, hence, our results can open unexplored avenues to derive improved classes of tractable deterministic approximations, along with novel computational methods to obtain these approximations. When specializing $\mathcal{D}$ to the  KL  minimized under VB a promising direction could be to solve such an optimization problem via automatic differentiation variational inference schemes  \citep{kucukelbir2017automatic}.
\label{re2}
\end{jrssremark}

Section~\ref{sec_34}  quantifies the aforementioned accuracy gains of $q_{n,{\btheta}^*}(\btheta)$ in asymptotic settings.

%###############################################################################
\subsection{\large 3.2 Asymptotic properties} \label{sec_34}
Theorems \ref{teo_1} and \ref{teo_2} provide  theoretical guarantees in finite samples for the improved approximation accuracy of the newly developed skew-symmetric solution $q_{n,{\btheta}^*}(\btheta)$, compared to its symmetric  counterpart, i.e.\ $\bar{q}_{n,{\btheta}^*}(\btheta)$. However, these results do not quantify the magnitude of such improvements. Below,  we address this important point from an asymptotic perspective for $n \to \infty$. This focus further clarifies how the overall quality of the proposed skew-symmetric approximation crucially  depends on the one achieved by its symmetric counterpart  $\bar{q}_{n,{\btheta}^*}(\btheta)$ in approximating the symmetrized posterior $\bar{\pi}_{n,{\btheta}^*}(\btheta)$. 

To address the above objectives, we study the skewed perturbations of two particular symmetric approximations, both centered at  the maximum a posteriori  $\btheta_{\textsc{map}}= \arg\max_{\btheta \in \Theta}\pi(\btheta) L(\btheta; \by_{1:n})$ (i.e.\ ${\btheta}^*= \btheta_{\textsc{map}}$), and approximating the target posterior density ${\pi}_{n}(\btheta)$ with the same rate. Since the symmetry points coincide, the perturbations of these two symmetric approximations share the same  skewing factor $w_{n, \btheta_{\textsc{map}}}(\btheta)$, as a direct consequence of \eqref{eq_w}. Nonetheless,  as clarified in the following, the resulting skew-symmetric perturbations achieve different rates in approximating ${\pi}_{n}(\btheta)$.  These rates coincide with those obtained by the symmetric counterparts in approximating the symmetrized posterior $\bar{\pi}_{n,{\btheta}_{\textsc{map}}}(\btheta)$, which can be substantially different from those achieved when the target is the actual posterior $\pi_n(\btheta)$. Note that, unlike for the finite sample properties in Section~\ref{sec_33}, which hold broadly for any skew-symmetric approximation, the asymptotic analysis we consider below focuses on the sub-class of skewed perturbations of Gaussians or higher-order extensions of these Gaussians, both centered at the posterior mode, i.e.\ $\btheta_{\textsc{map}}$. Although generalizations of such an asymptotic theory to the general class of skew-symmetric approximations can be envisioned, the focus on this sub-class facilitates theoretical derivations and allows to leverage available results on the rates of Gaussian approximations. In addition,  these Gaussian approximations and the corresponding higher-order versions are directly related to ubiquitous Laplace type methods, and hence,  it is natural to deepen the asymptotic analysis of the skew-symmetric approximations arising from the perturbation of such symmetric densities. Note that, although we consider the MAP as symmetry point, our theory relies on assumptions which allow to replace $\btheta_{\textsc{map}}$ with  any efficient estimator of $\btheta$.

To derive the aforementioned asymptotic results, let  $\ell_n(\btheta) = \ell(\btheta; \by_{1:n})=\log L(\btheta; \by_{1:n})$ and $\log \pi(\btheta)$ be, respectively, the log-likelihood function and log-prior density,  evaluated at $\btheta \in \Theta  \subseteq \mathbb{R}^{d}$. Moreover, denote with
\begin{equation*}
\ell^{(k)}_{n,\btheta} = (\partial^{\otimes k}/\partial^{\otimes k} \btheta )\ell_n(\btheta), \quad \mbox{and}  \quad \log \pi_\btheta^{(k)}= (\partial^{\otimes k}/\partial^{\otimes k} \btheta )\log \pi(\btheta),
\end{equation*}
the $d^k$-dimensional arrays containing the associated $k$-th order partial derivatives at $\btheta$. Under these settings, the first symmetric density whose skewed perturbation is studied in asymptotic regimes is the Gaussian approximation from the Laplace method. Namely,
\begin{equation}\label{glap}
	\bar{q}_{n,\btheta_{\textsc{map}} ,(1)}(\btheta)  = \phi_d(\btheta; \btheta_{\textsc{map}} , \bJ_{\btheta_{\textsc{map}}}^{-1} ),
\end{equation}
 where $\btheta_{\textsc{map}} = \arg\max_{\btheta \in \Theta} \pi(\btheta)L(\btheta; \by_{1:n})$ and  $\bJ_{\btheta_{\textsc{map}}} = - ( \ell^{(2)}_{n,\btheta_{\textsc{map}}} + \log \pi_{\btheta_{\textsc{map}}}^{(2)}   )$. 
 
 The second is, instead, a novel higher-order extension of the Gaussian in \eqref{glap}, which relates to the family of semi-nonparametric distributions \citep{gallant1987semi} and has density
\begin{equation}\label{SNP:modal:approx}
	\bar{q}_{n,\btheta_{\textsc{map}},(2)}(\btheta) = 	 \frac{\phi_d(\btheta; \btheta_{\textsc{map}} , \bJ_{ \btheta_{\textsc{map}}}^{-1} ) f(\btheta - \btheta_{\textsc{map}})}{ \int  \phi_d(\btheta; \btheta_{\textsc{map}} , \bJ_{ \btheta_{\textsc{map}}}^{-1} ) f(\btheta - \btheta_{\textsc{map}})  \mbox{d} \btheta},
\end{equation}
where $ f(\btheta - \btheta_{\textsc{map}})$ is a non-negative polynomial obtained from a  fourth-order expansion of the symmetrized log-posterior at $\btheta_{\textsc{map}}$ that yields
	\begin{equation*}
		\begin{split}
			&f(\btheta - \btheta_{\textsc{map}}) = \\
			& \quad  1 + \frac{1}{24} \langle \ell^{(4)}_{n,\btheta_{\textsc{map}}}, (\btheta - \btheta_{\textsc{map}})^{\otimes 4} \rangle + \frac{1}{2} \left[\frac{1}{24} \langle \ell^{(4)}_{n,\btheta_{\textsc{map}}}, (\btheta - \btheta_{\textsc{map}})^{\otimes 4} \rangle \right]^2 
			+ \frac{1}{2} \left[\frac{1}{6} \langle \ell^{(3)}_{n,\btheta_{\textsc{map}}}, (\btheta - \btheta_{\textsc{map}})^{\otimes 3}\rangle \right]^2.
		\end{split}
	\end{equation*} 
In the above expression, the generic quantity \smash{$\langle \ell^{(k)}_{n,\btheta_{\textsc{map}}}, (\btheta - \btheta_{\textsc{map}})^{\otimes k}\rangle$} denotes the polynomial term associated to the $k$-th element of the Taylor expansion for $\ell_n(\btheta)$ at $\btheta_{\textsc{map}}$. Notice that, in view of Lemma B.2 within Appendix~B of the Supplementary Material, $f(\btheta - \btheta_{\textsc{map}})$ is always non-negative, and hence,~\eqref{SNP:modal:approx} is a proper density function symmetric about $\btheta_{\textsc{map}}$. Although \eqref{SNP:modal:approx} is of less direct applicability than~\eqref{glap}, such a higher-order approximating density can capture the behavior of the symmetrized posterior more flexibly than the Gaussian in \eqref{glap}. Therefore, it provides an interesting theoretical alternative to \eqref{glap} for quantifying the gains of the proposed skew-symmetric family of approximations when applied to symmetric densities achieving different accuracies in capturing the behavior of the symmetrized posterior. Since \eqref{glap} and \eqref{SNP:modal:approx} share the same symmetry point $\btheta_{\textsc{map}}$, the resulting skew-symmetric perturbations 
\begin{equation*}
{q}_{n,\btheta_{\textsc{map}},(1)}(\btheta) = 2\bar{q}_{n,\btheta_{\textsc{map}},(1)}(\btheta)w_{n,\btheta_{\textsc{map}}}(\btheta), \quad \mbox{and} \quad {q}_{n,\btheta_{\textsc{map}},(2)}(\btheta)= 2\bar{q}_{n,\btheta_{\textsc{map}},(2)}(\btheta) w_{n,\btheta_{\textsc{map}}}(\btheta),
\end{equation*}
have the same skewing factor $w_{n,\btheta_{\textsc{map}}}(\btheta)$. 

Theorem \ref{teo_3} below formalizes the above discussion. When the dimension $d$ is fixed, this theorem states that several key divergences (including the TV distance, KL, reverse-KL and $\alpha$-divergences) between the target posterior  density and the skew-symmetric perturbation ${q}_{n,\btheta_{\textsc{map}},(1)}(\btheta) $ of~\eqref{glap} converge to zero in probability with a rate $1/n$, up to a poly-log term, which substantially improves the $1/\sqrt{n}$ rate of the unperturbed Gaussian approximation in~\eqref{glap}. This accuracy gain is further refined by ${q}_{n,\btheta_{\textsc{map}},(2)}(\btheta) $ which achieves a \smash{$1/n^2$} rate, again up to a poly-log term. Such a latter result is even more remarkable, provided that~\eqref{SNP:modal:approx} has the same rate of~\eqref{glap} in approximating the target posterior ${\pi}_{n}(\btheta)$. This is due to the fact that both  \eqref{glap} and \eqref{SNP:modal:approx} are not able to capture the skewness of the posterior distribution. However, when the target is the symmetrized posterior $\bar{\pi}_{n,{\btheta}_{\textsc{map}}}(\btheta)$ (i.e., the skewness is removed from the target), the higher-order extension~\eqref{SNP:modal:approx}  of~\eqref{glap}  improves the rates of such a latter Gaussian approximation. Combining this result with Theorem \ref{teo_1}, the same level of accuracy is maintained by  ${q}_{n,\btheta_{\textsc{map}},(2)}(\btheta) $ in approximating the target posterior density ${\pi}_{n}(\btheta)$.  This explains the improved rates of ${q}_{n,\btheta_{\textsc{map}},(2)}(\btheta) $ relative to ${q}_{n,\btheta_{\textsc{map}},(1)}(\btheta) $, and provides a result with important methodological consequences. In particular, it suggests to target the symmetrized posterior rather than the original one with state-of-the-art symmetric approximations, and then perturb these approximations with the proposed skewing factor.

The proof of Theorem \ref{teo_3} is provided in Appendix~A of the Supplementary Material and relies on a number of  regularity conditions listed below. These conditions are similar to those recently considered by \citet{durante2023skewed} for deriving a skewed extension of the  Bernstein--von Mises theorem \citep[e.g.][]{van2000asymptotic}. Recalling related discussions in \citet{durante2023skewed}, we shall emphasize that Assumptions \ref{cond:uni}--\ref{assump:LBtail} provide natural extensions of standard conditions in asymptotic studies of this type, and hold under several, commonly adopted, regular statistical models, such as generalized linear models. 

\begin{jrssassumption}{Assumption 1}
\noindent For every sample size $n \in \mathbb{N}$, the data $\by_{1:n}$ are realizations from a sequence of random variables ${\bf Y}_{1:n}$ having true underlying distribution $P_0^n$ which may not be necessarily included within the parametric family $\{ P_{\btheta}^n, \btheta \in \Theta\}$ defining the assumed statistical model with log-likelihood $\ell_n(\btheta)$. Furthermore, the KL projection $P^n_{\btheta_0}$ of $ P_0^n$ into $\{ P_{\btheta}^n, \btheta \in \Theta\},$ is unique.  
\label{cond:uni}
\end{jrssassumption}

\begin{jrssassumption}{Assumption 2}
\noindent The log-prior density $ \log \pi(\btheta)$ is four times continuously differentiable inside a neighborhood of $\btheta_0$, and $0 < \pi(\btheta_0) < \infty$. 	
\label{cond:3} 
\end{jrssassumption}

\begin{jrssassumption}{Assumption 3}
\noindent For every $M_n \to \infty$ there exists a positive constant $c_1$ such that 
$$\lim\nolimits_{n \to \infty } P_0^{n} \{ \sup\nolimits_{\|\btheta - \btheta_0\|> M_n \sqrt{d}/\sqrt{n}}\{ \ell_n(\btheta) - \ell_n(\btheta_0)  \}/n < - c_1 M_n^2d/n \} = 1,$$
where $ \ell_n(\btheta_0)$ is the log-likelihood evaluated at $\btheta_0$.
\label{cond:4}
\end{jrssassumption}

\begin{jrssassumption}{Assumption 4}
\noindent The maximum a posteriori estimator ${\btheta}_{\textsc{map}}=\arg\max_{\btheta \in \Theta} \{\ell_n(\btheta)+\log \pi(\btheta)\}$, satisfies  $\E_{0}^n\| {\btheta}_{\textsc{map}} - \btheta_0 \|^2 = O( d/n)$.
 \label{cond:m1}
\end{jrssassumption}

\begin{jrssassumption}{Assumption 5}
\noindent There exist two positive constants $\bar \eta_1$ and $\bar \eta_2$ such that the event $$\tilde A_{n,0}= \{ \lambda_{\textsc{min}}( \bJ_{\btheta_{\textsc{map}}}/n)> \bar \eta_1 \} \cap\{ \lambda_{\textsc{max}}(\bJ_{\btheta_{\textsc{map}}}/n)< \bar \eta_2 \}$$ holds with probability $P_0^n(\tilde A_{n,0}) = 1 - o(1)$. 
	Moreover, there exist two positive constants $\delta>0$ and $L>0$ such that the inequalities  $$
 \|\ell_n^{(3)}(\btheta)/n  \| < L, \quad  \| \ell_n^{(4)}(\btheta)/n \| < L, \quad \|\log \pi^{(2)}(\btheta) \|  < L,$$
 hold uniformly in $ B_{\delta}(\btheta_{\textsc{map}}) = \{ \btheta  \in \Theta \,:\, \| \btheta_{\textsc{map}} - \btheta \| < \delta \}$, with $P_{0}^{n}$-probability tending to one, where $\|{\cdot}\|$ denotes the spectral norm. When  ${q}_{n,\btheta_{\textsc{map}},(2)}(\btheta)$ is considered,  for the same $\delta>0$ and $L>0$ as above, also the inequalities $$\|\ell_n^{(5)}(\btheta)/n \| < L,  \ \  \| \ell_n^{(6)}(\btheta)/n  \| < L,  \ \  \|\log \pi^{(3)}(\btheta)  \| < L, \ \	\|\log \pi^{(4)}(\btheta) \|  < L,$$ hold uniformly in $ B_{\delta}(\btheta_{\textsc{map}})$, with $P_{0}^{n}$-probability tending to one.
 \label{cond:m3}	
\end{jrssassumption}

\begin{jrssassumption}{Assumption 6}
\noindent Let  ${\bh}_{\textsc{map}}=\sqrt{n}(\btheta-{\btheta}_{\textsc{map}})$. Then, there  exist  positive constants $c$ and $M$, such that $|\log \bar\pi_{n,{\btheta}_{\textsc{map}}}({\bh}_{\textsc{map}})|\leq c  \|{\bh}_{\textsc{map}}\|^M\vee1$.
\label{assump:LBtail}	
\end{jrssassumption}

Assumptions~\ref{cond:uni}--\ref{cond:4} are needed to guarantee that the assumed model is sufficiently regular and that the induced posterior distribution $\pi_n(\btheta)$ concentrates  asymptotically fast enough around $\btheta_0$. Similarly, for $d$ fixed, Assumption~\ref{cond:m1} provides the parametric $1/\sqrt{n}$ concentration rate of the  posterior mode ${\btheta}_{\textsc{map}}$ around $\btheta_0$, required to control the likelihood around the random location ${\btheta}_{\textsc{map}}$. Assumption~\ref{cond:m3} presents, instead, requirements on the regularity of the prior and the likelihood, depending on the approximation procedure considered. The second, more complex, approximation ${q}_{n,\btheta_{\textsc{map}},(2)}(\btheta)$  requires a stronger control, but the corresponding rates are also faster. These conditions allow us to control the accuracy of the Taylor expansion for the log-posterior, up to the required order, and for its approximation induced by ${q}_{n,\btheta_{\textsc{map}},(1)}(\btheta)$ or ${q}_{n,\btheta_{\textsc{map}},(2)}(\btheta)$.  Finally, Assumption \ref{assump:LBtail} controls the tail behaviour of the symmetrized posterior to ensure finite cross entropy between the skewed perturbations of~\eqref{glap} (or~\eqref{SNP:modal:approx}) and ${\pi}_n( \btheta )$. While this condition could be possibly relaxed, we conjecture that explicit control of the tails is necessary to obtain the results for the KL and reverse-KL  divergences in Theorem~\ref{teo_3}.

The combination of Assumptions \ref{cond:uni}--\ref{assump:LBtail} described above allows us to state the following theorem regarding the asymptotic accuracy of ${q}_{n,\btheta_{\textsc{map}},(1)}(\btheta)$ and ${q}_{n,\btheta_{\textsc{map}},(2)}(\btheta)$ in approximating the target posterior $\pi_n(\btheta)$.  

\begin{jrsstheorem}{Theorem 3}
\noindent (Asymptotic accuracy).
Under Assumptions \ref{cond:uni}--\ref{cond:m3}, it holds
	\begin{equation} \label{tot:variatindist:map:sym}
		\mathcal{D}[\pi_n\mid \mid  {q}_{n,\btheta_{\textsc{map}},(1)} ]  = O_{P_{0}^n}\big(M_n^{c_{3}}d^3/n\big), 
	\end{equation}
and
	\begin{equation} \label{tot:variatindist:map:sym2}
		\mathcal{D} [\pi_n \mid \mid   {q}_{n,\btheta_{\textsc{map}},(2)} ]  = O_{P_{0}^n}\big(M_n^{c_{4}}d^6/n^2\big), 
	\end{equation}
	where $M_n = \sqrt{c_0 \log n}$, the quantities $c_0, c_{3},  c_{4} >0 $ are fixed positive constants not depending on $n$ and $d$, and $\mathcal{D}$ denotes either the TV  distance or any generic $\alpha$-divergence with $\alpha \in (0,1)$. If Assumption~\ref{assump:LBtail} holds (in addition to \ref{cond:uni}--\ref{cond:m3}), then \eqref{tot:variatindist:map:sym}--\eqref{tot:variatindist:map:sym2} are valid also when $\mathcal{D}$ is the KL or reverse-KL  divergence.
 \label{teo_3}
\end{jrsstheorem}

  \begin{jrssremark}{Remark 5}
\noindent Under the same assumptions, it can be also shown that the convergence rate of the symmetric components \eqref{glap} and \eqref{SNP:modal:approx} to the target posterior density $\pi_n( \btheta )$ is \smash{$O_{P_0^n}(M_n^{c_{5}}d^{3/2}/\sqrt{n})$} for some $c_{5}>0$. This result suggests that the difference in  \eqref{tot:variatindist:map:sym}--\eqref{tot:variatindist:map:sym2} are due to how the symmetric component approximates $\bar{\pi}_{n,\btheta_{\textsc{map}}}(\btheta)$.
\label{sym_rate_sqrt}
\end{jrssremark}

  \begin{jrssremark}{Remark 6}
\noindent All the rates presented in Theorem~\ref{teo_3} and Remark~\ref{sym_rate_sqrt} rely on bounds which are guaranteed to vanish also when the dimension $d$ grows with the sample size $n$, as long as \smash{$d=o(n^{1/3})$}, up to a poly-log term. This condition  is common in high-dimensional studies of Gaussian approximations \citep[][]{panov2015finite,spokoiny2023inexact,spokoiny2025accuracy}. However, unlike these asymptotic studies, the rates presented in Theorem~\ref{teo_3}  vanish with $n$ (or even \smash{$n^2$}), up to a poly-log term, instead of \smash{$\sqrt{n}$}, for any fixed $d$. Finally, let us emphasize that the constants $c_0, c_{3}$ and $c_{4}$ in Theorem~\ref{teo_3} arise from specific technicalities in the proof that can be found in the Supplementary Material. For what concerns the interpretation of the rates in  Theorem~\ref{teo_3}, the only key requirement is that   $c_0, c_{3}$ and $c_{4}$  are fixed positive constants not depending on $n$ and $d$.
\label{rem_d_n}
\end{jrssremark}

Note that ${q}_{n,\btheta_{\textsc{map}},(1)}$ is closely related to the skew-modal approximation recently introduced in \citet{durante2023skewed}. These two approximations share a similar symmetric component and the same order of convergence to the target posterior, but have different  skewing factor. In \citet{durante2023skewed} such a quantity is derived using asymptotic arguments that involve the evaluation of third-order log-likelihood derivatives at $\btheta_{\textsc{map}}$. Conversely, the skewing factor $w_{n,\btheta_{\textsc{map}}}(\btheta)$ introduced in this article depends only on the  un-normalized posterior density, and hence, is simpler to implement. Moreover, while the focus of  \citet{durante2023skewed} is on perturbing Gaussian approximations from the Laplace method to refine the classical Bernstein--von Mises theorem, the proposed skew-symmetric approximations broadly apply to any symmetric density, beyond Gaussians centered at the MAP. Besides its practical benefits illustrated in Sections~\ref{sec_4} and \ref{sec_4_app} below, this generality yields also theoretical gains which are evident in \eqref{tot:variatindist:map:sym2}. This latter result clarifies that the theory in \citet{durante2023skewed} can be further expanded to higher-order approximations achieving  faster convergence rates, while remaining within the skew-symmetric family.

%###############################################################################
%###############################################################################
\section{\large 4. Simulation studies} \label{sec_4}
Before assessing the skew-symmetric approximations in real-data applications, let us first consider two simulation studies quantifying to what extent the accuracy gains and improved rates derived theoretically in Sections~\ref{sec_33} and \ref{sec_34} find empirical evidence also in practice, even beyond the settings explored by our theory. To this end, we compare the accuracies achieved by three popular Gaussian approximations from the  Laplace method \citep[][Chapter 13]{gelman1995bayesian}, black-box VB \citep[][]{kucukelbir2017automatic,ranganath2014black} and EP \citep[][]{minka2013expectation,vehtari2020expectation}, with those obtained under the corresponding skew-symmetric counterparts. These gains are quantified through different accuracy measures in two simulations studies. The first (see Section~\ref{sec_simu_uni}) aims at providing empirical evidence for the asymptotic rates presented in Theorem~\ref{teo_3} (see also Remark~\ref{sym_rate_sqrt}), under a parametric setting that meets the conditions underlying the theory within Section~\ref{sec_34}. In this first study, we consider a simple  one-dimensional scenario, which allows us to evaluate with precision the different divergences in Theorem~\ref{teo_3}  via numerical integration methods. This provides a reliable quantification of the rates, which is not affected by Monte Carlo error. Large $d$ settings are instead addressed in the second simulation study (see Section~\ref{sec_simu_high})  whose aim is to explore the accuracy improvements of the skew-symmetric approximations beyond the settings and quantities studied theoretically in Sections~\ref{sec_33} and \ref{sec_34}.  In particular, the focus of Section~\ref{sec_simu_high} is on assessing the accuracy in approximating the posterior marginals under high-dimensional regimes where $d$ grows with $n$ at a rate which does not meet the condition \smash{$d=o( n^{1/3})$} discussed in Remark~\ref{rem_d_n}.

%###############################################################################
\subsection{\large 4.1 One-dimensional Poisson model} \label{sec_simu_uni}
The upper bound \eqref{tot:variatindist:map:sym} presented in Theorem~\ref{teo_3} suggests that, for a sufficiently large $n$ and fixed $d$, the log-divergence between the target posterior and the proposed skew-symmetric perturbation of suitable Gaussian approximations should decrease as a linear function of $\log n$ with slope that is $-1$ or lower. According to the discussion in Remark~\ref{sym_rate_sqrt}, this slope should be, instead, $-0.5$, or lower, for the unperturbed Gaussian counterpart.

To assess whether the above results find empirical evidence in practice, we simulate i.i.d.\ data $\by_{1:n}$, for growing sample size $n$ from $n=15$ to $n=145$ with step size $10$, from a Poisson distribution having rate $\exp(\theta_0)=1$. Leveraging these samples of growing size we perform Bayesian inference by assuming a one-dimensional Poisson model with rate $\exp(\theta)$, where $\theta \in \mathbb{R}$ is assigned a Student-$t$ prior with one degree of freedom. Such a prior choice coincides with a hierarchical Bayesian formulation assuming a zero mean Gaussian prior for $\theta$ combined with a $\chi_1^2$ hyperprior for its precision parameter, thereby providing a realistic elicitation that is aligned with the common practice of considering higher-level priors for the hyperparameters to mitigate sensitivity issues. These settings yield a sequence of posteriors $\pi_n(\theta)$, $n=15,25,35, \ldots, 135,145$, for $\theta$, which we approximate under the three Gaussian densities resulting from the Laplace method \citep[][Chapter 13]{gelman1995bayesian}, black-box VB \citep[][]{kucukelbir2017automatic,ranganath2014black} and EP \citep[][]{minka2013expectation,vehtari2020expectation}, respectively, along with the corresponding skew-symmetric perturbations presented in Section~\ref{sec_3}. As anticipated in Sections~\ref{sec_1} and \ref{sec_3}, all these approximations can be readily derived leveraging standard softwares. In particular, the three Gaussians alternatives can be obtained as a  direct output of \texttt{rstan}  or simple \texttt{R} coding, whereas the skew-symmetric counterparts only require the additional calculation of the closed-form skewing factor, without further optimization.

\begin{table}[t]
\caption{\footnotesize{Empirical comparison of the rates achieved by state-of-the-art Gaussian approximations from the  Laplace method, black-box VB and EP, with those obtained under the corresponding skew-symmetric perturbations. For three routinely employed divergences $\mathcal{D}$ (i.e.\ TV,  KL and reverse-KL) the table displays the slope of the fitted linear regression between $\log \mathcal{D}(\pi_n \mid \mid q_n)$ and $\log n$, where $q_n$ is either $\bar{q}_{n,{\theta}^*}$ (i.e.\ Gaussian) or  ${q}_{n,{\theta}^*}$ (i.e.\ Skew-symmetric), and $n$ is defined on a grid from $15$ to $145$ with step size $10$. The results are averaged across 50 replicated experiments (standard error within brackets). Bold values indicate the best performance for each pair of Gaussian and skew-symmetric approximation.}}
\label{tab_0_pois}
\centering
\begin{tabular}{l|ccc}
	\hline
	& TV & KL & reverse-KL   \\ 
	\hline
	Laplace  (Gaussian)& $-0.48 \ (0.01)$  & $-0.93 \ (0.02)$ & $-0.97\  (0.02)$   \\ 
	Laplace  (Skew-symmetric) &  $\boldsymbol{-1.04} \ (0.02)$ & $\boldsymbol{-1.80}  \ (0.08)$ & $\boldsymbol{-3.11} \ (0.26)$  \\   \hline
	black-box VB  (Gaussian)& $ -0.48 \ (0.01)$   & $ -0.95  \ (0.02)$  & $ -0.98   \ (0.02)$   \\ 
	black-box VB  (Skew-symmetric)\qquad \qquad& $ \boldsymbol{-1.05}  \ (0.02)$  & $ \boldsymbol{-1.73}  \ (0.12)$  & $ \boldsymbol{-3.18}  \ (0.29)$   \\   \hline
		EP  (Gaussian)&   $-0.47   \  (0.01) $ &   $-0.93   \  (0.02) $ &   $-0.99   \  (0.02) $  \\ 
	EP  (Skew-symmetric)& $ \boldsymbol{-0.99}   \  (0.02) $ &  $\boldsymbol{-1.76}   \  (0.11) $ &   $\boldsymbol{-3.47}   \  (0.26) $  \\   \hline
\end{tabular}
\end{table}

Table~\ref{tab_0_pois} quantifies the rates achieved by the aforementioned approximations under routinely employed divergences among those considered in Theorem~\ref{teo_3}, i.e.\ TV,  KL and reverse-KL.  Consistent with the motivation underlying this simulation study, we display, in particular, the slope of the fitted linear regression between $\log \mathcal{D}(\pi_n \mid \mid q_n)$ and $\log n$, for $n=15,25,35, \ldots, 135,145$, where $\mathcal{D}$ is the  TV,  KL or reverse-KL, and $q_n$ is either one of the Gaussian approximations under analysis (i.e.\ $q_n=\bar{q}_{n,{\theta}^*}$) or its skew-symmetric counterpart (i.e.\ $q_n={q}_{n,{\theta}^*}$). In Table~\ref{tab_0_pois}, these slopes are averaged across 50 replicated experiments and the associated standard errors are reported within brackets. Consistent with Theorem~\ref{teo_3} and Remark~\ref{sym_rate_sqrt}, the proposed skew-symmetric approximations display slopes close to $-1$ or below it, across all divergences and approximation methods, whereas those of the unperturbed Gaussian counterparts are approximately $-0.5$ or lower. Moreover, the improvements in rates achieved by the proposed perturbation strategy over its symmetric alternative are systematic and by at least a multiplicative factor of $2$. These results provide  empirical evidence on the fact that the finite sample accuracy gains proved in Theorem~\ref{teo_1} and the asymptotic rates derived in Theorem~\ref{teo_3} are visible also in practice. Interestingly, the empirical results obtained for the KL and reverse-KL further suggest that sharper upper bounds could be derived in Theorem~\ref{teo_3}  for these two divergences, thereby stimulating future research along these lines. Section~\ref{sec_simu_high} quantifies to what extent similar gains can be preserved in more challenging settings and for quantities of interest beyond those studied by our theory in Section~\ref{sec_theo}.

%###############################################################################
\subsection{\large 4.2 High-dimensional Poisson regression} \label{sec_simu_high}
Let us now consider a high-dimensional setting where $d =  6  \lfloor \sqrt{n/2} \rfloor $, and the focus is on assessing the accuracy improvements in approximating posterior marginals for each $\theta_j$, $j=1, \ldots, d$. This setting is substantially more challenging than the one considered in  Section~\ref{sec_simu_uni} for two  reasons. First, it does not meet the \smash{$d =o(n^{1/3})$} condition that is required to obtain vanishing rates under Theorem~\ref{teo_3} (see also Remark~\ref{rem_d_n}). Second, the finite sample accuracy gains proved in Theorem~\ref{teo_1} focus on the approximation of the joint posterior density $\pi_n(\btheta)$, rather than its marginals $\pi_n(\theta_j)$, for $j=1, \ldots, d$. As such, it is of interest to assess whether the theoretical results in Section~\ref{sec_theo} and the empirical gains in Section~\ref{sec_simu_uni} find some numerical evidence also in these more challenging settings. In fact, as discussed in Section~\ref{sec_3}, the skewness-inducing correction we propose relies on a global notion of central symmetry. Therefore, it is unclear whether re-distributing the density with respect to such a notion translates into empirical gains in characterizing local asymmetries within the marginals, which are often of interest in practice \citep[e.g.][]{Chopin_2017,inla_paper,tierney1986accurate}.

To address this goal, we simulate data $\by_{1:n}$ for growing sample size $n=200,400,600,800,1000,$ and dimension \smash{$d= 6  \lfloor \sqrt{n/2} \rfloor$}, from a Poisson regression model with the generic response variable $y_i$ having rate $\exp(\bx_i^\intercal \btheta_0)$, $i=1, \ldots, n$. In this simulation, $\bx_i$ denotes a $d$-dimensional vector with unbalanced binary entries drawn independently from Bernoulli variables having low probabilities generated by a uniform in $[0.02,0.1]$, while $\btheta_0$ has entries simulated from independent uniforms in $[-0.5,0.5]$, and further rescaled to control the variance of the linear predictor as $d$ grows with $n$. Extending the formulation considered in Section~\ref{sec_simu_uni}, we model these data under a Bayesian  Poisson regression with rates $\exp(\bx_i^\intercal \btheta)$, for  $i=1, \ldots, n$, and assume independent Student-$t$ priors with one degree of freedom for each $\theta_j$, $j=1, \ldots, d$.

\begin{table}[t]
	\caption{Relative accuracy improvement (in percentage) achieved by each  skew-symmetric approximation with respect to three symmetric counterparts, corresponding to the Gaussians from the  Laplace method, black-box VB and EP. This improvement is quantified, for varying sample sizes   $n=200,400,600,800,1000,$ under four different measures, namely the bias in approximating the posterior means of the standardized parameters, and the TV,  KL and reverse-KL  divergences between the target posterior marginals and the corresponding approximation, for each $j=1, \ldots, d$ (with $d =  6  \lfloor \sqrt{n/2} \rfloor$). Standardization involves dividing the parameters by the standard deviations of the associated exact posterior. As such these standard deviations are estimated via Monte Carlo using 40{,}000 samples from the target posterior produced by HMC via \texttt{rstan}. The relative improvement is computed under the four measures analyzed for every $j=1, \ldots, d$. The table displays the median relative improvement over $j$, for each  $n=200,400,600,800,1000,$ averaged across 10 replicated experiments. The quantity $\%$ \textsc{improv} corresponds, instead, to the percentage of marginals for which the proposed skew-symmetric approximation improves the Gaussian counterpart with respect to the associated accuracy measure. This percentage is averaged across $n$. }
	\label{tab_high_marg_0} 
	\centering %
\begin{tabular}{l|ccccc|c}
	\hline
	& $n=200$ & $n=400$ & $n=600$ & $n=800$ & $n=1000$  & \\ 
		& $(d=60)$ & $(d=84)$ & $(d=102)$ & $(d=120)$ & $(d=132)$   & $\%$ \textsc{improv} \\ 
	\hline
	Laplace  (Skew-sym. vs. Gaus.) & &&   & && \\ 
		\hfill TV &  $12.9\%$ & $11.1\%$& $10.3\%$ & $9.2\%$ & $9.5\%$& $100\%$  \\    
		\hfill KL &  $28.0\%$ & $22.7\%$  & $20.2\%$ &  $18.6\%$ &$17.3\%$&$99.8\%$\\  
		\hfill reverse-KL &  $37.1\%$ & $39.2\%  $& $40.7\%$ &  $40.9\%$ & $41.2\%$& $100\%$ \\        
		\hfill bias &  $32.8\%$ & $34.3\%  $& $34.0\%$ &  $34.8\%$ &$35.3\%$&$99.8\%$\\        
				\hline  
	black-box VB  (Skew-sym. vs. Gaus.) & &&   & && \\   
	\hfill TV  &  $14.1\%$ & $12.3\%  $ & $12.0\%$ &  $12.7\%$& $11.6\%$&$100\%$\\    
	\hfill KL &  $31.6\%$ & $25.7\%  $ & $22.0\%$ &  $20.9\%$&$19.1\%$& $100\%$\\  
\hfill reverse-KL&  $42.4\%$ & $42.2\%  $ & $41.0\%$ &  $43.7\%$&41.5\%& $100\%$\\   
	\hfill bias &  $55.2\%  $& $48.7\%$ &  $51.2\%$ &$56.1\%$&$50.8\%$& $95\%$\\             
				\hline  
		EP (Skew-sym. vs. Gaus.)  &  &  & & &&  \\ 
		\hfill TV  &  $17.3\%$ & $16.1\%  $ & $16.7\%$ &  $16.7\%$&$16.8\%$& $100\%$\\    
		\hfill KL &  $33.1\%$ & $28.7\%  $ & $26.4\%$ &  $24.9\%$& $24.1\%$& $100\%$\\  
				\hfill reverse-KL &  $43.3\%$ & $44.6\%  $ & $46.6\%$ &  $46.6\%$& $46.9\%$& $100\%$\\    
					\hfill bias &  $24.7\%  $& $31.9\%$ &  $30.6\%$ &$32.3\%$&$36.8\%$& $86\%$\\                
	  \hline
\end{tabular}
\end{table}

Table~\ref{tab_high_marg_0} quantifies, for each sample size  $n$, whether the posterior marginals $\pi_n({\theta_j})$, $j=1, \ldots, d$, resulting from the above Bayesian model are more accurately approximated by those of the three Gaussians considered in Section~\ref{sec_simu_uni} (i.e.\ $\bar{q}_{n,{\btheta}^*}(\theta_j)$, $j=1, \ldots, d$), or by the corresponding skew-symmetric perturbations (i.e.\ ${q}_{n,{\btheta}^*}(\theta_j)$,  $j=1, \ldots, d$). This is accomplished by computing, for every $n$ and $j$, the relative improvement (in percentage) achieved by the marginals of the proposed skew-symmetric approximation with respect to those of the Gaussian counterparts, under four accuracy measures. These are the TV, KL and reverse-KL divergences between the target posterior marginals and the corresponding approximations,  along with the bias in estimating the posterior expectations of the standardized parameters (see the caption of Table~\ref{tab_high_marg_0} for details on the standardization). This bias is useful to assess if the accurate characterization of the marginal densities translates also into practical improvements in approximating functionals of key interest for inference. Note that, unlike in Section~\ref{sec_simu_uni} (where the focus is on one-dimensional settings), all the performance measures discussed above are estimated here via Monte Carlo leveraging 10{,}000 i.i.d.\ draws from each of the approximations studied, while considering $4$ chains of size $10{,}000$ produced by HMC (via \texttt{rstan}) as benchmark samples from the target posterior $\pi_n(\btheta)$.

 Table~\ref{tab_high_marg_0} displays the medians of the above relative accuracy gains across $j$, further averaged over 10 replicated experiments. The results show that, albeit tested in a more challenging setting, the proposed skew-symmetric approximation still showcases systematic improvements over the Gaussian counterparts. Such gains are non-negligible for all the measures, and particularly remarkable for the reverse-KL and bias. Interestingly, the ability to reduce both the KL and the reverse-KL divergences clarifies that the proposed perturbation is able to improve the original Gaussians in both mode seeking and mass covering  behaviors, leading also to reductions in bias. According to $\%$ \textsc{improv}, these improvements are again systematic across marginals, sample sizes and approximation methods. Notice that, the $86\%$ associated with the bias of EP should be interpreted as a remarkable result. In fact, EP aims, by design, at providing Gaussian approximations centered at the actual posterior mean. As such, its skewed perturbation may not yield, in principle, noticeable improvements when the focus is specifically on the bias. Table~\ref{tab_high_marg_0} shows that, in practice, the skew-symmetric approximation provides clear gains also under this measure, likely due to a positive effect in removing possible biases that arise from EP optimization schemes.

Section~\ref{sec_4_app} clarifies that the  improvements achieved by the skew-symmetric approximations in  simulations are found also in real-data applications, both in relative and absolute terms.

%###############################################################################
%###############################################################################
\section{\large 5. Real-data applications} \label{sec_4_app}
We conclude by showcasing the practical gains of the proposed skew-symmetric approximation in two real-data applications requiring realistic hierarchical Bayesian implementations. The first (see Section~\ref{sec_zip}) shows that relative improvements similar to those presented in the simulation study in Section~\ref{sec_simu_high} can be obtained also under a more sophisticated zero-inflated negative binomial regression aimed at modeling school attendance behaviors leveraging the dataset \texttt{Attendance} in the \texttt{R} library \texttt{mixpoissonreg}. The second (see Section~\ref{sec_logistic}) considers, instead, a semi-parametric hierarchical logistic regression model to infer the determinants of contraceptive usage behaviors in 33 indian states based on data retrieved from the HDS II survey \citep[see, e.g.][]{rigon2019bayesian}. In this latter application $n=30{,}524$ and $d=62$, thereby providing an interesting setting to study also the runtimes of inference under the proposed skew-symmetric approximation, in addition to its gains with respect to the Gaussian counterparts in absolute terms.

Consistent with the simulation studies in Section~\ref{sec_4}, the above gains are again quantified with a focus on the skew-symmetric perturbations for the Gaussian approximations from the Laplace method \citep[][Chapter 13]{gelman1995bayesian}, black-box VB \citep[][]{kucukelbir2017automatic, ranganath2014black} and EP \citep[][]{minka2013expectation,vehtari2020expectation}, leveraging the four  accuracy measures in Table~\ref{tab_high_marg_0} (i.e.\  TV,   KL, reverse-KL and bias). As for the analyses within Section~\ref{sec_simu_high}, also in the two applications these measures are estimated via Monte Carlo based on  10{,}000 samples drawn from each of the approximations studied. Inference under the target posterior, which provides the benchmark to assess the accuracy of the approximations analyzed, leverages instead 4 chains of length 10{,}000 under HMC from the \texttt{R} library \texttt{rstan}. This library is also employed for obtaining the Laplace and black-box VB  approximations, via the functions \texttt{optimizing} and \texttt{vb}, respectively. The EP solution is instead implemented via a custom-built \texttt{R} code. Once these three symmetric approximations are available, the corresponding skew-symmetric counterparts can be obtained as in Definition~\ref{def_1}, and i.i.d.\ sampling from these skewed corrections proceeds via Algorithm~\ref{alg1}.

In interpreting the results in Sections~\ref{sec_zip}--\ref{sec_logistic} (along with those in Sections~\ref{sec_simu_uni}--\ref{sec_simu_high})  it is important to emphasize that the proposed  skewness-inducing mechanism  is constrained to redistribute the density of the symmetric approximation only among  pairs $\{\btheta, 2\btheta^* - \btheta\}$, $\btheta \in \Theta$. Thus, the systematic improvements achieved by the  skew-symmetric approximations in  Tables~\ref{tab_0_pois}--\ref{tab2_pois_neg} and Figure~\ref{figure:post} are arguably remarkable when accounting for the simplicity of the proposed perturbation strategy. These results suggest future research to devise even more sophisticated  skewness-inducing mechanisms redistributing the density at more flexible configurations than  $\{\btheta, 2 \btheta^* - \btheta\}$, $\btheta \in \Theta$. Such extensions should, however, account for the tractability of the resulting approximation. Recalling Section~\ref{sec_32} and Remark~\ref{rema1}, this is a key property of the  skewed perturbation we propose, which is fundamental to facilitate its broad  applicability. To extend this applicability even further, we  illustrate in the two real-data applications the use of  skew-symmetric approximations as improved proposals in importance sampling \citep[see, e.g.][Chapter 8]{chopin2020introduction}. This perspective yields remarkable gains in the effective sample sizes (ESS) relative to those obtained when leveraging the Gaussian counterparts as proposals.

%###############################################################################
\subsection{\large 5.1 Zero-inflated negative binomial regression} \label{sec_zip}
As a first application let us consider a zero-inflated  negative binomial regression \citep[e.g.][]{lawal2012zero} applied to a study on school attendance of $n=314$ students. The dataset is available within the \texttt{R} library \texttt{mixpoissonreg} and comprises, for each student $i$, information on the number of days of absence at school $y_i$, along with two categorical explanatory variables $x_{i1}$ and $x_{i2}$ corresponding, respectively, to the gender of the student (\texttt{female/male}) and the instructional program in which such a student is enrolled (\texttt{General/Academic/Vocational}). To learn how the school attendance behaviors relate to these two variables, while accounting for the excess of zeros and overdispersion in the observed responses $y_1, \ldots, y_n$, we consider a zero-inflated negative binomial regression \citep[e.g.][]{lawal2012zero}. This model incorporates gender and instructional program effects both in the zero inflation probabilities $\psi_i$, and also  in the negative binomial expected counts $\mu_i$. In particular, denoting with $\mathds{1}({\cdot})$ the indicator function, we let
\begin{equation*}
\begin{split}
\mbox{logit}(\psi_i)&=\alpha_0+\alpha_1\mathds{1}(x_{i1}=\texttt{male})+ \alpha_2\mathds{1}(x_{i2}=\texttt{Academic})+\alpha_3 \mathds{1}(x_{i2}=\texttt{Vocational}),\\
\log \mu_i&=\beta_0+\beta_1 \mathds{1}(x_{i1}=\texttt{male})+ \beta_2 \mathds{1}x_{i2}=\texttt{Academic})+\beta_3 \mathds{1}(x_{i2}=\texttt{Vocational}),
\end{split}
\end{equation*}
 for every unit $i=1, \ldots, n$,  and assume a constant overdispersion $\exp(\gamma)$, with $\gamma \in \mathbb{R}$. In order to perform Bayesian inference on the parameters \smash{$\btheta=(\gamma,\alpha_0,\alpha_1,\alpha_2,\alpha_3,\beta_0,\beta_1,\beta_2,\beta_3) \in \mathbb{R}^9$} we rely on independent Gaussian priors having mean $0$ and variance $2$, and then assess in Table~\ref{tab2_pois_neg} whether the proposed skewed perturbation improves the Gaussian alternatives from the Laplace method \citep[e.g.][Chapter 13]{gelman1995bayesian}, black-box VB \citep[][]{kucukelbir2017automatic, ranganath2014black} and EP \citep[][]{minka2013expectation,vehtari2020expectation} in approximating the resulting posterior. Note that, unlike for the simulations in Sections~\ref{sec_simu_uni}--\ref{sec_simu_high}, we consider here Gaussian priors rather than Student-$t$ ones. This is motivated by the aim to assess our skew-symmetric approximations under several elicitations of routine use, including those relying on Gaussian priors.

\begin{table}[t]
	\caption[]{Relative accuracy improvement (in percentage) achieved by each skew-symmetric approximation with respect to three symmetric counterparts, corresponding to the Gaussians from the  Laplace method, black-box VB and EP. This improvement is quantified under four different measures, namely the bias in approximating the posterior means of the standardized parameters, and the TV,   KL and reverse-KL  divergences between the target posterior marginals and the corresponding approximation.  Standardization proceeds as discussed in the caption of Table~\ref{tab_high_marg_0}.  The relative gain is computed under these four measures for every $j=1, \ldots, 9$. The table displays the median relative improvement over $j$. The quantity $\%$ \textsc{improv} corresponds, instead, to the percentage of marginals for which the skew-symmetric approximation improves the Gaussian counterpart, averaged across the accuracy measures.}
\label{tab2_pois_neg}	
	\centering %
	{%
\begin{tabular}{l|cccc|c}
	\hline
	& TV & KL & reverse-KL & bias &$\%$ \textsc{improv}  \\ 
	\hline
	Laplace  (Skew-sym. vs.\ Gaus.) &  $22.0\%$ & $39.5\%  $ & $47.5\%$ &  $29.5\%$& $100\%$\\   
	black-box VB  (Skew-sym. vs.\ Gaus.) \qquad& $22.0\%$ & $46.1\%  $ & $33.4\%$  & $26.8\%$& $96\%$\\   
		EP (Skew-sym. vs.\ Gaus.)  & $21.2\%$ & $36.3\%  $ & $34.2\%$ & $27.0\%$ &  $90\%$ \\ \hline
\end{tabular}}
\end{table}

Consistent  with the simulations in Section~\ref{sec_4},  Table~\ref{tab2_pois_neg} demonstrates that the proposed skew-symmetric approximations systematically outperform the corresponding symmetric counterparts for all the measures considered also in this realistic real-data application. Recalling the remarks before Section~\ref{sec_zip}, these relative gains are particularly remarkable, especially in the light of the simplicity of the perturbation strategy employed, and can be further strengthened by studying its use as an improved proposal within importance sampling targeting the joint posterior density $\pi_n(\btheta)$ \citep[e.g.][Chapter 8]{chopin2020introduction}. Such an additional assessment is not only useful to clarify the broad applicability of the proposed skew-symmetric approximation in Bayesian computation, but also complements the analyses  within  Table~\ref{tab2_pois_neg}. In fact, the magnitude of the gains in the ESS  \citep[e.g.][Chapter 8.6]{chopin2020introduction} relative to those obtained under the unperturbed Gaussian proposals are useful for quantifying the improvements of the proposed skew-symmetric approximation in characterizing the shape of the entire target posterior density $\pi_n(\btheta)$, beyond its marginals. The percentage gains in the ESS, estimated in 100 replicated studies, are $90.95\%$ (sd:\ $10.10\%$), $87.25\%$ (sd:\ $7.37\%$) and $89.91\%$ (sd:\ $6.57\%$) for the skew-symmetric perturbations of the Gaussian proposals from the Laplace method, black-box VB and EP, respectively. This  means that replacing such Gaussian proposals with the corresponding skewed counterparts essentially doubles the ESS, at negligible computational costs. In fact, in importance sampling, the multiple likelihood evaluations we require to compute the skewing factor are also necessary under the Gaussian proposals to evaluate the importance weights. These findings strengthen the impact of our skew-symmetric approximation, and motivate further research on its use  within importance sampling and its extensions.

%###############################################################################
\subsection{\large 5.2 Hierarchical semi-parametric logistic regression} \label{sec_logistic}
Let us conclude by illustrating the accuracy gains of skew-symmetric approximations in a demographic application with  $n=30{,}524$ and $d=62$. The data are retrieved from the HDS II survey whose focus, among others, is to study women contraceptive choices in 33 indian states. Recently, this dataset has been analyzed by \citet{rigon2019bayesian} via a Bayesian semi-parametric sequential logistic regression to investigate how different demographic factors affect the probability of four different contraceptive choices (including \texttt{no contraceptive}). In the following, we focus instead on the binary decision on whether to use or not contraceptives as a function of state of residence,  age, education (factor with four levels), religion (factor with four levels) and area (urban or rural).  

Let $\bx_{ik}^\intercal$ be the 8-dimensional vector comprising the binary encoding of the variables education, religion and area for the $i$-th woman within state $k=1, \ldots, K$. Consistent with the hierarchical structure of the data, we assume that the binary decision $y_{ik} \in \{0,1\}$ to use or not contraceptives is modeled, for each woman $i = 1,\dots,n_k$ in state $k = 1,\dots,K$, through a  Bernoulli  variable with probability $\psi_{ik}$ defined as
\begin{equation*}
\mbox{logit}(\psi_{ik}) =\alpha_k + g(\mbox{age}_{ik})+\bx_{ik}^\intercal \bbeta, \qquad \mbox{for every} \quad i = 1,\dots,n_k, \quad \mbox{and} \quad k = 1,\dots,K,
\end{equation*} 
where $n_k$ denotes the number of observations in state $k$, $\alpha_k$ is a state-specific intercept, $\bbeta \in \mathbb{R}^8$ denotes a vector encoding the fixed effects, while $g(\cdot)$ is a flexible functional effect for the variable age. Following \citet{rigon2019bayesian}, $g(\cdot)$  is defined via a linear combination of B-spline functions $B_m(\cdot)$, $m = 1,\dots,M$, with coefficients vector $\bgamma = (\gamma_1,\ldots,\gamma_M)^\intercal \in \smash{\mathbb{R}^M}$. Bayesian inference under this model proceeds via suitably defined priors for the $d=62$ parameters in $\btheta=(\bbeta,\balpha,\bgamma)$. In particular,  we assume improper uniform priors for the entries in $\bbeta$, and independent Student-$t$ with 3 degrees of freedom for  the $K$ intercepts within $\balpha$ and the $M$ entries of \smash{$\bSigma^{-1/2}\bgamma $}, where $\bSigma$ is  specified as in  \citet{rigon2019bayesian} to enforce similar coefficient values for contiguous splines. Note that, as discussed in Section~\ref{sec_simu_uni}, the choice of Student-$t$ coincides with the realistic assumption of Gaussian priors for the parameters combined with $\chi_3^2$ hyperpriors for the corresponding precisions.

\begin{figure}[t]
\centering
   \includegraphics[scale=0.7]{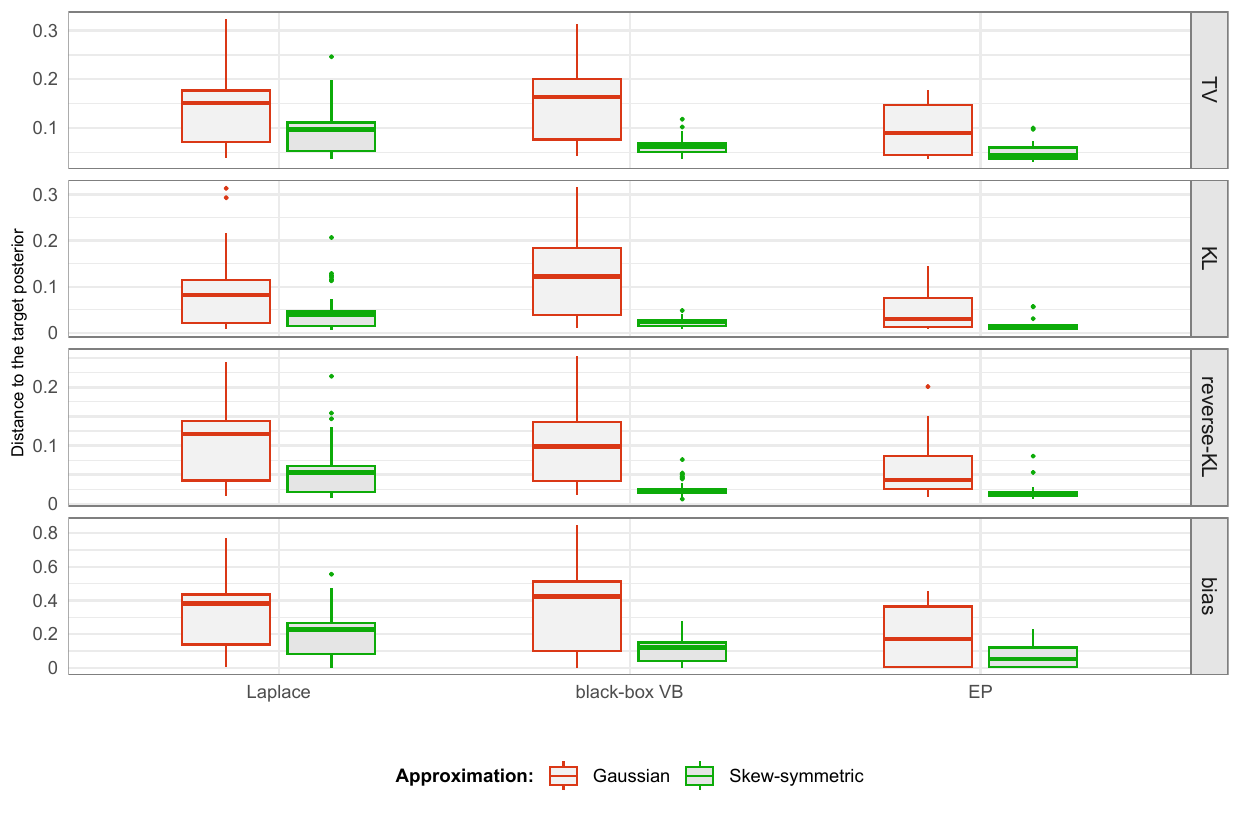}
		\caption{\footnotesize Empirical comparison of the accuracy achieved by three state-of-the-art Gaussian approximations from the Laplace method, black-box VB and EP, versus the corresponding skew-symmetric perturbations. For three routinely-employed divergences $\mathcal{D}$ (TV,  KL, reverse-KL) the first three panels display the boxplots of $\mathcal{D}(\pi_{j,n} \mid \mid q_{j,n})$, $j=1, \ldots, 62$, where $q_{j,n}$ is the $j$th marginal of either $\bar{q}_{n,{\btheta}^*}$ (Gaussian) or ${q}_{n,{\btheta}^*}$ (Skew-symmetric). The fourth panel shows instead the boxplot of the absolute differences between the  approximated and actual posterior means of the $d=62$ standardized parameters (standardization proceeds as discussed in the caption of Table~\ref{tab_high_marg_0}). 
 }
		\label{figure:post}
\end{figure}

Figure~\ref{figure:post} extends the comparisons in Section~\ref{sec_zip} by studying the accuracy in both relative and absolute terms of the skew-symmetric approximations over the unperturbed Gaussian counterparts from the   Laplace method, black-box VB and EP, with focus on the marginals of the posterior induced by the above model. More specifically, the first three panels within Figure~\ref{figure:post} display the boxplots of the estimated $\mathcal{D}(\pi_{j,n} \mid \mid q_{j,n})$, $j=1, \ldots, 62$, where $q_{j,n}$ is the $j$th marginal of either $\bar{q}_{n,{\btheta}^*}$ (Gaussian) or  ${q}_{n,{\btheta}^*}$ (Skew-symmetric), and $\mathcal{D}$ corresponds to the TV,  KL and reverse-KL divergences, respectively. The fourth panel shows instead the boxplot of the absolute differences between the  approximated and actual posterior means of the $d=62$ standardized parameters. Similarly to the empirical results in Sections~\ref{sec_simu_uni}--\ref{sec_simu_high} and \ref{sec_zip}, correcting for asymmetry yields also in this case approximations that systematically outperform the symmetric counterparts on all  accuracy measures. Comparing the medians of the boxplots in Figure~\ref{figure:post}, these improvements are, generally, by a multiplicative factor of 2 or more, and provide evident gains also in absolute terms. For instance, since the bias is computed on  standardized parameters, achieving systematic reductions of $0.2$ or more via simple perturbations is a remarkable result in absolute terms. 

Notice also that, among the deterministic approximations in Figure~\ref{figure:post}, the skewed perturbation of EP provides the most competitive strategy in terms of overall accuracy across the different summaries. Recalling, e.g.\ \cite{vehtari2020expectation}, EP already provides highly-accurate deterministic approximations in its standard Gaussian form. As a result, perturbing these symmetric densities through the proposed skewing factor yields even more accurate  skewed approximations. This point is also confirmed by the analysis of the relative improvements in the ESS when leveraging the skew-symmetric approximations instead of the  Gaussian counterparts as proposals in importance sampling. In this case, the percentage gains in the ESS, averaged across 100 replicated studies, are $76.88\%$ (sd:\ $9.25\%$), $59.88\%$ (sd:\ $6.90\%$) and $87.59\%$ (sd:\ $7.17\%$) for the skewed perturbations of the Gaussian proposals from the Laplace method, black-box VB and EP, respectively. 

Notably, the above improvements are obtained at negligible computational costs. In this application with $n=30{,}524$ and $d=62$, the runtime of Monte Carlo inference via $10{,}000$ samples from Algorithm~\ref{alg1}  is approximately $8$ seconds on a standard laptop, without parallel implementations.

%###############################################################################
%###############################################################################
\section{\large 6. Discussion} \label{sec_5}
This article introduces a novel and general strategy to perturb any given symmetric approximation of a generic posterior density for obtaining an improved, yet similarly tractable, skew-symmetric counterpart. Such an approximation is shown to improve, both theoretically and practically, the accuracy of the symmetric density which is perturbed. Unlike recently developed deterministic approximations based on generalizations of skew-normal distributions, the proposed solution (i) applies to generic posterior densities and to any symmetric approximation of such densities, from e.g.\ Laplace, VB and EP, (ii) does not imply additional optimization costs relative to those required for obtaining the original symmetric approximation to be improved, (iii) is substantially simpler and can be applied directly to any output of state-of-the-art softwares yielding symmetric approximations of posterior densities, (iv) has strong theoretical guarantees in terms of accuracy, and (v) achieves systematic empirical gains over the unperturbed symmetric counterparts in both simulation studies and real-data applications. Such advantages, combined with the limited computational burden of inference under these skew-symmetric approximations (see Sections~\ref{sec_32}--\ref{eff_ev} and \ref{sec_logistic}), motivate extensive use of this novel solution, especially in situations where the posterior density displays non-negligible skewness and the symmetric approximation to be perturbed provides an accurate characterization of the symmetrized posterior. 

Recalling Section~\ref{sec_32}, our proposal arises from a yet unexplored skew-symmetric representation of posterior densities in Proposition~\ref{prop1} that is of independent interest and stimulates several directions of future research. A promising one is discussed  in Remark~\ref{re2} and refers to the case in which also the symmetric approximating density  $\bar{q}_{n,{\btheta}^*}(\btheta)$  is unknown and part of the  optimization problem formalized in Corollary~\ref{cor_1}. This extension, combined with our results in Section~\ref{sec_3},  supports a change of perspective that suggests to focus on symmetrized posterior densities defined  in~\eqref{eq1}, rather than on the original posterior, as the target of symmetric approximations. To the best of our knowledge, such a perspective has not been  considered so far, thus stimulating active research motivated by the novel questions associated with this task. For instance,  this would require extracting a suitable symmetric component from the target posterior density, which is both as close as possible to such a posterior and can be also accurately approximated by a tractable symmetric density whose perturbation yields the final skew-symmetric approximation of ${\pi}_{n}(\btheta)$. A related problem can be found in classical, yet overlooked,  literature aimed at improving estimators of empirical distribution functions for symmetric densities   \citep{hinkley1976estimating,lo1985estimation,schuster1975estimating,schuster1987identifying}. Inheriting these results within our framework for deterministic approximations provides a promising direction to address the objectives discussed in Remark~\ref{re2}. 

The above perspective is also useful to further extend the theory in Section~\ref{sec_34} on the asymptotic accuracy of skew-symmetric approximations. In particular, as a consequence of the proof of Theorem~\ref{teo_3}, a natural  direction is to show that these solutions can be made arbitrarily accurate via higher order approximations only for the symmetrized component.
From a more practical perspective, another key direction is to illustrate the empirical performance of the proposed skew-symmetric perturbation when applied to other symmetric approximations and in more complex hierarchical formulations beyond those considered in Sections~\ref{sec_4}--\ref{sec_4_app}. Notice that our focus in Sections~\ref{sec_4}--\ref{sec_4_app} on Gaussian approximations arising from commonly adopted Laplace, black-box VB and EP methods is motivated by the attempt to stimulate implementation of the proposed skew-symmetric approximation in routinely used softwares. Nonetheless, showcasing the empirical gains of our proposal within the broader context of symmetric approximations discussed in Section~\ref{sec_2} (e.g.\ ALA, INLA, $t$-exponential approximating families and delta-method variational inference), is a natural direction of future research which could further boost the impact of our contribution. Recalling Section~\ref{sec_2}, it is also of interest to explore the benefits of the proposed  skew-symmetric approximation not only as a direct perturbation applied to the symmetric outputs of such methods, but also as a building-block within the corresponding optimization routines. For instance, as shown by a recent contribution  \citep{kock2025variational} appeared after our article and leveraging our results, the skew-symmetric approximation we derived could be incorporated directly within the algorithms of stochastic black-box extensions of VB \citep{ranganath2014black}, including ADVI  \citep{kucukelbir2017automatic}, without affecting the tractability of the stochastic gradient step, while enlarging (and hence improving) the class of approximating distributions. This is because the density of the skew-symmetric approximation is available in closed form, and i.i.d.\ samples from it can be obtained via a direct perturbation of those from the symmetric counterparts routinely used in black-box implementations of VB. Such a tractability could also motivate the use  of the proposed skewed perturbation to improve the initial Gaussian densities often employed in variational inference with normalizing flows \citep[e.g.][]{rezende2015variational}, while stimulating  research to assess the benefits of this idea.  Finally, the remarkable improvements in the effective sample size achieved by the skew-symmetric solution in Section~\ref{sec_4_app}, motivate its future use beyond the context of deterministic approximations, with a particular focus on importance sampling and its extensions  \citep[][]{chopin2020introduction}.

Recalling the discussions and results in Sections~\ref{sec_32}, \ref{eff_ev} and \ref{sec_logistic}, the computational burden associated with the proposed skew-symmetric approximation is not a major source of concern in practice, and can be further reduced through parallel implementations both across the i.i.d.\ samples and likelihood factors. Nonetheless, future research exploring ideas to further reduce the cost of the likelihood evaluations (both in the sample size $n$ and dimension $d$) is of interest. Data subsampling could be a direction, but naive implementations of this idea would not be compatible with the skew-symmetric representation of the posterior  induced by the entire sample, and hence, developing a principled solution  requires a careful treatment.

%###############################################################################
%###############################################################################
\section{\bf Acknowledgments}
The authors are grateful to the Editor, the Associate Editor and the referees for the constructive suggestions, which helped in improving the initial version of this article. In addition, the authors would like to thank Giacomo Zanella for the useful feedbacks on the ideas and methods developed in this article.
\\
\\

\section{\bf Funding}
Francesco Pozza is funded by the European Union (ERC, PrSc-HDBayLe, n.\ 101076564). Botond
Szabo is co-funded by the European Union (ERC, BigBayesUQ, n.\ 101041064). Views and opinions expressed are however those of the author(s) only and do not necessarily reflect those of the European
Union or the European Research Council Executive Agency. Neither the European Union nor the granting
authority can be held responsible for them. Daniele Durante acknowledges the support by the MUR-PRIN
2022 project “CARONTE” (2022KBTEBN, funded~by~the European Union -- Next Generation EU, Mission 4, CUP: J53D23009400001).

\section{\bf Data availability}
Data and codes are available at \url{https://github.com/Francesco16p/SkewAppr_Post}.

\vspace{40pt}
\begin{center}
{\Large \bf Supplementary Material to ``Skew--symmetric approximations of posterior distributions''}
\end{center}

\setcounter{section}{0}
\renewcommand{\thesection}{\Alph{section}}
\numberwithin{equation}{section}

\section{\large A. Proofs}\label{Proof}
Appendix \ref{Proof} contains the proofs of Lemma \ref{lemma_1} and Theorems \ref{teo_1}--\ref{teo_3}. Those of Proposition \ref{prop1}--\ref{prop2} and Corollary~\ref{cor_1} are direct, and discussed in the article. In the  proofs of Lemma~\ref{lemma_1}~and~Theorems~\ref{teo_1}--\ref{teo_2} we  assume, without loss~of~generality, that $\btheta^* = {\bf 0}$. Moreover, for notational convenience~we~omit $\btheta^*$~from~$\bar{\pi}_{n,{\btheta}^*}(\btheta)$,~$\bar{q}_{n,{\btheta}^*}(\btheta)$,  ${q}_{n,{\btheta}^*}(\btheta)$ and~${w}_{n,{\btheta}^*}(\btheta)$, thereby focusing on  $\bar{\pi}_{n}(\btheta)$, $\bar{q}_{n,}(\btheta)$,  ${q}_{n}(\btheta)$~and ${w}_{n}(\btheta)$, respectively. Furthermore, let us consider the partition $\Theta = \Theta_{+} \cup \Theta_{-}$, where $\btheta \in \Theta_{+}$ implies~$-\btheta \in \Theta_{-}.$ 

\label{sec_proofs}
\subsection{\large A1. Proof of Lemma \ref{lemma_1} } \label{appendix:lemma_1}
\begin{proof}.
Let us first prove Lemma \ref{lemma_1} under the  $\alpha$--divergence $ \mathcal{D}_{\alpha}(\cdot || \cdot)$ for some fixed $\alpha\in\mathbb{R}\backslash\{0,1\}$. In view of the skew--symmetric representation  \eqref{eq2} of the posterior distribution $\pi_n(\btheta)$, combined with the fact that both $\bar{\pi}_{n}(\btheta)$ and $\bar{q}_n(\btheta)$ are symmetric and  $w_n(-\btheta) = 1-w_n(\btheta)$ , $0\leq w_n(\btheta) \leq1$,  we obtain 
	\begin{equation} \label{help:dv:sym:1}
		\begin{aligned}
			&\mathcal{D}_{\alpha}[\pi_n \mid \mid \bar{q}_n] -	\mathcal{D}_{\alpha}[\bar{\pi}_{n} \mid \mid \bar{q}_n]  = \frac{1}{\alpha (1-\alpha)} \int \big[ \bar{\pi}_{n}(\btheta)^{\alpha}  \bar{q}_n(\btheta)^{1-\alpha}  - \{2\bar{\pi}_{n}(\btheta)w_n(\btheta)\}^{\alpha}   \bar{q}_n(\btheta)^{1-\alpha} \big] d\btheta\\
		& \qquad \qquad \qquad \qquad\quad	=  \int_{\Theta_{+}} \frac{1}{\alpha (1-\alpha)} \big[  2 - \{2w_n(\btheta)\}^{\alpha}  -\{2(1-w_n(\btheta))\}^{\alpha} \big]  \bar{\pi}_{n}(\btheta)^{\alpha}   \bar{q}_n(\btheta)^{1-\alpha}  d\btheta.
		\end{aligned}
	\end{equation}
	To proceed with the proof, note that the  first two derivatives of the function $k\mapsto \kappa(k) =~\big[ 2-(2k)^{\alpha} - \{2(1-k)\}^{\alpha}\big]/\{\alpha (1-\alpha)\}$ are
	$2 \{ -(2k)^{\alpha-1} + \{2(1-k)\}^{\alpha-1}\}/(1-\alpha)$ and $ 4 \{ (2k)^{\alpha-2} + \{2(1-k)\}^{\alpha-2}\}$, respectively. Observe that the first derivative is zero  if and only if $k= 0.5$, whereas  the second derivative is positive~for $k \in\big( 0,1\big)$, implying that $k = 0.5$ is a point of minimum. Since $\kappa(0.5) = 0$ and $\bar{\pi}_{n}(\btheta)^{\alpha}  \bar{q}_n(\btheta)^{1-\alpha}>0$ the last integral in \eqref{help:dv:sym:1} is non--negative. This implies  $\mathcal{D}_{\alpha}[\pi_n \mid \mid  \bar{q}_n] -	\mathcal{D}_{\alpha}[\bar{\pi}_{n} \mid \mid  \bar{q}_n] \geq 0$, and hence
	\begin{equation} \label{alpha:div:ineq:1}
		\mathcal{D}_{\alpha}[\bar{\pi}_{n} \mid \mid  \bar{q}_n] \leq \mathcal{D}_{\alpha}[\pi_n \mid \mid  \bar{q}_n],
	\end{equation}   
	for any $\alpha\in\mathbb{R}\backslash\{0,1\}$.

To conclude the analysis under the $\alpha$--divergence let us  study the two limiting cases $\alpha \to 1$ and $\alpha \to 0$. To this end, recall that, for two generic densities $p(\btheta)$ and $q(\btheta)$, we have $\lim_{\alpha \to 1}\mathcal{D}_{\alpha}[p\mid \mid q ] = \textsc{kl}[p\mid \mid q ] $ and $\lim_{\alpha \to 0}\mathcal{D}_{\alpha}[p\mid \mid q] = \textsc{kl}[q \mid \mid p ] $ \citep[e.g.,][]{cichocki2010families}. As a consequence, by exploiting again $w_n(-\btheta) = 1-w_n(\btheta)$, $0\leq w_n(\btheta) \leq1$  and the symmetry of $\bar{\pi}_{n}(\btheta)$ and $\bar{q}_n(\btheta)$ we have
	\begin{equation*} \label{help:dv:sym:2}
		\begin{aligned}
			&\lim_{\alpha \to 1}( \mathcal{D}_{\alpha}[\pi_n \mid \mid \bar{q}_n] -	\mathcal{D}_{\alpha}[\bar{\pi}_{n} \mid \mid \bar{q}_n]) =  \textsc{kl}[\pi_n \mid \mid \bar{q}_n] -	\textsc{kl}[\bar{\pi}_{n} \mid \mid \bar{q}_n]  	\\
			&\qquad \qquad \qquad {=} 	 \int \Big[ 2 w_n(\btheta) \log [2 w_n(\btheta)]  \Big] \bar{\pi}_{n}(\btheta)d \btheta \\
			&\qquad \qquad \qquad {=} {\int_{\Theta_{+}}} \big[ 2 w_n(\btheta) \log [2 w_n(\btheta) ] +  2 (1- w_n(\btheta)) \log [ 2\{1- w_n(\btheta)\} ] \big]\bar{\pi}_{n}(\btheta)d \btheta.
		\end{aligned}
	\end{equation*}
Now, the first and second derivatives of the function $k\mapsto \kappa_1(k) =  2 k \log \big( 2 k \big) +  2 (1- k) \log [ 2(1- k )]$~are  $2[ \log(k) -  \log(1-k) ]$ and $2[ 1/k +1/(1-k) ]$, respectively. Since the second derivative is positive for every $k \in\big( 0,1\big)$ and the first derivative is zero if and only if $k= 0.5$, $\kappa_1(\cdot)$~has~a minimum in $k = 0.5$. Therefore, combining such a result with the fact that $\kappa_1(0.5) = 0$~and $\bar{\pi}_{n}(\btheta)>0$, implies that the last integral in the above expression is non--negative. Hence, $\textsc{kl}[\pi_n \mid \mid \bar{q}_n] -	\textsc{kl}[\bar{\pi}_{n} \mid \mid \bar{q}_n]  \geq 0$, which leads to
	\begin{equation} \label{alpha:div:ineq:2}
		 	\textsc{kl}[\bar{\pi}_{n} \mid \mid \bar{q}_n]\leq \textsc{kl}[\pi_n \mid \mid \bar{q}_n].
	\end{equation}
	
\noindent	Finally, if $\alpha\to 0$, it is possibile to apply a similar reasoning as above to obtain
	\begin{equation*} \label{help:dv:sym:3}
		\begin{aligned}
			\lim_{\alpha \to 0} \mathcal{D}_{\alpha}[\pi_n \mid \mid \bar{q}_n] -	\mathcal{D}_{\alpha}[\bar{\pi}_{n} \mid \mid \bar{q}_n]  &=  \textsc{kl}[\bar{q}_n \mid \mid \pi_n ] -	\textsc{kl}[\bar{q}_n\mid \mid \bar{\pi}_{n}]  \\
			&=  \int - \log( 2  w_n(\btheta)  )\bar{q}_n(\btheta) d \btheta\\
			&= \int_{\Theta_{+}} \big[- \log( 2  w_n(\btheta)  )-  \log( 2[1- w_n(\btheta)]  ) \big] \bar{q}_n(\btheta) d \btheta.
		\end{aligned}
	\end{equation*}
In this case, the	 function $k\mapsto \kappa_2(k) = -[ \log( 2  k  )+  \log( 2(1- k)  ) ] $, has a first derivative  $- 1/k + 1/(1-k) $ which is zero if and only if $k=0.5$, while the second derivative $ 1/k^2 + 1/(1-k)^2$ is positive for $k\in(0,1)$. Hence, the minimum of   $\kappa_2(\cdot)$ is $\kappa_2(0.5) = 0$. Such a result, together with the fact that $\bar{q}_n(\btheta)>0$, implies~that $\textsc{kl}[\bar{q}_n \mid \mid \pi_n ] -	\textsc{kl}[\bar{q}_n\mid \mid \bar{\pi}_{n}] \geq 0$, and therefore
	\begin{equation} \label{alpha:div:ineq:3}
			\textsc{kl}[\bar{q}_n\mid \mid \bar{\pi}_{n}] \leq \textsc{kl}[\bar{q}_n \mid \mid \pi_n ].
	\end{equation}
	The statement of Lemma \ref{lemma_1} regarding the $\alpha$--divergences, reverse--KL and KL follows by combining \eqref{alpha:div:ineq:1}, \eqref{alpha:div:ineq:2} and \eqref{alpha:div:ineq:3}.
	
	Let us conclude the proof by showing that the same inequality holds under the TV distance. Also in~this case, in view of the skew--symmetric representation  \eqref{eq2} of the posterior distribution $\pi_n(\btheta)$, combined with the fact that both $\bar{\pi}_{n}(\btheta)$ and $\bar{q}_n(\btheta)$ are symmetric and  $w_n(-\btheta) = 1-w_n(\btheta)$, $0\leq w_n(\btheta) \leq1$,  we obtain, leveraging the triangle inequality, that
	\begin{equation*}
		\begin{aligned}
			&2{\cdot}   \mathcal{D}_{\textsc{tv}}[\pi_n \mid \mid \bar{q}_n] - 2{\cdot}\mathcal{D}_{\textsc{tv} }[\bar{\pi}_{n}  \mid \mid \bar{q}_n]  \\
		 &\qquad = 	\int_{\Theta_{+}} \Big[ \big| 2\bar{\pi}_{n}(\btheta)w_n(\btheta) -  \bar{q}_n(\btheta) \big| +  \big| 2\bar{\pi}_{n}(\btheta)(1-w_n(\btheta)) -  \bar{q}_n(\btheta) \big| \Big] d \btheta - 2\int_{\Theta_{+}}  \big|\bar{\pi}_{n}(\btheta) - \bar{q}_n(\btheta) \big|d \btheta\\
		 &\qquad 	\geq  2 \int_{\Theta_{+}} \big| \bar{\pi}_{n}(\btheta) -  \bar{q}_n(\btheta) \big| d \btheta - 2\int_{\Theta_{+}} \big|\bar{\pi}_{n}(\btheta) - \bar{q}_n(\btheta) \big|d \btheta = 0.
		\end{aligned} 
	\end{equation*}
	Therefore, $  \mathcal{D}_{\textsc{tv} }[\bar{\pi}_{n} \mid \mid \bar{q}_n] \leq  \mathcal{D}_{\textsc{tv}}[\pi_n\mid \mid \bar{q}_n]   $, which concludes the proof.
\end{proof}

%###############################################################################

\vspace{5pt}
\subsection{\large A.2 Proof of Theorem \ref{teo_1}} \label{appendix:teo_1}
\begin{proof}.
	First notice that Lemma \ref{lemma_1} combined with the result in \eqref{skew:sym:equiv:sym} directly implies  \eqref{skew:sym:non:asymp}. Hence, we~are left to prove \eqref{skew:sym:equiv:sym}. To this end, let us first consider the TV distance $\mathcal{D}_{\textsc{tv}}({\cdot} {\mid} {\mid} {\cdot})$. Leveraging \eqref{eq2}--\eqref{skew:sym:post},~along~with $w_n(-\btheta) = 1-w_n(\btheta)$ , $0\leq w_n(\btheta) \leq1$, and the symmetry of $\bar{\pi}_{n}(\btheta)$ and $\bar{q}_n(\btheta)$ we~obtain
	\begin{equation*}
		\begin{aligned}
			\mathcal{D}_{\textsc{tv}}[\pi_n \mid \mid q_n]  &= \frac{1}{2}\int \big| 2\bar{\pi}_{n}(\btheta)w_n(\btheta) -  2 \bar{q}_n(\btheta)w_n(\btheta) \big|d \btheta\\
			&=   \frac{1}{2}\int_{\Theta_{+}} [ 2 w_n(\btheta) \big | \bar{\pi}_{n}(\btheta) -  \bar{q}_n(\btheta) \big| +  2 (1-w_n(\btheta))\big| \bar{\pi}_{n}(\btheta) - \bar{q}_n(\btheta) \big| ] d \btheta\\
			&=  \int_{\Theta_{+}} \big| \bar{\pi}_{n}(\btheta) -  \bar{q}_n(\btheta) \big| d \btheta = \frac{1}{2} \int \big| \bar{\pi}_{n}(\btheta) -  \bar{q}_n(\btheta) \big| d \btheta=\mathcal{D}_{\textsc{tv}}[\bar{\pi}_{n}  \mid \mid \bar{q}_n],
		\end{aligned} 
	\end{equation*}
	which coincides with \eqref{skew:sym:equiv:sym} under the TV distance. Similar arguments lead also to
	\begin{equation} \label{skew:sym:eq:sym:alpha}
		\begin{aligned}
			\mathcal{D}_{\alpha}[\pi_n \mid \mid q_n] &= \frac{1}{\alpha(1-\alpha)}[ 1 - \int 2w_n(\btheta)\bar{\pi}_{n}(\btheta)^{\alpha}  \bar{q}_n(\btheta)^{1-\alpha} d \btheta] \\
			&= \frac{1}{\alpha(1-\alpha)}[ 1 - 2 \int_{\Theta_{+}} \bar{\pi}_{n}(\btheta)^{\alpha}  \bar{q}_n(\btheta)^{1-\alpha} d \btheta ] = \mathcal{D}_{\alpha}[\bar{\pi}_{n}  \mid \mid \bar{q}_n],
		\end{aligned}
	\end{equation}
	since $ 2 \int_{\Theta_{+}} \bar{\pi}_{n}(\btheta)^{\alpha}  \bar{q}_n(\btheta)^{1-\alpha} d \btheta=\int \bar{\pi}_{n}(\btheta)^{\alpha}  \bar{q}_n(\btheta)^{1-\alpha} d \btheta$. To conclude the proof, notice that \eqref{skew:sym:eq:sym:alpha} holds also for $\alpha \to 0$ and $\alpha \to 1$, thereby recovering the same result for KL and reverse--KL. 
\end{proof}

\vspace{5pt}

%###############################################################################

\subsection{\large A.3 Proof of Theorem \ref{teo_2}} \label{appendix:teo_2}
\begin{proof}.	
	Let us first prove Theorem \ref{teo_2} for the TV distance. Similarly to $\bar{\pi}_{n}(\btheta)$, $\bar{q}_{n,}(\btheta)$,  ${q}_{n}(\btheta)$~and ${w}_{n}(\btheta)$, we omit $\btheta^*$ also in $ \tilde{q}_{n,\btheta^*}(\btheta)$ and $\tilde{w}_{\btheta^*}(\btheta)$ for notational convenience, obtaining  $ \tilde{q}_{n}(\btheta)$ and $\tilde{w}(\btheta)$.  In view~of  the skew--symmetric representation in \eqref{eq2}, the  triangle inequality, and equation \eqref{skew:sym:equiv:sym} in Theorem \ref{teo_1},~we~have
	\begin{equation*}
		\begin{aligned}
			 \mathcal{D}_{\textsc{tv}}[{\pi}_{n} \mid \mid \tilde{q}_n]  &= \, \frac{1}{2} \int \big| 2\bar{\pi}_{n}(\btheta)w_n(\btheta) -  2\bar{q}_n(\btheta)\tilde{w}(\btheta) \big| d \btheta \\
			&= \, \frac{1}{2}  \int_{\Theta_{+}} \big| 2\bar{\pi}_{n}(\btheta)w_n(\btheta) - 2\bar{q}_n(\btheta)\tilde{w}(\btheta) \big| + \big| 2\bar{\pi}_{n}(\btheta)(1-w_n(\btheta)) - 2\bar{q}_n(\btheta)(1-\tilde{w}(\btheta)) \big| d \btheta  \\ 
			&\geq  \, \frac{1}{2} \int \big| \bar{\pi}_{n}(\btheta) -  \bar{q}_n(\btheta) \big| d \btheta =  \frac{1}{2} \int \big| {\pi}_{n}(\btheta) -  {q}_n(\btheta) \big| d \btheta=\mathcal{D}_{\textsc{tv}}[{\pi}_{n}\mid \mid q_n],
		\end{aligned}
	\end{equation*}  
which proves 	Theorem \ref{teo_2} under the TV distance. Let us now consider the $\alpha$--divergence for $\alpha\in\mathbb{R}\backslash\{0;1\}$. Then, combining results in  \eqref{eq2}, \eqref{skew:sym:post}, and \eqref{skew:sym:equiv:sym}, it follows
	\begin{equation*}
		\begin{aligned}
			&\mathcal{D}_{\alpha}[{\pi}_{n} \mid \mid \tilde{q}_n] - \mathcal{D}_{\alpha}[{\pi}_{n} \mid \mid q_n]\\
			&\qquad  \quad = \frac{1}{\alpha(1-\alpha)} \Big[ \int 	\bar{\pi}_{n}(\btheta)^{\alpha}  \bar{q}_n(\btheta)^{1-\alpha} d \btheta - \int 	[2\bar{\pi}_{n}(\btheta)w_n(\btheta)]^{\alpha} [2 \bar{q}_n(\btheta)\tilde{w}(\btheta)  ]^{1-\alpha} d \btheta \Big]	\\
		&\qquad \quad    =	2\int_{\Theta_{+}} \frac{1}{\alpha(1-\alpha)}  \big[ 1 -w_n(\btheta)^{\alpha}\tilde{w}(\btheta)^{1-\alpha} -[1-w_n(\btheta)]^{\alpha}[1-\tilde{w}(\btheta)]^{1-\alpha}  \big] \bar{\pi}_{n}(\btheta)^{\alpha}   \bar{q}_n(\btheta)^{1-\alpha} d \btheta.
		\end{aligned}
	\end{equation*} 
Note that the function $v\mapsto \kappa_3(k,v)=\big[ 1 -k^{\alpha}v^{1-\alpha} -(1-k)^{\alpha}(1-v)^{1-\alpha}  \big]/[\alpha(1-\alpha)]$ satisfies  
	\begin{equation*} 
		\lim_{v \to 0^+}\kappa_{3}(k,v) \geq 0, \quad \text{and} \quad	\lim_{v \to 1^-} \kappa_{3}(k,v) \geq 0,
	\end{equation*}	
	for all $k \in (0,1)$. Moreover, $v\mapsto \kappa_{3}(k,v)$ is continuous in $(0,1)$ with first derivative  $\big[ -k^{\alpha}v^{-\alpha} +(1-k)^{\alpha}(1-v)^{-\alpha}  \big]/\alpha$, which is $0$ if and only if $k = v$, and second derivative $ k^{\alpha}v^{-(\alpha+1)} +(1-k)^{\alpha}(1-v)^{-(\alpha+1)}\geq 0$ for all~$v \in(0,1)$. Since $\kappa_{3}(k,k)=0$, then $\kappa_{3}(k,v)\geq 0$ for all $k,v\in(0,1)$. Thus, it remains to study $\kappa_{3}(k,v)$~at~the~boundaries.

First notice that $k=v=0$ is discarded, since it corresponds to the case $\pi_n(\btheta)=\tilde{q}_n(\btheta)=q_n(\btheta)=0$. Next, for $k \to 0^+$ and $v\in (0,1)$ fixed, we have that $\lim_{k \to 0^+} \kappa_{3}(k,v) \geq 0$~if~\smash{$ \alpha \in (0,1)$}, $\lim_{k \to 0^+} \kappa_{3}(k,v)= +\infty$ if $ \alpha <0$, and $\lim_{k \to 0^+} \kappa_{3}(k,v) \geq 0 $ if $ \alpha> 1$. The same limits hold also for $k \to 1^-$ and $v \in (0,1)$ fixed.~Moreover,  $\lim_{k \to 1^-,v \to 0^+ } \kappa_{3}(k,v) > 0$ if $ \alpha \in (0,1)$ and $\lim_{k \to 1^-,v \to 0^+ } \kappa_{3}(k,v) = +\infty$ if $ \alpha\notin[0,1]$,~which~hold, by symmetry, also for $k \to 0^+$~and~\smash{$v \to 1^-$}.  Finally, for $ \alpha \in (0,1)$, we have that $\lim_{k \to 1^-,v \to 1^- } \kappa_{3}(k,v) = 0 $. At the same time, for $ \alpha\notin[0,1]$, we have that $-(1-k)^{\alpha}(1-v)^{1-\alpha}/[\alpha(1-\alpha)]\geq 0$, hence $\kappa_{3}(k,v) \geq ( 1- k^{\alpha}v^{1-\alpha} )/[\alpha(1-\alpha)]$~with $\lim_{k \to 1^-,v \to 1^- } ( 1- k^{\alpha}v^{1-\alpha} )/[\alpha(1-\alpha)] = 0$. Leveraging symmetry arguments, these limits hold~also for  $k \to 0^+$, $v \to 0^+ $. Combining all these results, implies,~for~$\alpha\in\mathbb{R}\backslash\{0;1\}$,
	\begin{equation*}
		\mathcal{D}_{\alpha}[{\pi}_{n} \mid \mid \tilde{q}_n]  \geq\mathcal{D}_{\alpha}[{\pi}_{n} \mid \mid q_n].
	\end{equation*}
	
	To conclude, it remains to consider the limits $\alpha \to 0$ and $\alpha \to 1$. As in the proof of Lemma~\ref{lemma_1}, let us leverage again the connection with Kullback--Leibler divergences. For $\alpha \to 1$, we have 
	\begin{equation} \label{skew:opt:alpha1}
		\begin{aligned}
			& \lim_{\alpha \to 1} 	 \mathcal{D}_{\alpha}[{\pi}_{n} \mid \mid \tilde{q}_n]  = \textsc{kl}[{\pi}_{n} \mid \mid \tilde{q}_n]  = \int \log \Big( \frac{2\bar{\pi}_{n}(\btheta)w_n(\btheta)}{ 2\bar{q}_n(\btheta)\tilde{w}(\btheta)} \Big) 2\bar{\pi}_{n}(\btheta)w_n(\btheta) d \btheta \\
			&\qquad= \int \log \Big( \frac{2\bar{\pi}_{n}(\btheta)w_n(\btheta)}{2\bar{\pi}_{n}(\btheta)\tilde{w}(\btheta)} \Big) 2\bar{\pi}_{n}(\btheta)w_n(\btheta) d \btheta + \int \log \Big( \frac{2\bar{\pi}_{n}(\btheta)\tilde{w}(\btheta)}{ 2\bar{q}_n(\btheta)\tilde{w}(\btheta) } \Big) 2\bar{\pi}_{n}(\btheta)w_n(\btheta) d \btheta\\
			&\qquad= \textsc{kl}[{\pi}_{n} \mid \mid 2\bar{\pi}_{n}\tilde{w}]  +  \textsc{kl}[\bar{\pi}_{n} \mid \mid \bar{q}_n]  \geq \textsc{kl}[\bar{\pi}_{n} \mid \mid \bar{q}_n]  = \textsc{kl}[{\pi}_{n} \mid \mid q_n],
		\end{aligned}
	\end{equation} 
	where the second equality follows by adding and subtracting $\int \log (2\bar{\pi}_{n}(\btheta)\tilde{w}(\btheta) ) $ $2\bar{\pi}_{n}(\btheta)w_n(\btheta) d \btheta$, whereas the third by the definition of KL divergence, together with $w_n(-\btheta) = 1-w_n(\btheta) $. The inequality is the result of the positive definiteness of the KL divergence. Finally, the last equality follows from equation~\eqref{skew:sym:equiv:sym} in Theorem \ref{teo_1}. Reversing the role of  ${\pi}_{n}(\btheta)$ and  $\tilde{q}_n(\btheta)$, and leveraging a similar reasoning as in \eqref{skew:opt:alpha1},~yields~also
	\begin{equation*}
		\begin{aligned}
			 \lim_{\alpha \to 0} \mathcal{D}_{\alpha}[{\pi}_{n} \mid \mid \tilde{q}_n] = \textsc{kl}[\tilde{q}_n \mid \mid {\pi}_{n}]
			 \geq \textsc{kl}[ \bar{q}_n \mid \mid \bar{\pi}_{n}]   =   \textsc{kl}[q_n \mid \mid {\pi}_{n}],
		\end{aligned}
	\end{equation*}
which concludes the proof of  Theorem \ref{teo_2}.
\end{proof}

%###############################################################################

\vspace{5pt}
\subsection{\large A.4 Proof of Theorem \ref{teo_3}}
\begin{proof}. Let us prove the rates in  \eqref{tot:variatindist:map:sym2} first for the TV distance, and then for  $\alpha$--divergences~with $\alpha \in (0,1)$, including the limiting KL ($\alpha \to 0$) and reverse--KL ($\alpha \to 1$) forms. This latter~setting addresses routinely--employed special cases of $\alpha$–divergences \citep[e.g.,][]{margossian2024ordering,hernandez2016black}. 
	The proof for  \eqref{tot:variatindist:map:sym} follows as a direct consequence~leveraging~simpler arguments and derivations. 
	
	For notational simplicity define $\hat{\btheta}:=\btheta_{\textsc{map}} $. Then, in view of the invariance of the TV distance,  KL, reverse--KL and generic $\alpha$--divergences with respect to invertible affine transformations, in the following we will consider, for convenience, the re--parameterization \smash{$  \hat{\bh}  = \sqrt{n}(\btheta - \hat\btheta) $}; recall~that \smash{$  \hat{\bh} := \bh_{\textsc{map}}$}. Notice~that the density \smash{$ q_{n,\hat{\btheta},(2)}(\hat \bh)$} takes  the form $$ q_{n,\hat {\btheta},(2)}(\hat \bh)=2 \phi_d(\hat \bh ; {\bf 0}, \hat \bOmega)f(\hat \bh) w_{\bf 0}(\hat \bh)/\E_{{\bf 0},\hat \bOmega}  \{ f(\hat \bh)\},$$ where $ \hat \bOmega^{-1} = \bJ_{ \hat \btheta}/n$, \smash{$w_{\bf 0}(\hat \bh) = w_{n,\hat \btheta}(\hat \btheta+ \hat \bh/\sqrt{n}) $}, $\E_{{\bf 0},\hat \bOmega} \{ f(\hat \bh)\} = \int f(\hat \bh)  \phi_d(\hat \bh; {\bf 0}, \hat \bOmega ) d \hat \bh$ and
\begin{equation}\label{def:Ptilde}
 f(\hat \bh) = 1 + \frac{1}{24n} \Big \langle \frac{\ell^{(4)}_{n,\hat \btheta}}{n}, \hat \bh^{\otimes 4} \Big \rangle + \frac{1}{2} \Big(\frac{1}{24 n} \Big \langle \frac{\ell^{(4)}_{n,\hat \btheta}}{n}, \hat \bh^{\otimes 4} \Big \rangle  \Big)^2 
	+ \frac{1}{2} \Big(\frac{1}{6 \sqrt{n}} \Big \langle \frac{ \ell_{n,\hat \btheta}^{(3)}}{n}, \hat \bh^{\otimes 3}\Big \rangle \Big)^2.
\end{equation}
	 Let us also define  the event \smash{$ B_n=\lbrace  \| \hat{\btheta} - \btheta_0 \| \leq M_n\sqrt{d/n} \rbrace$} together~with~the two sets $\tilde K_n = \{ \btheta \, : \,  \|\sqrt{n}(\btheta - \smash{\hat \btheta})\| < 2 M_n\sqrt{d}  \}$ and $ K_n = \{ \btheta \, : \,  \|\sqrt{n}(\btheta - \btheta_0)\| < M_n\sqrt{d}   \}$. 

\vspace{10pt}	
\noindent	{\bf Proof for the \textsc{TV} distance.} By the triangle inequality, we have that
		\begin{equation} \label{triangle:tv:map:sym}
			\begin{aligned}
				&\int| \pi_n(\hat \bh) - { q_{n,\hat{\btheta},(2)}(\hat \bh)} | d\hat\bh\\
				&  \leq  	\int| \pi_n(\hat \bh) -  \pi_n^{\tilde K_n}(\hat \bh) | d\hat \bh + 	{\int}| \pi_n^{\tilde K_n}( \hat \bh) - { q^{\tilde K_n}_{n,\hat{\btheta},(2)}(\hat \bh)}| d\hat \bh +  {\int}| { q_{n,\hat{\btheta},(2)}(\hat \bh)}  - { q^{\tilde K_n}_{n,\hat{\btheta},(2)}(\hat \bh)}| d\hat \bh,  
			\end{aligned}
	\end{equation} 
	with
		\smash{$\pi_n^{\tilde K_n}( \hat \bh)= \pi_{n}(\hat \bh)\mathbbm{1}_{\hat \bh \in \tilde K_n}/\int_{\tilde K_n} \pi_n(\hat \bh) d\hat \bh$} and \smash{${q^{\tilde K_n}_{n,\hat{\btheta},(2)}(\hat \bh)= q_{n,\hat{\btheta},(2)}(\hat \bh)\mathbbm{1}_{\hat \bh \in \tilde K_n}/ \int_{\tilde K_n} q_{n,\hat{\btheta},(2)}(\hat \bh)d\hat \bh}$}, denoting the restrictions of \smash{$\pi_n(\hat \bh)$} and \smash{${ q_{n,\hat{\btheta},(2)}(\hat \bh)}$} to the set $\tilde K_n$. Moreover, note that $w_{\bf 0}(-\hat \bh) = 1-w_{\bf 0}(\hat \bh)$, and~that $\tilde K_n$ is symmetric about {\bf 0} ({when expressed as a function of $\hat \bh$}). Therefore, the normalizing constant of~the restricted skew--symmetric approximation is {$\int_{\tilde K_n} \smash{q_{n,\hat{\btheta},(2)}(\hat \bh)d\hat \bh}  = \int_{\tilde K_n} \phi_d( \hat \bh ; {\bf 0}, \hat \bOmega)f(\hat \bh)/\E_{{\bf 0},\hat \bOmega} \{ f(\hat \bh)\}  d\hat \bh$}.
		
Let us now deal with each of the three terms on the right hand side of  \eqref{triangle:tv:map:sym} separately.~In~view of~a~standard inequality of the TV distance we have \smash{$\int| \pi_n(\hat \bh) -  \pi_n^{\tilde K_n}(\hat \bh) | d\hat  \bh \leq 2 \Pi_n(\tilde K_n^c)$}. Moreover, combining~Assumption~\ref{cond:m1}, with the triangle and Markov's inequality, yields
	\begin{equation*}
		\Pi_n(\tilde K_n^c)\mathbbm{1}_{B_n} 		\leq \Pi_n(K_n^c)\mathbbm{1}_{B_n},
	\end{equation*}
	with $P_0^n(B_n) = 1 - o(1)$. From Lemma \ref{lemma:post:contr}, for a sufficiently large choice of $c_0$ in $M_n$,~the~right--hand--side of the above display is of order $O_{P_{0}^n}(n^{{-2}})$. This  implies, in turn,
	\begin{equation} \label{bound:out:pi:map:sym}
		\int| \pi_n(\hat \bh) -  \pi_n^{\tilde K_n}(\hat \bh) | d\hat \bh 		= O_{P_{0}^n}(n^{{-2}}).
	\end{equation}	
	To deal with the third term in the right--hand--side of \eqref{triangle:tv:map:sym}, let us exploit the same TV inequality. This, the symmetry of the set \smash{$\tilde K_n^c$ around ${\bf 0}$} ({when expressed as a function of \smash{$\hat \bh$}}) and the skew--symmetric~invariance with respect to even functions \citep[][Prop.\ 1.4]{azzalini2013skew}~give 
	\begin{equation*}
		\begin{aligned}
			 \int| { q_{n,\hat{\btheta},(2)}(\hat \bh)}  -{ q^{\tilde K_n}_{n,\hat{\btheta},(2)}(\hat \bh)}| d\hat \bh &\leq 2 \int_{ \hat \bh \,:\, \|\hat \bh\|>2 M_n\sqrt{d} }{ q_{n,\hat{\btheta},(2)}(\hat \bh)}  d\hat \bh\\
			 &\qquad  \qquad =2 \int_{ \hat \bh \,:\, \|\hat \bh\|>2 M_n\sqrt{d} }\frac{\phi_d( \hat \bh ; {\bf 0}, \hat \bOmega)f(\hat \bh)}{\E_{{\bf 0},\hat \bOmega} \{ f(\hat \bh)\}}  d\hat \bh.\\
			&\qquad  \qquad :=2 P_{{\bf 0},\hat  \bOmega, f(\cdot)}(\|\hat \bh\| > 2M_n\sqrt{d}).
		\end{aligned}
	\end{equation*}
	Furthermore, in view of \smash{$\hat \bOmega^{-1} = \bJ_{\hat \btheta}/n$} and $\tilde A_{n,0}= \{ \lambda_{\textsc{min}}(\hat \bOmega^{-1})> \bar \eta_1 \} \cap\{ \lambda_{\textsc{max}}(\hat \bOmega^{-1})< \bar \eta_2 \}$, it follows that, conditioned on \smash{$\tilde A_{n,0}$}, the quantity \smash{$\E_{{\bf 0},\hat \bOmega} \{ f(\hat \bh)\} $} lies on a bounded positive range and, for $n$ sufficiently~large, {$ 1-\log\{ f(\hat \bh)\}/(\hat \bh^\intercal \hat \bOmega^{-1} \hat \bh/2) > 0.5, $ uniformly in $\hat \bh\in \tilde{K}_n^c$}. As a consequence, for large $n$,
	\begin{equation} \label{thm1:help0:map:sym}
		\begin{aligned}
			& \int_{ \hat \bh \,:\, \|\hat \bh\|>2 M_n\sqrt{d} } \phi_d( \hat \bh ; {\bf 0}, \hat \bOmega)f(\hat \bh)/\E_{{\bf 0},\hat \bOmega}\{ f(\hat \bh)\} d\hat \bh \mathbbm{1}_{\tilde A_{n,0}} \lesssim  \int_{ \hat \bh \,:\, \|\hat \bh\|>2 M_n\sqrt{d} } \phi_d( \hat \bh ; {\bf 0}, 2 \hat \bOmega) d\hat \bh \mathbbm{1}_{\tilde A_{n,0}}.
		\end{aligned}
	\end{equation}
	Moreover, for every $\epsilon>0$, Assumption \ref{cond:m3}, the tail behavior of the Gaussian distribution, see, e.g., Lemma D.4 in \cite{durante2023skewed} and \eqref{thm1:help0:map:sym},  imply, for a sufficiently large choice of $c_0$ within $M_n = \sqrt{ c_0 \log n }$, that
	\begin{equation} \label{thm1:help1:map:sym}
		\begin{aligned}
			P_0^n ( n^2  P_{{\bf 0},\hat  \bOmega, f(\cdot)}(\|\hat \bh\| > 2M_n\sqrt{d}) > \epsilon  ) \, &= \, P_0^n ( \{ n^2  P_{{\bf 0},\hat  \bOmega, f(\cdot)}(\|\hat \bh\| > 2M_n\sqrt{d}) > \epsilon \} \cap \tilde  A_{n,0}  ) + o(1)\\
			&= o(1).
		\end{aligned}
	\end{equation} 
	By combining the above results we obtain
	\begin{equation} \label{bound:out:skw:map:sym} 
		\ \int| { q_{n,\hat{\btheta},(2)}(\hat \bh)}  -{ q^{\tilde K_n}_{n,\hat{\btheta},(2)}(\hat \bh)}| d\hat \bh		= O_{P_0^n}(n^{-2}).
	\end{equation}
To conclude the proof, it remains to deal with the second term in the right--hand--side of \eqref{triangle:tv:map:sym}. With~this goal in mind, let~us~define the event
	$$ \textstyle \tilde A_{n,1} =  \tilde A_{n,0} \cap \{ \int_{\tilde K_n}\pi_n(\hat \bh) d\hat \bh > 0 \} \cap \{ \int_{\tilde K_n} { q_{n,\hat{\btheta},(2)}(\hat \bh)}d\hat \bh > 0 \}.$$ 
	Notice that, as a consequence of Lemma~\ref{lemma:post:contr}, \smash{$P_0^n \{ \int_{\tilde K_n}\pi_n(\hat \bh) d\hat \bh > 0 \} =1-o(1)$}. Similarly, by the Assumption~\ref{cond:m3} and assertions~\eqref{thm1:help0:map:sym}--\eqref{thm1:help1:map:sym} we have \smash{$P_0^n  \{ \int_{\tilde K_n} { q_{n,\hat{\btheta},(2)}(\hat \bh)}d\hat \bh > 0 \} = 1-o(1)$},~and hence,  \smash{$P_0^n \tilde A_{n,1} = 1 - o(1)$}, respectively. Furthermore, recall that \smash{$ \pi_n(\hat \bh) = 2 \bar{\pi}_{n,\hat \btheta}(\hat \bh) w_{\bf 0}(\hat \bh)$}, and define with \smash{$ \bar{\pi}^{\tilde K_n}_{n,\hat \btheta}(\hat \bh) $} and \smash{$\bar{q}^{\tilde K_n}_{n,\hat \btheta,(2)}(\hat \bh)$} the symmetrized posterior~density and the symmetric component of ${ q_{n,\hat{\btheta},(2)}(\hat \bh)}$, respectively, restricted within the set $\tilde K_n$. By the symmetry~of \smash{$ \bar{\pi}^{\tilde K_n}_{n, \hat \btheta}(\hat \bh) $} and \smash{$\bar{q}^{\tilde K_n}_{n,\hat \btheta,(2)}(\hat \bh)$}, and the previously-discussed invariance properties, we have
		\begin{equation*}
		\begin{aligned}
			& \int| \pi_n^{\tilde K_n}( \hat \bh) - { q^{\tilde K_n}_{n,\hat{\btheta},(2)}(\hat \bh)}| d\hat \bh = \int |  \bar{\pi}^{\tilde K_n}_{n,\hat \btheta}(\hat \bh) -\bar{q}^{\tilde K_n}_{n,\hat \btheta,(2)}(\hat \bh) | d \hat \bh.
		\end{aligned}
	\end{equation*}
	From here onwards we can adapt the proof of Theorem 2.1 in \citet{durante2023skewed} to our setting. In~particular, let us consider
		\begin{equation}\label{eq_mid_asy}
		\begin{split}
				&\int |  \bar{\pi}^{\tilde K_n}_{n,\hat \btheta}(\hat \bh) - \bar{q}^{\tilde K_n}_{n,\hat \btheta,(2)}(\hat \bh) | d \hat \bh \mathbbm{1}_{\tilde A_{n,1}} \\
				& \qquad  \qquad =  {\int_{\tilde K_n}}  \Big| 1 - \int_{\tilde K_n} \frac{\bar{q}^{\tilde K_n}_{n,\hat \btheta,(2)}(\hat \bh)}{\bar{q}^{ \tilde K_n}_{n,\hat \btheta,(2)}(\hat \bh')} \frac{\bar{\pi}^{\tilde K_n}_{n,\hat \btheta}(\hat \bh')}{ \bar{\pi}^{\tilde K_n}_{n,\hat \btheta}(\hat \bh)}\bar{q}^{\tilde K_n}_{n,\hat \btheta,(2)}(\hat \bh')d\hat \bh'\Big| \bar{\pi}^{\tilde K_n}_{n,\hat \btheta}(\hat \bh) d\hat \bh  \mathbbm{1}_{\tilde A_{n,1}},
				\end{split}
	\end{equation}
where the ratios  \smash{$ \bar{q}^{\tilde K_n}_{n,\hat \btheta,(2)}(\hat \bh) /\bar{q}^{\tilde K_n}_{n,\hat \btheta,(2)}(\hat \bh') $} and \smash{$ \bar{\pi}^{\tilde K_n}_{n,\hat \btheta}(\hat \bh')/ \bar{\pi}^{\tilde K_n}_{ n,\hat \btheta}(\hat \bh)$} in the above expression are equal to their unrestricted versions \smash{$\bar{q}_{n,\hat \btheta,(2)}(\hat \bh) /\bar{q}_{n,\hat \btheta,(2)}(\hat \bh')$} and \smash{$ \bar{\pi}_{n,\hat \btheta}(\hat \bh')/ \bar{\pi}_{n,\hat \btheta}(\hat \bh)$}, for $\hat \bh, \hat \bh' \in \tilde{K}_n$. This result, together with Jensen's inequality, imply that \eqref{eq_mid_asy} is upper bounded by 
		\begin{equation*}
\int_{ \tilde K_n \times \tilde K_n} \Big| 1 - \frac{\bar{q}_{n,\hat \btheta,(2)}(\hat \bh)}{\bar{q}_{n,\hat \btheta,(2)}(\hat \bh')} \frac{\bar{\pi}_{ n,\hat \btheta}(\hat \bh')}{\bar{\pi}_{ n,\hat \btheta}(\hat \bh)} \Big| \bar{\pi}^{\tilde K_n}_{ n,\hat \btheta}(\hat \bh) \bar{q}^{ \tilde K_n}_{n,\hat \btheta,(2)}(\hat \bh')d \hat \bh d\hat \bh'  \mathbbm{1}_{\tilde A_{n,1}}. \\
	\end{equation*} 
Recall $\bar{\pi}_{n,\hat \btheta}(\hat \bh) \propto [\pi(\hat \btheta + \hat \bh /\sqrt{n})L(\hat \btheta + \hat \bh /\sqrt{n}; \by_{1:n})+ \pi(\hat \btheta - \hat \bh /\sqrt{n})L(\hat \btheta - \hat \bh /\sqrt{n};\by_{1:n})]/[2\pi(\hat \btheta )L(\hat \btheta; \by_{1:n})$]. As a consequence, leveraging Lemma \ref{lemma:symmetrized:modal}, together with  the definitions of \smash{$\hat \bOmega$} and \smash{$f(\hat \bh)$}, and the expansion $e^x = 1 + x + e^{\beta x}x^2/2$, for some $\beta \in (0,1)$, we obtain 
	\begin{equation} \label{upper:cond:Kn:sym}
		\begin{aligned}
			\int |  \bar{\pi}^{\tilde K_n}_{ n,\hat \btheta}(\hat \bh) -\bar{q}^{\tilde K_n}_{n,\hat \btheta,(2)}(\hat \bh) | d \hat \bh \mathbbm{1}_{\tilde A_{n,1}} &\leq \int_{\tilde K_n \times \tilde K_n} \Big| 1 - e^{\tilde r_{n,2}({\hat \bh'}) - \tilde r_{n,2}({\hat \bh})} \Big| \bar{\pi}^{\tilde K_n}_{n, \hat \btheta}(\hat \bh) \bar{q}^{\tilde K_n}_{n,\hat \btheta,(2)}(\hat \bh') d \hat \bh d\hat \bh'  \mathbbm{1}_{\tilde A_{n,1}}\\
			& \leq 2 \tilde r_{n,2} + 2 \exp(2\beta \tilde r_{n,2})\tilde r_{n,2}^2 = O_{P_0^n}( M_n^{\tilde{c}_2}d^6/n^2),		
		\end{aligned}
	\end{equation}  
	where \smash{$\tilde r_{n,2}  = \sup_{\hat \bh \in \tilde K_n} |\tilde r_{n,2}(\hat \bh)|=O_{P_0^n}( M_n^{\tilde{c}_2}d^6/n^2) $} (see Lemma \ref{lemma:symmetrized:modal}). Setting {$c_4=\tilde{c}_2$}, and~recalling the~invariance of the TV distance with respect to invertible affine transformations,  concludes the proof~for~TV.

\vspace{150pt}	
\noindent	{\bf Proof for KL and reverse--KL.}
	Note that by symmetry and equation~\eqref{skew:sym:equiv:sym}, we have
	\begin{align*}
		&\mathcal{D}_{\textsc{kl}}\big[\pi_n\|q_{n,\hat{\btheta},(2)}\big]+\mathcal{D}_{\textsc{kl}}\big[q_{n,\hat{\btheta},(2)}\|\pi_n \big]
		=\mathcal{D}_{\textsc{kl}}\big[\bar\pi_{n,\hat{\btheta}}\|\bar{q}_{n,\hat{\btheta},(2)}\big]+\mathcal{D}_{\textsc{kl}}\big[\bar{q}_{n,\hat{\btheta},(2)}\|\bar\pi_{n,\hat{\btheta}} \big]\nonumber\\
		&\qquad \qquad  = {\int}{\int}\log \Big(\frac{\bar\pi_{n,\hat{\btheta}}(\hat{\bh})}{\bar{q}_{n,\hat{\btheta},(2)}(\hat{\bh})} \Big)\bar\pi_{n,\hat{\btheta}}(\hat{\bh}) \bar{q}_{n,\hat{\btheta},(2)}(\hat{\bh}') d\hat{\bh} d\hat{\bh}'\\
		&\qquad \qquad  \qquad \qquad \qquad \qquad \qquad +  {\int}{\int}\log \Big(\frac{\bar{q}_{n,\hat{\btheta},(2)}(\hat{\bh}')}{\bar\pi_n(\hat{\bh}')} \Big) \bar{q}_{n,\hat{\btheta},(2)}(\hat{\bh}')  \bar\pi_{n,\hat{\btheta}}(\hat{\bh})d\hat{\bh}' d\hat{\bh}\nonumber\\
		&\qquad \qquad  = {\int}{\int}\log \Big(\frac{\bar\pi_{n,\hat{\btheta}}(\hat{\bh})}{\bar{q}_{n,\hat{\btheta},(2)}(\hat{\bh})} \frac{\bar{q}_{n,\hat{\btheta},(2)}(\hat{\bh}')}{\bar\pi_{n,\hat{\btheta}}(\hat{\bh}')}\Big)\bar\pi_{n,\hat{\btheta}}(\hat{\bh}) \bar{q}_{n,\hat{\btheta},(2)}(\hat{\bh}')d \hat{\bh} d\hat{\bh}'.
	\end{align*}
	We now split the domains of the above integrals into $\tilde K_n$ and $\tilde K_n^c$ and~study the parts separately.
	
	In particular,
	\vspace{-3pt}
	\begin{align*}
		\int\int_{\tilde{K}_n^c}& \log \Big(\frac{\bar\pi_{n,\hat{\btheta}} (\hat{\bh})}{\bar{q}_{n,\hat{\btheta},(2)}(\hat{\bh})} \frac{\bar{q}_{n,\hat{\btheta},(2)}(\hat{\bh}')}{\bar\pi_{n,\hat{\btheta}} (\hat{\bh}')}\Big)\bar\pi_{n,\hat{\btheta}} (\hat{\bh})   \bar{q}_{n,\hat{\btheta},(2)}(\hat{\bh}') d\hat{\bh}d\hat{\bh}'\nonumber\\
		&\leq  \int_{\tilde{K}_n^c} \bar\pi_{n,\hat{\btheta}} (\hat{\bh}) \big(|\log \bar\pi_{n,\hat{\btheta}} (\hat{\bh})|+ |\log \bar{q}_{n,\hat{\btheta},(2)}(\hat{\bh})|\big)d\hat{\bh}
		+  \bar\Pi_n(\tilde{K}_n^c)  \mathcal{D}_{\textsc{kl}}\big[q_{n,\hat{\btheta},(2)}\|\pi_n \big].
	\end{align*}
	Similarly, we have that
	\begin{align*}
		\int\int_{\tilde{K}_n^c}& \log \Big(\frac{\bar\pi_{n,\hat{\btheta}}(\hat{\bh})}{\bar{q}_{n,\hat{\btheta},(2)}(\hat{\bh})} \frac{\bar{q}_{n,\hat{\btheta},(2)}(\hat{\bh}')}{\bar\pi_{n,\hat{\btheta}}(\hat{\bh}')}\Big)\bar{q}_{n,\hat{\btheta},(2)}(\hat{\bh}') \bar\pi_{n,\hat{\btheta}}(\hat{\bh})d\hat{\bh}'d\hat{\bh}   \\
		&\leq  \int_{\tilde{K}_n^c} \bar{q}_{n,\hat{\btheta},(2)}(\hat{\bh}') \big(|\log \bar\pi_{n,\hat{\btheta}}(\hat{\bh}')|+ |\log \bar{q}_{n,\hat{\btheta},(2)}(\hat{\bh}')|\big)d\hat{\bh}'
		+  P_{ \bar q_{n,\hat{\btheta},(2)}}(\tilde{K}_n^c)  \mathcal{D}_{\textsc{kl}}\big[\pi_n\| q_{n,\hat{\btheta},(2)} \big].
	\end{align*}
	Recall that in view of Lemma \ref{lemma:post:contr}, Proposition 6 in  \citet{wang2004skew}, \eqref{thm1:help1:map:sym}, and the results in the proof for the TV distance we have that $\bar\Pi_n(\tilde{K}_n^c)= \Pi_n(\tilde{K}_n^c)=o_{P_0^n}(\smash{n^{-2}})$ and $\smash{P_{ \bar q_{n,\hat{\btheta},(2)}}}(\tilde{K}_n^c)=o_{P_0^n}(\smash{n^{-2}})$.~Moreover, in view of Assumption \ref{assump:LBtail} and noting that $\|\log {\bar q}_{n,\smash{\hat \btheta},(2)}(\smash{\hat{\bh}})\|\lesssim \|\smash{\hat{\bh}}\|^2$ for~$\|\smash{\hat{\bh}}\|{\geq} 1$, we have, as a result of (B.8)\footnote{ Note that in \citet{durante2023skewed} an $O_{P_0^n}(n^{-1})$ upper bound was given, but by sufficiently large choice of $c_0$ in the definition of $K_n$ and $\tilde K_n$ this upper bound can be set to arbitrary polynomial of $n$. Here~it~suffices~to~take~$o_{P_0^n}(n^{-2})$.} in \citet{durante2023skewed}, that
	\begin{align*}
		\int_{\tilde{K}_n^c} \bar\pi_{n,\hat{\btheta}}(\hat{\bh}) \big(|\log \bar\pi_{n,\hat{\btheta}}(\hat{\bh})|+ |\log \bar q_{n,\hat{\btheta},(2)}(\hat{\bh})|\big)d\hat{\bh}\leq 2 \int_{\tilde{K}_n^c}  \|\hat{\bh}\|^{M\vee 2} \pi_n(\hat{\bh})  d\hat{\bh}=o_{P_0^n}(n^{-2}).
	\end{align*}
	Furthermore, following from the display below (B.8) in  \citet{durante2023skewed} and \eqref{thm1:help0:map:sym}, it holds
	\begin{align*}
		{\int_{\tilde{K}_n^c}}\bar q_{n,\hat{\btheta},(2)}(\hat{\bh}') \big(|\log \bar\pi_{n,\hat{\btheta}}(\hat{\bh}')|+ |\log \bar q_{n,\hat{\btheta},(2)}(\hat{\bh}')|\big)d\hat{\bh}'\leq  2{\int_{\tilde{K}_n^c}}  \|\hat{\bh}'\|^{M\vee 2}  q_{n,\hat{\btheta},(2)}(\hat{\bh}')d\hat{\bh}'=o_{P_0^n}(n^{-2}).
	\end{align*}
	Finally, we deal with the integral on $\tilde K_n \times \tilde K_n$. Note, that in view of Lemma \ref{lemma:symmetrized:modal}, we have
	\begin{align*}
		\int_{\tilde{K}_n}\int_{\tilde{K}_n} &\log \Big(\frac{\bar\pi_{n,\hat{\btheta}}(\hat{\bh})}{\bar q_{n,\hat{\btheta},(2)}(\tilde{\bh})} \frac{\bar q_{n,\hat{\btheta},(2)}(\hat{\bh}')}{\bar\pi_{n,\hat{\btheta}}(\hat{\bh}')}\Big)\bar\pi_{n,\hat{\btheta}}(\hat{\bh}) \bar q_{n,\hat{\btheta},(2)}(\hat{\bh}')d \hat{\bh} d\hat{\bh}'\\
		&\leq \int_{\tilde{K}_n}\int_{\tilde{K}_n}\big(|\tilde{r}_{n,2}(\hat{\bh})|+|\tilde{r}_{n,2}(\hat{\bh}')| \big)\bar\pi_{n,\hat{\btheta}}(\hat{\bh})\bar q_{n,\hat{\btheta},(2)}(\hat{\bh}')d \hat{\bh} d\hat{\bh}' \\
		&\qquad \qquad \quad  \leq  \int_{\tilde{K}_n}|\tilde{r}_{n,2}(\hat{\bh})| \bar\pi_{n,\hat{\btheta}}(\hat{\bh})d\hat{\bh}+ \int_{\tilde{K}_n}|\tilde{r}_{n,2}(\hat{\bh}')|   \bar q_{n,\hat{\btheta},(2)}(\hat{\bh}')d \hat{\bh}'= O_{P_0^n}(M_n^{\tilde{c}_2}d^6/n^2).
	\end{align*}	
	By combining the preceding displays, we get that
	\begin{align*}
		&\mathcal{D}_{\textsc{kl}}\big[\pi_n\|q_{n,\hat{\btheta},(2)}]+\mathcal{D}_{\textsc{kl}}\big[q_{n,\hat{\btheta},(2)}\|\pi_n \big]\\
		& \qquad  \qquad \leq (\mathcal{D}_{\textsc{kl}}\big[\pi_n\|q_{n,\hat{\btheta},(2)}\big]+\mathcal{D}_{\textsc{kl}}\big[q_{n,\hat{\btheta},(2)}\|\pi_n \big]) o_{P_{0}^n}(n^{-2})+ o_{P_0^n}(n^{-2})+ O_{P_0^n}(M_n^{\tilde{c}_2}d^6/n^2),
	\end{align*}
	where the leading term is given by $O_{P_0^n}(M_n^{\tilde{c}_2}d^6/n^2)$. 
	
	Since both terms \smash{$\mathcal{D}_{\textsc{kl}}\big[\pi_n\|q_{n,\hat{\btheta},(2)}\big]$} and \smash{$\mathcal{D}_{\textsc{kl}}\big[q_{n,\hat{\btheta},(2)}\|\pi_n \big]$} are upper bounded  by $\smash{\mathcal{D}_{\textsc{kl}}\big[\pi_n\|q_{n,\hat{\btheta},(2)}\big]}+\smash{\mathcal{D}_{\textsc{kl}}\big[q_{n,\hat{\btheta},(2)}\|\pi_n \big]}$, the above result concludes the proof for  the \textsc{kl} and reverse--\textsc{kl}, after noticing that also these two divergences are invariant with respect to invertible affine transformations.

	\vspace{20pt}	
\noindent	{\bf Proof for $\alpha$--divergences with $\alpha\in(0,1)$.}
First notice that by H\"older's inequality, assertion \eqref{thm1:help1:map:sym} and the argument above \eqref{bound:out:pi:map:sym} we get that
	\begin{align}
		\int_{\tilde{K}_n^c}\bar\pi_{n,\hat{\btheta}}^\alpha(\hat{\bh}) ( \bar q_{n,\hat{\btheta},(2)}(\hat{\bh}))^{1-\alpha}d\hat{\bh}\leq \Big(\int_{\tilde{K}_n^c} \bar\pi_{n,\hat{\btheta}}(\hat{\bh})d\hat{\bh}\Big)^\alpha \Big(\int_{\tilde{K}_n^c} \bar q_{n,\hat{\btheta},(2)}(\hat{\bh})d\hat{\bh}\Big)^{1-\alpha}=O_{P_0^n}(n^{-2}).\label{eq:UB:complement}
	\end{align}
Next, let us introduce the notations
	\begin{align*}
		\bar\pi_{n,\hat \theta}^{\mbox{\tiny un}}(\hat  \bh)&= \frac{\pi(\hat \btheta+\hat {\bh}/\sqrt{n})L(\hat\btheta+\hat{\bh}/\sqrt{n};\by_{1:n})+\pi(\hat \btheta-\hat {\bh}/\sqrt{n})L(\hat \btheta-\hat {\bh}/\sqrt{n};\by_{1:n})}{2\pi(\hat \btheta)L(\hat \btheta;\by_{1:n})},\\
		\bar q^{\mbox{\tiny un}}_{n,\hat {\btheta},(2)}(\hat  \bh)&= \exp\{-\hat{\bh}^\intercal \hat\bOmega^{-1}\hat {\bh}/2\}f(\hat {\bh}),
	\end{align*}
	and recall that 
	\begin{align*}
		\bar\pi_{n,\hat {\btheta}}(\hat  \bh)=\frac{\bar\pi_{n,\hat {\btheta}}^{\mbox{\tiny un}}(\hat  \bh)}{\int \bar\pi_{n,\hat {\btheta}}^{\mbox{\tiny un}}(\hat  \bh) d\hat  \bh},\qquad \qquad 
		\bar q_{n,\hat {\btheta},(2)}(\hat  \bh)=\frac{ \bar q^{\mbox{\tiny un}}_{n,\hat {\btheta},(2)}(\hat  \bh)}{\int  \bar q^{\mbox{\tiny un}}_{n,\hat {\btheta},(2)}(\hat  \bh) d\hat {\bh}}.
	\end{align*}
	Furthermore, following again from  \eqref{thm1:help1:map:sym} and the argument above \eqref{bound:out:pi:map:sym}, we also have that
	\begin{align}
		\int_{\tilde{K}_n} \bar\pi^{\mbox{\tiny un}}_{n,\hat {\btheta}}(\hat \bh) d\hat {\bh}&=\big(1+O_{P_0^n}(n^{-2})\big)\int  \bar\pi^{\mbox{\tiny un}}_{n,\hat {\btheta}}(\hat  \bh)d\hat {\bh},\nonumber\\
		\int_{\tilde{K}_n}  \bar q^{\mbox{\tiny un}}_{n,\hat {\btheta},(2)}(\hat  \bh) d\hat {\bh}&=\big(1+O_{P_0^n}(n^{-2})\big)\int \bar q^{\mbox{\tiny un}}_{n,\hat {\btheta},(2)}(\hat  \bh) d\hat {\bh}.\label{eq:supp_Kn}
	\end{align}
Next we note that
\begin{equation}
	\begin{split}
		&\int_{\tilde{K}_n}\bar\pi_{n,\hat {\btheta}}^\alpha(\hat {\bh}) (\bar q_{n,\hat {\btheta},(2)}(\hat {\bh}))^{1-\alpha}d\hat {\bh}\\
		&=\int_{\tilde{K}_n} e^{\displaystyle{\alpha\big(\log \bar\pi_{n,\hat {\btheta}}^{\mbox{\tiny un}}(\hat  \bh))- \log {\bar q}^{ \mbox{\tiny un}}_{n,\hat {\btheta},(2)}(\hat  \bh)\big)+\alpha\big( \log {\textstyle \int} {\bar q}^{ \mbox{\tiny un}}_{n,\hat {\btheta},(2)}(\hat  \bh) d\hat {\bh}-\log {\textstyle \int}\bar\pi_{n,\hat {\btheta}}^{\mbox{\tiny un}}(\hat  \bh) d\hat  \bh \big)}}{\bar q}_{n,\hat {\btheta},(2)}(\hat  \bh) d\hat {\bh}.
	\end{split}
	\label{eq:UB:Renyi}
	\end{equation}
	Then in view of Lemma \ref{lemma:symmetrized:modal}, it follows
	\begin{align*}
		\sup\nolimits_{\hat {\bh}\in\tilde{K}_n}\big(\log \bar\pi_{n,\hat {\btheta}}^{\mbox{\tiny un}}(\hat \bh)- \log  \bar q^{\mbox{\tiny un}}_{n,\hat {\btheta},(2)}(\hat  \bh)\big)\leq \sup\nolimits_{\hat {\bh}\in\tilde{K}_n}\tilde{r}_{n,2}(\hat {\bh})=O_{P_0^n}(M_n^{\tilde{c}_2}d^6/n^2).
	\end{align*}
	Furthermore, in view of 
	$$\frac{\int_A f(x)dx}{\int_A g(x)dx}=\frac{\int_A \frac{f(x)}{g(x)}g(x)dx}{\int_A g(x)dx}\leq \sup_{x\in A}\frac{f(x)}{g(x)},$$
	we also have, following from \eqref{eq:supp_Kn} and Lemma \ref{lemma:symmetrized:modal}, that
	\begin{align*}
		&\log \int  \bar q^{\mbox{\tiny un}}_{n,\hat {\btheta},(2)}(\hat  \bh) d\hat {\bh}-\log \int \bar\pi_{n,\hat {\btheta}}^{\mbox{\tiny un}}(\hat  \bh) d\hat  \bh
		=\log \frac{ \int_{\tilde{K}_n}  \bar q^{\mbox{\tiny un}}_{n,\hat {\btheta},(2)}(\hat \bh) d\hat {\bh}}{\int_{\tilde{K}_n}   \bar\pi_{n,\hat {\btheta}}^{\mbox{\tiny un}}(\hat  \bh) d\hat  \bh}
		+ \log\big(1+O_{P_0^n}(n^{-2})\big)\\
		&\qquad \qquad \leq \sup_{\hat  \bh\in \tilde{K}_n}\log \frac{ \bar q^{\mbox{\tiny un}}_{n,\hat {\btheta},(2)}(\hat  \bh) }{ \bar\pi_{n,\hat {\btheta}}^{\mbox{\tiny un}}(\hat  \bh)}+O_{P_0^n}(n^{-2})=\sup_{\hat {\bh}\in\tilde{K}_n}\tilde{r}_{n,2}(\hat {\bh})+O_{P_0^n}(n^{-2})=O_{P_0^n}(M_n^{\tilde{c}_2}d^6/n^2).
	\end{align*}
	Finally, by substituting into \eqref{eq:UB:Renyi} the above upper bounds and using \eqref{eq:UB:complement} together with the invariance of $\alpha$--divergences with respect to invertible affine transformations, we have, for some universal constant $C$, that
	\begin{align*}
		&\mathcal{D}_{\alpha}\big[\pi_n(\btheta)\|q_{n,\hat {\btheta},(2)}(\btheta)\big]=\mathcal{D}_{\alpha}\big[\bar\pi_{n,\hat {\btheta}}(\btheta)\|\bar q_{n,\hat {\btheta},(2)}(\btheta)\big]\\
		&= \frac{1}{\alpha(1-\alpha)}\\
		&\quad\times\Big(1-{\displaystyle{\int_{\tilde{K}_n}}} \Big[e^{\displaystyle{\alpha\big(\log \bar\pi_{n,\hat {\btheta}}^{\mbox{\tiny un}}(\hat \bh)){-} \log {\bar q}^{\mbox{\tiny un}}_{n,\hat {\btheta},(2)}(\hat  \bh)\big)}}\\
		&\qquad \qquad \qquad \qquad \times e^{\displaystyle{\alpha\big( \log {\textstyle \int} {\bar q}^{\mbox{\tiny un}}_{n,\hat {\btheta},(2)}(\hat \bh) d\hat {\bh}{-}\log {\textstyle \int}\bar\pi_{n,\hat {\btheta}}^{\mbox{\tiny un}}(\hat  \bh) d\hat  \bh \big)}}{\bar q}_{n,\hat {\btheta},(2)}(\hat  \bh) \Big]d\hat {\bh}+O_{P_0^n}(n^{-2})\Big)\\
		&\leq  \frac{1}{\alpha(1-\alpha)}\Big(1- \int_{\tilde{K}_n} e^{-\alpha C M_n^{\tilde{c}_2}d^6/n^2}{\bar q}_{n,\hat {\btheta},(2)}(\hat  \bh) d\hat {\bh}\Big)+O_{P_0^n}(n^{-2})\\
		&=\frac{1}{\alpha(1-\alpha)}\Big(1- e^{-\alpha C M_n^{\tilde{c}_2}d^6/n^2}\big(1+O_{P_0^n}(n^{-2})\big)\Big)+O_{P_0^n}(n^{-2})=O_{P_0^n}(M_n^{\tilde{c}_2}d^6/n^2),
	\end{align*}
	concluding the proof of Theorem \ref{teo_3}.
\end{proof}

\newtheorem{lemma}{Lemma}[section]

\vspace{10pt}

\section{\large B. Technical lemmas}\label{AB}
\vspace{10pt}
\begin{lemma}[Posterior contraction] \label{lemma:post:contr}
	Let $K_n = \{\btheta \in \Theta \, : \, \|\btheta - \btheta_0\| < M_n \sqrt{d/n}\}$. Then, under Assumptions \ref{cond:4}, \ref{cond:m1} and \ref{cond:m3}, we have, for $c_0>0$ sufficiently large (not depending on $n$ and $d$) and  $M_n = \sqrt{c_0 \log n}$,~that 
	\begin{equation*}
		\lim_{n \to \infty } P_0^n \{ \Pi_n( K_n^c ) < 2n^{-dc_0c_1/2} \} = 1,  
	\end{equation*} 
	where  $K_n^c$ denotes the complement of $K_n$.
\end{lemma}
\begin{proof}.
Lemma~\ref{lemma:post:contr} is a direct consequence of Lemma D.5 in \citet{durante2023skewed} when the focus is on asymptotic regimes. As such, its proof follows from related reasoning and derivations~as those~considered to prove Lemma D.5, after noticing that   Assumptions \ref{cond:m1}--\ref{cond:m3} imply Assumption 5 in \citet{durante2023skewed}.
\end{proof}

\vspace{5pt}

\begin{lemma} \label{lemma:symmetrized:modal}
	Let $\hat \bh = \sqrt{n}(\btheta -\hat  \btheta) $, where $\hat \btheta$ denotes the \textsc{map}, and $\tilde K_n = \{ \hat  \bh \, :\, \|\hat  \bh\| < 2 M_n\sqrt{d} \}$ with $M_n = \sqrt{c_0 \log n }$ for some $c_0 >0$. Moreover, let $\tilde{\ell}_n(\btheta) = \log \pi(\btheta) L(\btheta; \by_{1:n})$ be the un--normalized log-posterior~and define
\begin{eqnarray}
&& \tilde r_{n,1}( \hat  \bh):=\label{taylor:sym:1} \\
		&&\ \ \log ( \big [ \exp \lbrace\tilde{\ell}_n( \hat {\btheta} + \hat \bh /\sqrt{n} ) -\tilde{\ell}_n( \hat{\btheta}) \rbrace + \exp \lbrace \tilde{\ell}_n( \hat{\btheta} - \hat \bh /\sqrt{n} ) -\tilde{\ell}_n(\hat{\btheta})  \rbrace \big]/2 )+ \hat \bh^\top \hat\bOmega^{-1}\hat \bh/2, \qquad  \qquad \nonumber 
\end{eqnarray}
\vspace{-20pt}
\begin{eqnarray}
&& \tilde r_{n,2}( \hat \bh):= \label{taylor:sym:2}	 \\
		&& \log ( \big [ \exp \lbrace\tilde{\ell}_n(\hat{\btheta} + \hat \bh /\sqrt{n} )- \tilde{\ell}_n( \hat{\btheta}) \rbrace + \exp \lbrace \tilde{\ell}_n( \hat{\btheta} - \hat \bh /\sqrt{n} )-\tilde{\ell}_n( \hat{\btheta})  \rbrace \big]/2) +\hat \bh^\top \hat\bOmega^{-1} \hat\bh/2 - \log f(\hat \bh),  \nonumber 
\end{eqnarray}
	with $\hat \bOmega^{-1} = \bJ_{\hat \btheta}/n$, and $ f(\hat \bh)$ as  in \eqref{def:Ptilde}.
Then, under Assumption \ref{cond:m3}, there are constants  $\tilde{c}_1,\tilde{c}_2>0$ such~that
	\begin{equation}\label{order:taylor:sym:1}
		\tilde r_{n,1} := \sup\nolimits_{\hat \bh \in \tilde  K_n} |\tilde r_{n,1}(\hat \bh)| = O_{P_0^n}( M_n^{\tilde{c}_1}d^3/n),
	\end{equation}  
	and 
	\begin{equation} \label{order:taylor:sym:2}
		\tilde r_{n,2} := \sup\nolimits_{\hat \bh \in \tilde  K_n} |\tilde r_{n,2}(\hat \bh) |= O_{P_0^n}( M_n^{\tilde{c}_2 } d^6/n^2).
	\end{equation}
\end{lemma}

\begin{proof}.
Let us focus on proving \eqref{order:taylor:sym:2}. The proof of \eqref{order:taylor:sym:1}  follows in a similar manner, leveraging simpler derivations. {Consistent with this goal, recall that under Assumption \ref{cond:m1} the event $\lbrace  \| \smash{\hat{\btheta}} - \btheta_0 \| \leq \delta \rbrace$ has, for any $\delta >0$, a probability tending to 1}. As a consequence, from Assumption~\ref{cond:m3}~the log--likelihood~ratio~around $\smash{\hat \btheta}$ can be expanded as
	\begin{equation*}
		\log \frac{p_{\hat{\btheta} + \hat \bh/\sqrt{n} }}{p_{\hat{\btheta}}}(\bY_{1:n}) \, = \, \sum\nolimits_{k = 1}^{5} \frac{ \langle \ell^{(k)}_{ n,\hat \btheta},\hat \bh^{\otimes k} \rangle }{ k! n^{k/2} } + O_{P_0^n}\Big(\frac{\{ \|\hat \bh\|^6 \vee 1\}}{{n^2}}\Big).
	\end{equation*}
	Similarly, the Taylor expansion of the log--prior takes the form 
	\begin{equation*}
		\log \frac{\pi(\hat{\btheta} +\hat  \bh/\sqrt{n} )}{ \pi(\hat{\btheta}) } \, = \, \sum\nolimits_{k = 1}^{3} \frac{ \langle \log \pi_{\hat \btheta}^{(k)}, \hat \bh^{\otimes k} \rangle }{ k! n^{k/2} } + O_{P_0^n}\Big(\frac{\{ \|\hat \bh\|^4 \vee 1\}}{n^2}\Big).
	\end{equation*}
	Combining the above displays, the log--posterior ratio can expressed as 
	\begin{equation} \label{log:post:ratio}
		\tilde{\ell}_n(\hat \btheta+ \hat \bh/\sqrt{n}) -\tilde{\ell}_n( \hat \btheta)  \, = \, \sum\nolimits_{k = 1}^{5} \frac{ \langle a^{(k)}_{ \hat \btheta}, \hat\bh^{\otimes k} \rangle }{ k! n^{(k-2)/2} } + O_{P_0^n}\Big(\frac{\{ \|\hat \bh\|^6 \vee 1\}}{n^2}\Big),
	\end{equation}
	with \smash{$a_{\hat{\btheta}}^{(k)} = (\ell_{n,\hat{\btheta}}^{(k)} + \log \pi_{\hat \btheta}^{(k)} )/n, $} if $k \leq 3,$ and \smash{$a_{\hat{\btheta}}^{(k)} = \ell_{n,\hat{\btheta}}^{(k)}/n$} otherwise. Moreover,~since~$\hat \btheta$~is~the~\textsc{map}, we have that $a_{\hat{\btheta}}^{(1)} = 0.$\\
Therefore, in view of \eqref{log:post:ratio} and $e^x = 1 + O(x)$, for $x \to 0$, observe that
		\begin{equation} \label{sym:kernel2}
			\begin{aligned}
				&\exp \lbrace\tilde{\ell}_n( \hat{\btheta} + \hat \bh /\sqrt{n} )- \tilde{\ell}_n( \hat{\btheta}) \rbrace + \exp \lbrace \tilde{\ell}_n( \hat{\btheta} - \hat \bh /\sqrt{n} )-\tilde{\ell}_n( \hat{\btheta}) \rbrace\\
				& =\Big[\exp\Big( \sum_{k = 2}^{5} \frac{ \langle a^{(k)}_{ \hat \btheta}, \hat \bh^{\otimes k} \rangle }{ k! n^{(k-2)/2} } \Big)+ \exp\Big(\sum_{k = 2}^{5} \frac{ \langle a^{(k)}_{\hat \btheta}, - \hat\bh^{\otimes k} \rangle }{ k! n^{(k-2)/2} } \Big)\Big] \times{\Big[1 + O_{P_0^n}\Big(\frac{\{ \|\hat \bh\|^6 \vee 1\}}{n^2}\Big)\Big]}, \\
				& = \exp\Big( \sum_{k = 1}^{2} \frac{ \langle a^{(2k)}_{ \hat\btheta}, \hat\bh^{\otimes 2k} \rangle }{ (2k)! n^{k-1} } \Big) \Big[ \exp\Big( \sum_{k = 1}^{2} \frac{ \langle a^{(2k+1)}_{ \hat \btheta}, \hat \bh^{\otimes (2k+1)} \rangle }{ (2k+1)!n^{(2k-1)/2} } \Big) + \exp\Big( \sum_{k = 1}^{2} \frac{ \langle a^{(2k+1)}_{ \hat \btheta}, -\hat \bh^{\otimes (2k+1)} \rangle }{ (2k+1)!n^{(2k-1)/2} } \Big) \Big] \\
				&\qquad\qquad \qquad\qquad\qquad \qquad\times\Big[1 + O_{P_0^n}\Big(\{ \|\hat \bh\|^6 \vee 1\}/n^2\Big)\Big].
			\end{aligned}
	\end{equation} 
In view of Assumption \ref{cond:m3}, we have  \smash{$a^{(k)}_{ \hat \btheta} = O_{P_0^n}(1)$}, for $k \geq 2$. Therefore, by $e^x = 1 + x+ x^2/2 + x^3/6 + O(x^4)$, the two summands in the right--hand--side of \eqref{sym:kernel2} can be written as	
\begin{equation*}
		\begin{aligned}
			\exp\Big( \sum_{k = 1}^{2} \frac{ \langle a^{(2k+1)}_{ \hat \btheta}, \hat \bh^{\otimes (2k+1)} \rangle }{ (2k+1)!n^{(2k-1)/2} } \Big) &=  1 + \sum_{k = 1}^{2} \frac{ \langle a^{(2k+1)}_{\hat \btheta}, \hat \bh^{\otimes (2k+1)} \rangle }{ (2k+1)!n^{(2k-1)/2} } + \frac{1}{2}\left( \sum_{k = 1}^{2} \frac{ \langle a^{(2k+1)}_{ \hat \btheta}, \hat \bh^{\otimes (2k+1)} \rangle }{ (2k+1)!n^{(2k-1)/2} } \right)^2 \\
			&\qquad + \frac{1}{6}\left( \sum_{k = 1}^{2} \frac{ \langle  a^{(2k+1)}_{ \hat \btheta}, \hat\bh^{\otimes (2k+1)} \rangle }{ (2k+1)!n^{(2k-1)/2} } \right)^3  +  O_{P_0^n}\Big(\frac{\{ \|\hat \bh\|^{12} \vee 1\}}{n^2}\Big) ,
		\end{aligned}
	\end{equation*}
	and 
	\begin{equation*}
		\begin{aligned}
			\exp\Big( \sum_{k = 1}^{2} \frac{ \langle a^{(2k+1)}_{ \hat \btheta}, -\hat \bh^{\otimes (2k+1)} \rangle }{ (2k+1)!n^{(2k-1)/2} } \Big) &=  1 - \sum_{k = 1}^{2} \frac{ \langle a^{(2k+1)}_{ \hat \btheta}, \hat \bh^{\otimes (2k+1)} \rangle }{ (2k+1)!n^{(2k-1)/2} } + \frac{1}{2}\left( \sum_{k = 1}^{2} \frac{ \langle a^{(2k+1)}_{\hat\btheta}, \hat \bh^{\otimes (2k+1)} \rangle }{ (2k+1)!n^{(2k-1)/2} } \right)^2 \\
			&\qquad - \frac{1}{6}\left( \sum_{k = 1}^{2} \frac{ \langle a^{(2k+1)}_{ \hat \btheta}, \hat \bh^{\otimes (2k+1)} \rangle }{ (2k+1)!n^{(2k-1)/2} } \right)^3  +  O_{P_0^n}\Big(\frac{\{ \|\hat \bh\|^{12} \vee 1\}}{n^2}\Big).
		\end{aligned}
	\end{equation*}
	
	\noindent As a consequence, and recalling the asymptotic behaviour of $ a_{\hat \btheta}^{(2)}, \ldots, a_{\hat \btheta}^{(5)}$, for $\hat{\bh}\in\tilde{K}_n$, we have
	\begin{equation}
		\begin{aligned}
			&\frac{1}{2}\Big[ \exp\Big( \sum_{k = 1}^{2} \frac{ \langle a^{(2k+1)}_{ \hat \btheta}, \hat \bh^{\otimes (2k+1)} \rangle }{ (2k+1)!n^{(2k - 1)/2} } \Big) + \exp\Big( \sum_{k = 1}^{2} \frac{ \langle a^{(2k+1)}_{\hat \btheta}, -\hat \bh^{\otimes (2k+1)} \rangle }{ (2k+1)!n^{(2k-1)/2} } \Big) \Big] \\
			&\qquad = 1 + \frac{1}{2} \left( \sum_{k = 1}^{2} \frac{ \langle a^{(2k+1)}_{\hat \btheta}, \hat \bh^{\otimes (2k+1)} \rangle }{ (2k+1)!n^{(2k-1)/2} } \right)^2 +  O_{P_0^n}\Big(\frac{\{ \|\hat \bh\|^{12} \vee 1\}}{n^2}\Big) \\
			&\qquad=1 + \frac{1}{2} {\left( \frac{1}{6 \sqrt{n}}\langle a^{(3)}_{\hat \btheta}, \hat \bh^{\otimes 3} \rangle \right)^2} +  O_{P_0^n}\Big(\frac{\{ \|\hat \bh\|^{12} \vee 1\}}{n^2}\Big)\\
			&\qquad= 1 + \frac{1}{2} {\left( \frac{1}{6 \sqrt{n}}\big \langle \frac{\ell^{(3)}_{ n,\hat \btheta}}{n}, \hat \bh^{\otimes 3} \big \rangle \right)^2} +  O_{P_0^n}\Big(\frac{\{ \|\hat \bh\|^{12} \vee 1\}}{n^2}\Big),
		\end{aligned}
	\end{equation}	
	\noindent with the last line following by the fact that, from the Assumption \ref{cond:m3}, $ a^{(3)}_{ \hat \btheta}/(6\sqrt{n}) = \ell^{(3)}_{n, \hat \btheta}/(6 n^{3/2}) +\smash{O_{P_0^n}(n^{-3/2})} $.
	
	Combining the above displays with $e^x = 1 + x + x^2/2 + O(x^3)$ and $\sum_{i=0}^{2k}x^i/i!>0$ for any $k\in\mathbb{N}$, yields
	\vspace{5pt}
		\begin{equation*}
			\begin{aligned}
				&\frac{1}{2}\exp \lbrace\tilde{\ell}_n( \hat{\btheta} + \hat \bh /\sqrt{n} )- \tilde{\ell}_n( \hat{\btheta}) \rbrace + \frac{1}{2}\exp \lbrace \tilde{\ell}_n( \hat{\btheta} -\hat \bh /\sqrt{n} )-\tilde{\ell}_n( \hat{\btheta}) \rbrace \\
				&\qquad =\exp\Big( \frac{ \langle a^{(2)}_{\hat \btheta}, \hat \bh^{\otimes 2} \rangle }{ 2 } \Big)\Big( 1 + \frac{ \langle a^{(4)}_{\hat \btheta}, \hat \bh^{\otimes 4} \rangle }{ 24 n } + \frac{1}{2}\Big( \frac{ \langle a^{(4)}_{ \hat \btheta}, \hat \bh^{\otimes 4} \rangle }{ 24 n }   \Big)^2 + O_{P_0^n}\Big(\frac{\{ \|\hat \bh\|^{12} \vee 1\}}{n^3}\Big)  \Big)\\
				&\qquad \qquad \qquad\times \Big( 1 + \frac{1}{2}\Big( \frac{1}{6 \sqrt{n}}\big \langle \frac{\ell^{(3)}_{ n,\hat \btheta}}{n}, \hat \bh^{\otimes 3} \big \rangle \Big)^2 +O_{P_0^n}\Big(\frac{\{ \|\hat \bh\|^{12} \vee 1\}}{n^2}\Big) \Big) \Big[1 + O_{P_0^n}\Big(\frac{\{ \|\hat \bh\|^6 \vee 1\}}{n^2}\Big)\Big] \\
				&\qquad = \exp\Big( \frac{ \langle a^{(2)}_{ \hat \btheta}, \hat \bh^{\otimes 2} \rangle }{ 2 } \Big)\Big( 1 + \frac{1}{24 n} \Big \langle  \frac{\ell^{(4)}_{ n,\hat \btheta}}{n}, \hat \bh^{\otimes 4}  \Big \rangle + \frac{1}{2}\Big(  \frac{1}{24 n} \Big \langle  \frac{\ell^{(4)}_{ n,\hat \btheta}}{n}, \hat \bh^{\otimes 4}  \Big \rangle   \Big)^2 + \frac{1}{2}\Big( \frac{1}{6 \sqrt{n}}\Big \langle \frac{\ell^{(3)}_{ n,\hat \btheta}}{n},\hat \bh^{\otimes 3} \Big \rangle \Big)^2\Big) \\
				& \qquad  \qquad  \qquad \times \Big[1 + O_{P_0^n}\Big(\frac{\{ \|\hat \bh\|^{12} \vee 1\}}{n^2}\Big)\Big]\\
				&\qquad = \exp\Big(-\frac{1}{2} \hat \bh^\top\hat \bOmega^{-1} \hat \bh \Big) f(\hat \bh)  \Big[1 + O_{P_0^n}\Big(\frac{\{ \|\hat \bh\|^{12} \vee 1\}}{n^2}\Big)\Big],
			\end{aligned} 
	\end{equation*}
		\vspace{5pt}

\noindent	where the third line follows from \smash{$a^{(4)}_{ \hat \btheta} = \ell^{(4)}_{n,\hat \btheta}/n$} and the last from the fact that  \smash{$ \hat \bOmega^{-1} = - a^{(2)}_{\hat \btheta}$}, and the definition~of~ \smash{$f(\hat \bh)$}. 
	
	Therefore, in view of  $\log(1 + x) =  O(x)$ for $x \to 0 $, $\tilde r_{n,2}( \hat \bh)$ defined in \eqref{taylor:sym:2} satisfies $\tilde r_{n,2}( \hat \bh)  = O_{P_0^n}( \{ \|\hat \bh\|^{12} \vee 1\}/n^2),$
which in turn implies $\tilde r_{n,2} := \sup_{\hat \bh \in \tilde  K_n} |\tilde r_{n,2}(\hat \bh)| = O_{P_0^n}( M_n^{12}d^6/n^2)$. Setting $\tilde{c}_2= 12$, concludes the proof.
\end{proof}

%###############################################################################
%###############################################################################

\end{document}